\documentclass[a4paper,UKenglish,cleveref, autoref, thm-restate]{lipics-v2021}

\usepackage{amsmath}

\usepackage{amsthm}
\usepackage{amssymb}
\usepackage{amsfonts}
\usepackage{mathrsfs}
\usepackage{wrapfig} 

\usepackage{thmtools} 
\usepackage{thm-restate}
\usepackage{cite}

\usepackage{caption} 
\usepackage{subcaption}
\usepackage{changepage} 
\usepackage{mathtools} 

\usepackage{graphicx}

\usepackage{tikz}

\usetikzlibrary{arrows,decorations.pathmorphing,decorations.pathreplacing,backgrounds,positioning,fit,matrix}
\usetikzlibrary{shapes,calc,patterns,arrows.meta}
\tikzset{
	vert/.style={circle,inner sep=1.5,fill=white,draw,minimum size=.3cm},
	edge/.style={color=black, thick},
	diredge/.style={->,>={Stealth[width=8pt,length=8pt]},color=black, thick},
	timelabel/.style={fill=white,font=\footnotesize, text centered},
	wave/.style={decorate,decoration={coil,aspect=0}},
	dirwave/.style={->, >={Stealth[width=8pt,length=8pt]},decorate,decoration={coil,aspect=0}},
	diredge2/.style={->,>={Stealth[width=8pt,length=8pt]}}
}

\usepackage{enumerate}
\usepackage{todonotes}
\usepackage{comment}

\usepackage{algorithm}
\usepackage[noend]{algpseudocode}

\usepackage[T1]{fontenc}
\usepackage[utf8]{inputenc}
\usepackage{lmodern}

\usepackage{hyperref}

\crefname{claim}{Claim}{Claims}

\newcommand{\ie}{i.\,e.,\ }
\renewcommand{\st}{s.\,t.,\ }
\newcommand{\NP}{\textrm{NP}}

\usepackage{tabularx}
\newcommand{\problemdef}[3]{
	\begin{center}
		\begin{minipage}{0.95\textwidth}
			\noindent
			#1
			\vspace{5pt}\\
			\setlength{\tabcolsep}{3pt}
			\begin{tabularx}{\textwidth}{@{}lX@{}}
				\textbf{Input:}& #2 \\
				\textbf{Question:}& #3
			\end{tabularx}
		\end{minipage}
	\end{center}
}

\newcounter{guesscounter}

\newcommand{\deltaExactLong}{\textsc{Simple periodic Temporal Graph Realization}}
\newcommand{\deltaExact}{\textsc{Simple TGR}}

\usepackage{cite} 

\title{Temporal graph realization from fastest paths} 

\author{Nina Klobas}{Department of Computer Science, Durham University, UK}{nina.klobas@durham.ac.uk}{ https://orcid.org/0000-0002-8024-5782}{}

\author{George B. Mertzios}{Department of Computer Science, Durham University, UK}{george.mertzios@durham.ac.uk}{https://orcid.org/0000-0001-7182-585X}{Supported by the EPSRC grant EP/P020372/1.}

\author{Hendrik~Molter}{Department of Computer Science, Ben-Gurion~University~of~the~Negev, 
	Beer-Sheva, 
	Israel}{molterh@post.bgu.ac.il}{https://orcid.org/0000-0002-4590-798X}{Supported by the ISF, grant No.~1456/18, and the ERC, grant number 949707.}

\author{Paul G. Spirakis}{Department of Computer Science, University of Liverpool, UK}{p.spirakis@liverpool.ac.uk}{https://orcid.org/0000-0001-5396-3749}{Supported by the EPSRC grant EP/P02002X/1.}

\authorrunning{Nina Klobas, George B. Mertzios, Hendrik Molter, and Paul G. Spirakis} 

\Copyright{Nina Klobas, George B. Mertzios, Hendrik Molter, and Paul G. Spirakis} 

\ccsdesc[500]{Theory of computation~Graph algorithms analysis}
\ccsdesc[500]{Mathematics of computing~Discrete mathematics}

\keywords{Temporal graph, periodic temporal labeling, fastest temporal path, graph realization, temporal connectivity, parameterized complexity.} 

\category{} 

\relatedversion{} 

\nolinenumbers 

\EventEditors{John Q. Open and Joan R. Access}
\EventNoEds{2}
\EventLongTitle{42nd Conference on Very Important Topics (CVIT 2016)}
\EventShortTitle{CVIT 2016}
\EventAcronym{CVIT}
\EventYear{2016}
\EventDate{December 24--27, 2016}
\EventLocation{Little Whinging, United Kingdom}
\EventLogo{}
\SeriesVolume{42}
\ArticleNo{23}

\begin{document}
	\maketitle
	
	\begin{abstract}
		In this paper we initiate the study of the \emph{temporal graph realization} problem with respect to the fastest path durations among its vertices, 
		while we focus on periodic temporal graphs. 
		Given an $n \times n$ matrix $D$ and a $\Delta \in \mathbb{N}$, the goal is to construct a $\Delta$-periodic temporal graph with $n$ vertices 
		such that the duration of a \emph{fastest path} from $v_i$ to $v_j$ is equal to $D_{i,j}$, or to decide that such a temporal graph does not exist. 
		The variations of the problem on static graphs has been well studied and understood since the 1960's (e.g.\ [Erd\H{o}s and Gallai, 1960], [Hakimi and Yau, 1965]).
		
		As it turns out, the periodic temporal graph realization problem has a very different computational complexity behavior than its static (\ie non-temporal) counterpart. 
		First we show that the problem is NP-hard in general, but polynomial-time solvable if the so-called underlying graph is a tree.
		Building upon those results, we investigate its parameterized computational complexity with respect to structural parameters of the underlying static graph which measure the ``tree-likeness''. We prove a tight classification between such parameters that allow fixed-parameter tractability (FPT) 
		and those which imply W[1]-hardness. 
		We show that our problem is W[1]-hard when parameterized by the \emph{feedback vertex number} (and therefore also any smaller parameter such as \emph{treewidth}, \emph{degeneracy}, and \emph{cliquewidth}) of the underlying graph, while we show that it is in FPT when parameterized by the \emph{feedback edge number} (and therefore also any larger parameter such as \emph{maximum leaf number}) of the underlying graph.

	\end{abstract}
	

	\section{Introduction}\label{intro-sec}
	
	The (static) \emph{graph realization} problem with respect to a graph property $\mathcal{P}$ is to find a graph that satisfies property $\mathcal{P}$, or to decide that no such graph exists. 
	The motivation for graph realization problems stems both from ``verification'' and from network design applications in engineering. 
	In \emph{verification} applications, given the outcomes of some experimental measurements (resp.~some computations) on a network, 
	the aim is to (re)construct an input network which complies with them. 
	If such a reconstruction is not possible, this proves that the measurements are incorrect or implausible (resp.~that the algorithm which made the computations is incorrectly implemented). 
	One example of a graph realization (or reconstruction) problem is the recognition of probe interval graphs, in the context of the physical mapping of DNA, see~\cite{McMorris98,McConnellS02} and~\cite[Chapter 4]{GolumbicTrenk04}.
	In \emph{network design} applications, the goal is to design network topologies having a desired property~\cite{augustine2022distributed,grotschel1995design}.
	Analyzing the computational complexity of the graph realization problems for various natural and fundamental graph properties $\mathcal{P}$ requires a deep understanding of these properties.
	Among the most studied such parameters for graph realization 
	are constraints on the distances between vertices~\cite{barNoy2022GraphRealization,barNoy2021composed,hakimi1965distance,chung2001distance,bixby1988almost,culberson1989fast}, 
	on the vertex degrees~\cite{GolovachM17,gomory1961multi,hakimi1962realizability,Bar-NoyCPR20,erdos1960graphs}, 
	on the eccentricities~\cite{barNoy2020efficiently,hell2009linear,behzad1976eccentric,lesniak1975eccentric}, and on connectivity~\cite{fulkerson1960zero,frank1992augmenting,chen1966realization,frank1994connectivity,frank1970connectivity,gomory1961multi}, among others.

	In the simplest version of a (static) graph realization problem with respect to vertex distances, 
	we are given a symmetric $n \times n$ matrix $D$ and we are looking for an $n$-vertex undirected and unweighted graph $G$ such that $D_{i,j}$ equals the distance between vertices $v_i$ and $v_j$ in~$G$. This problem can be trivially solved in polynomial time in two steps~\cite{hakimi1965distance}: First, we build the graph $G=(V,E)$ such that $v_i v_j \in E$ if and only if $D_{i,j}=1$. Second, from this graph $G$ we compute the matrix $D_G$ which captures the shortest distances for all pairs of vertices. If $D_G = D$ then $G$ is the desired graph, otherwise there is no graph having $D$ as its distance matrix. 
	Non-trivial variations of this problem have been extensively studied, such as for weighted graphs~\cite{hakimi1965distance,Patrinos-Hakimi-72}, as well as for cases where the realizing graph has to belong to a specific graph family~\cite{hakimi1965distance, barNoy2021composed}. Other variations of the problem include the cases where every entry of the input matrix $D$ may contain a range of consecutive permissible values~\cite{barNoy2021composed,Rubei16,Tamura93}, or even an arbitrary set of acceptable values~\cite{barNoy2022GraphRealization} for the distance between the corresponding two vertices. 
	
	In this paper we make the first attempt to understand the complexity of the graph realization problem with respect to vertex distances in the context of \emph{temporal graphs}, \ie of graphs whose \emph{topology changes over time}. 
	
	\begin{definition}[temporal graph~\cite{KKK00}]
		\label{temp-graph-def} A \emph{temporal graph} is a pair $(G,\lambda)$,
		where $G=(V,E)$ is an underlying (static) graph and $\lambda :E\rightarrow 2^\mathbb{N}$ is a \emph{time-labeling} function which assigns to every edge of $G$ a set of discrete time-labels.
	\end{definition}
	
	Here, whenever $t \in \lambda(e)$, we say that the edge $e$ is \emph{active} or \emph{available} at time $t$. In the context of temporal graphs, where the notion of vertex adjacency is time-dependent, the notions of path and distance also need to be redefined. The most natural temporal analogue of a path is that of a \emph{temporal} (or \emph{time-dependent}) path, which is motivated by the fact that, due to
	causality, entities and information in temporal graphs can ``flow'' only along sequences of
	edges whose time-labels are strictly increasing.
	
	\begin{definition}[fastest temporal path] \label{def:temporalPath+Duration}
		Let $(G,\lambda)$ be a temporal graph. A \emph{temporal path} 
		in $(G,\lambda)$ is a sequence $(e_1,t_1),(e_2,t_2),\ldots,(e_k,t_k)$, 
		where $P=(e_1,\ldots,e_k)$ is a path in the underlying static graph $G$, 
		$t_i\in \lambda(e_i)$ for every $i=1,\ldots,k$, and $t_1<t_2<\ldots<t_k$. 
		The \emph{duration} of this temporal path 
		is $t_k - t_1 + 1$.
		A \emph{fastest} temporal path from a vertex $u$ to a vertex $v$ in $(G,\lambda)$ is a temporal path from $u$ to $v$ with the smallest duration.
		The duration of the \emph{fastest} temporal path from $u$ to $v$ is denoted by $d(u,v)$.
	\end{definition}
	

	In this paper we consider \emph{periodic} temporal graphs, \ie temporal graphs in which the temporal availability of each edge of the underlying graph is periodic. 
	Many natural and technological systems exhibit a periodic temporal behavior. For example, in railway networks an edge is present at a time step $t$ if and
	only if a train is scheduled to run on the respective rail segment at time $t$~\cite{Arrighi2023Multi}. 
	Similarly, a satellite, which makes pre-determined periodic movements, can establish a communication link (\ie a temporal edge) with another satellite whenever they are sufficiently close to each other; the existence of these communication links is also periodic. 
	In a railway (resp.~satellite) network, a fastest temporal path from $u$ to $v$ represents the fastest railway connection between two stations 
	(resp.~the quickest communication delay between two moving satellites). 
	Furthermore, periodicity appears also in (the otherwise quite complex) social networks which describe the dynamics of people meeting~\cite{snapnets,sapiezynski2015tracking}, as every person individually follows mostly a weekly routine.

	Expanding the work on periodic temporal graphs have already been studied 
	(see~\cite[Class 8]{casteigts2012time} and~\cite{Arrighi2023Multi,ErlebachS20,morawietz2021timecop,morawietz2020timecop}), 
	our study represents the first attempt to understand the complexity of a graph realization problem in the context of temporal graphs. 
	Therefore, we focus in this paper on the most fundamental case, where all edges have the same period $\Delta$ 
	(while in the more general case, each edge $e$ in the underlying graph has a period $\Delta_e$).
	As it turns out, the periodic temporal graph realization problem with respect to a given $n \times n$ matrix $D$ of the fastest duration times has a very different computational complexity behavior than the classic graph realization problem with respect to shortest path distances in static graphs.




	Formally, let $G=(V,E)$ and $\Delta\in \mathbb{N}$, and let $\lambda: E \rightarrow \{1,2,\ldots,\Delta\}$ be an edge-labeling function that assigns to every edge of $G$ exactly one of the labels from $\{1,\ldots,\Delta\}$. 
	Then we denote by $(G,\lambda,\Delta)$ the \emph{$\Delta$-periodic temporal graph} $(G,L)$, where for every edge $e\in E$ we have $L(e)=\{i\Delta + x : i\geq 0, x\in \lambda(e)\}$. 
	In this case we call $\lambda$ a \emph{$\Delta$-periodic labeling} of~$G$; see \cref{fig:periodic-example} for an illustration. 
	When it is clear from the context, we drop $\Delta$ from the notation and 
	we denote the ($\Delta$-periodic) temporal graph by $(G,\lambda)$.
	Given a duration matrix~$D$, it is easy to observe that, similarly to the static case, if $D_{i,j}=1$ then $v_i$ and $v_j$ must be connected by an edge. We call the graph defined by these edges the \emph{underlying graph} of $D$.

	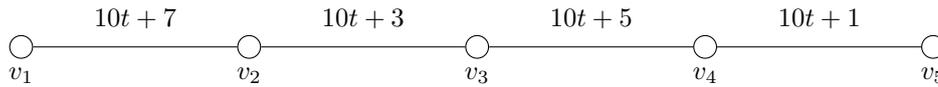
\begin{figure}[t]
		\centering
		\begin{tikzpicture}[xscale=1.5]
			\node[vert,label=below:$v_1$] (1) at (1,0) {};
			\node[vert,label=below:$v_2$] (2) at (3,0) {};
			\node[vert,label=below:$v_3$] (3) at (5,0) {};
			\node[vert,label=below:$v_4$] (4) at (7,0) {};
			\node[vert,label=below:$v_5$] (5) at (9,0) {};
			\draw (1) -- node[label=above:$10t+7$] {} (2) -- node[label=above:$10t+3$] {}  (3) -- node[label=above:$10t+5$] {} (4) -- node[label=above:$10t+1$] {} (5);	
		\end{tikzpicture}
		\caption{An example of a $\Delta$-periodic temporal graph $(G,\lambda,\Delta)$, where $\Delta = 10$ and the 10-periodic labeling $\lambda: E \rightarrow \{1,2,\ldots,10\}$ is as follows: $\lambda(v_1 v_2)=7$, $\lambda(v_2 v_3)=3$, $\lambda(v_3 v_4)=5$, and $\lambda(v_4 v_5)=1$. 
			Here, the fastest temporal path from $v_1$ to $v_2$ traverses the first edge $v_1v_2$ at time $7$, second edge $v_2v_3$ a time $13$, third edge $v_3v_4$ at time $15$ and the last edge $v_4v_5$ at time $21$.
			This results in the total duration of $21 - 7 + 1 = 15$ for the fastest temporal path from $v_1$ to $v_5$.
			\label{fig:periodic-example}}
	\end{figure}
	
	\subparagraph{Our contribution.}
	We initiate the study of naturally motivated graph realization problems in the temporal setting. 
	Our target is not to model unreliable communication, but instead to \emph{verify} that particular measurements regarding fastest temporal paths in a periodic temporal graph are plausible (\ie ``realizable''). 
	To this end, we introduce and investigate the following problem, capturing the setting described above:

	\problemdef{\deltaExactLong\ (\deltaExact)}
	{An integer $n \times n$ matrix $D$, a positive integer $\Delta$.}
	{Does there exist a graph $G=(V,E)$ with vertices $\{v_1,\ldots,v_{n}\}$ 
		and a $\Delta$-periodic labeling $\lambda: E \rightarrow \{1,2,\ldots,\Delta\}$ such that, 
		for every $i,j$, the duration of the fastest temporal path from $v_i$ to $v_j$ in the $\Delta$-periodic temporal graph $(G,\lambda,\Delta)$ is $D_{i,j}$?}

	We focus on exact algorithms. We start by showing NP-hardness of the problem (\cref{thm:NPhardness}), even if $\Delta$ is a small constant. To establish a baseline for tractability, we show that \deltaExact\ is polynomial-time solvable if the underlying graph is a tree (\cref{thm:deltaExact-PolyTimeTrees}).
	
	Building upon these initial results, we explore the possibilities to generalize our polynomial-time algorithm using the \emph{distance-from-triviality} parameterization paradigm~\cite{FJR13,GHN04}. That is, we investigate the parameterized computational complexity of \deltaExact\ with respect to structural parameters of the underlying graph that measure its ``tree-likeness''.
	
	We obtain the following results. We show that \deltaExact\ is W[1]-hard when parameterized by the feedback vertex number of the underlying graph (\cref{thm:W1wrtFVS}). 
	To this end, we first give a reduction from \textsc{Multicolored Clique} parameterized by the number of colors~\cite{fellows2009multipleinterval} to a variant of \deltaExact\ where the period $\Delta$ is infinite, that is, when the labeling is non-periodic. We use a special gadget (the ``infinity'' gadget) which allows us to transfer the result to a finite period $\Delta$. The latter construction is independent from the particular reduction we use, and can hence be treated as a reduction from the non-periodic to the periodic setting.
	Note that our parameterized hardness result rule out fixed-parameter tractability for several popular graph parameters such as \emph{treewidth}, \emph{degeneracy}, \emph{cliquewidth}, 
	\emph{distance to chordal graphs}, and \emph{distance to outerplanar graphs}.
	
	We complement this hardness result by showing that \deltaExact\ is fixed-parameter tractable (FPT) with respect to the \emph{feedback edge number} $k$ of the underlying graph (\cref{thm:FPTwrtFES}). 
	This result also implies an FPT algorithm for any larger parameter, such as the \emph{maximum leaf number}. 
	A similar phenomenon of getting W[1]-hardness with respect to the feedback vertex number, while getting an FPT algorithm with respect to the feedback edge number, has been observed only in a few other temporal graph problems related to the connectivity between two vertices~\cite{casteigts2021finding,FMNR22a,EMM22}.
	
	Our FPT algorithm works as follows on a high level. 
	First we distinguish $O(k^2)$ vertices which we call ``important vertices''. 
	Then, we guess the fastest temporal paths for each pair of these important vertices; as we prove, the number of choices we have for all these guesses is upper bounded by a function of~$k$. 
	Then we also need to make several further guesses (again using a bounded number of choices), which altogether leads us to specify a small (\ie bounded by a function of $k$) number of different configurations for the fastest paths between \emph{all pairs} of vertices. For each of these configurations, we must then make sure that the labels of our solution will not allow any other temporal path from a vertex $v_i$ to a vertex $v_j$ have a \emph{strictly smaller} duration than~$D_{i,j}$.
	This naturally leads us to build one Integer Linear Program (ILP) for each of these configurations. We manage to formulate all these ILPs by having a number of variables that is upper-bounded by a function of $k$. Finally we use Lenstra's Theorem~\cite{Lenstra1983Integer} to solve each of these ILPs in FPT time. At the end, our initial instance is a \textsc{Yes}-instance if and only if at least one of these ILPs is feasible.

	The above results provide a fairly complete picture of the parameterized computational complexity of \deltaExact\ with respect to structural parameters of the underlying graph which measure ``tree-likeness''. To obtain our results, we prove several properties of fastest temporal paths, which may be of independent interest.

	\subparagraph{Related work.} Graph realization problems on static graphs have been studied since the 1960s. We provide an overview of the literature in the introduction. 
	To the best of our knowledge, we are the first to consider graph realization problems in the temporal setting. 
	Very recently, Erlebach et al.~\cite{ErlebachMW-SAND24} have built upon our results and, among others, studied the case where edges might appear more than once in each period. 
	Many other connectivity-related problems have been studied in the temporal setting~\cite{Mertzios-transitivity21,Akrida-explorer-21,enright2021deleting,MolterRZ21,klobas2023interference,deligkas2022optimizing,erlebach2021temporal,Flu+19a,Zsc+19,CasteigtsCS22,FuchsleMNR22}, most of which are much more complex and computationally harder than their non-temporal counterparts, and some of which do not even have a non-temporal counterpart.
	
	There are some problem settings that share similarities with ours, which we discuss now in more detail.
	
	Several problems have been studied where the goal is to assign labels to (sets of) edges of a given static graph in order to achieve certain connectivity-related properties~\cite{KlobasMMS22,MertziosMS19,akrida2017complexity,enright2021assigning}. The main difference to our problem setting is that in the mentioned works, the input is a graph and the sought labeling is not periodic. Furthermore, the investigated properties are temporal connectivity between all vertices~\cite{KlobasMMS22,MertziosMS19,akrida2017complexity}, temporal connectivity among a subset of vertices~\cite{KlobasMMS22}, or reducing reachability among the vertices~\cite{enright2021assigning}. In all these cases, the duration of the temporal paths has not been considered.


	Finally, there are many models for dynamic networks in the context of distributed computing~\cite{Kuhn2011Dynamic}. 
	These models have some similarity to temporal graphs, in the sense that in both cases the edges appear and disappear over time. 
	However, there are notable differences. For example, one important assumption in the distributed setting 
	can be that the edge changes are adversarial or random (while obeying some constraints such as connectivity), 
	and therefore they are not necessarily known in advance~\cite{Kuhn2011Dynamic}.

	\subparagraph{Preliminaries and notation.}
	We already introduced the most central notion and concepts. There are some additional definitions we need, to present our proofs and results which we give in the following. 
	
	An interval in $\mathbb N$ from $a$ to $b$ is denoted by $[a,b] = \{ i\in \mathbb N  :  a \leq i \leq b\}$; similarly, $[a] = [1,a]$.
	An undirected graph~$G=(V,E)$ consists of a set~$V$ of vertices 
	and a set~$E \subseteq V \times V$ of edges.
	For a graph~$G$, we also denote by~$V(G)$ and~$E(G)$ the vertex and edge set of~$G$, respectively.
	We denote an edge $e \in E$ between vertices $u,v \in V$ as a set $e=\{u,v\}$.
	For the sake of simplicity of the representation, an edge $e$ is sometimes also denoted by $uv$. 
	A path~$P$ in $G$ is a subgraph of $G$ with vertex set~$V(P)=\{v_1,\dots,v_k\}$ and edge set~$E(P)=\{\{v_i,v_{i+1}\} :  1\leq i<k\}$
	(we often represent path~$P$ by the tuple~$(v_1,v_2,\dots,v_k)$).
	
	Let $v_1,v_2,\ldots,v_n$ be the $n$ vertices of the graph $G$. 
	For simplicity of the presentation (and with a slight abuse of notation) we refer during the paper to the entry $D_{i,j}$ of the matrix $D$ as $D_{a,b}$, where $a=v_i$ and $b=v_j$. 
	That is, we put as indices of the matrix $D$ the corresponding vertices of $G$ whenever it is clear from the context. 
	
	Let $P=(u=v_1, v_2, \dots, v_p=v)$ be a path from $u$ to $v$ in $G$. 
	Recall that, in our paper, every edge has exactly one time label in every period of $\Delta$ consecutive time steps.
	Therefore, as we are only interested in the fastest duration of temporal paths, 
	many times we refer to $(P,\lambda,\Delta)$ as any of the temporal paths from $u=v_1$ to $v=v_p$ along the edges of $P$, which starts at the edge $v_1 v_2$ at time $\lambda(v_1 v_2) + c \Delta$, for some $c\in \mathbb{N}$, and then sequentially visits the rest of the edges of $P$ as early as possible. 
	We denote by $d(P,\lambda,\Delta)$, or simply by $d(P,\lambda)$ when $\Delta$ is clear from the context, the duration of any of the temporal paths $(P,\lambda,\Delta)$; note that they all have the same duration. 
	Whenever we use the term \emph{label of an edge} $e$, we actually mean $\lambda(e) \in [\Delta]$. Note that for a given path $(P, \lambda, \Delta)$ that passes through the edge $e$, the label used by $P$ at that edge is $\lambda(e) + c \Delta$, for some $c \geq 0$.
	Many times we also refer to a path $P=(u=v_1, v_2, \dots, v_p=v)$ from $u$ to $v$ in $G$,
	as a temporal path in $(G,\lambda,\Delta)$,
	where we actually mean that $(P,\lambda,\Delta)$ is a temporal path with $P$ as its underlying (static) path.

	We remark that a fastest path between two vertices in a temporal graph can be computed in polynomial time~\cite{xuan_computing_2003,Wu2016Efficient}.
	Hence, given a $\Delta$-periodic temporal graph $(G,\lambda,\Delta)$, we can compute in polynomial-time  the matrix $D$
	which consists of durations of fastest temporal paths among all pairs of vertices in $(G,\lambda,\Delta)$.
	
	We use standard terminology from parameterized complexity theory~\cite{DF13,FG06,Cyg+15}.
	Let~$\Sigma$ denote a
	finite alphabet.
	A parameterized problem~$L\subseteq \{(x,k)\in \Sigma^*\times \mathbb N_0\}$ is a subset of all instances~$(x,k)$ from~$\Sigma^*\times \mathbb N_0$,
	where~$k$ denotes the \emph{parameter}.
	A parameterized problem~$L$ is 
	FPT (\emph{fixed-parameter tractable}) if there is an algorithm that decides every instance~$(x,k)$ for~$L$ in~$f(k)\cdot |x|^{O(1)}$ time,
	where~$f$ is any computable function only depending on the parameter.
	If a parameterized problem $L$ is W[1]-hard, then it is presumably not
	fixed-parameter tractable.

	\subparagraph{Organization of the paper.}
	In \cref{sec:hardness} we present our hardness results, first the NP-hardness in \cref{sec:nphardness} and then the parameterized hardness in \cref{sec:w1hardness}. In \cref{sec:algos} we present our algorithmic results. First we give in \cref{sec:treealgo} a polynomial-time algorithm for the case where the underlying graph is a tree. In \cref{sec:FPT} we generalize this and present our FPT result, which is the main result in the paper. Finally, we conclude in \cref{sec:conclusion} and discuss some future work directions.
	
	\section{Hardness results for \deltaExact}\label{sec:hardness}
	In this section we present our main computational hardness results. In \cref{sec:nphardness} we show that \deltaExact\ is NP-hard even for constant $\Delta$. In \cref{sec:w1hardness} we investigate the parameterized computational hardness of \deltaExact\ with respect to structural parameters of the underlying graph. We show that \deltaExact\ is W[1]-hard when parameterized by the feedback vertex number of the underlying graph.
	
	\subsection{NP-hardness of \deltaExact}\label{sec:nphardness}
	In this section we prove that in general it is NP-hard to determine a $\Delta$-periodic temporal graph $(G,\lambda)$ respecting a duration matrix $D$,
	even if $\Delta$ is a small constant.
	
	\begin{theorem}\label{thm:NPhardness}
		\deltaExact\ is \NP-hard for all $\Delta \geq 3$.
	\end{theorem}
	
	\begin{proof}
		We present a polynomial-time reduction from the NP-hard problem NAE 3-SAT~\cite{Schaefer1978complexity}. Here we are given a formula $\phi$ that is a conjunction of so-called NAE (not-all-equal) clauses, where each clause contains exactly 3 literals (with three distinct variables).
		A NAE clause evaluates to \textsc{true} if and only if not all of its literals are equal, that is, at least one literal evaluates to \textsc{true} and at least one literal evaluates to \textsc{false}.
		We are asked whether $\phi$ admits a satisfying assignment.

		Given an instance $\phi$ of NAE 3-SAT, we construct an instance $(D,\Delta)$ of \deltaExact\ as follows.
		
		We start by describing the vertex set of the underlying graph $G$ of $D$.
		\begin{itemize}
			\item For each variable $x_i$ in $\phi$, we create three variable vertices $x_i, x_i^T, x_i^F$.
			\item For each clause $c$ in $\phi$, we create one clause vertex $c$.
			\item We add one additional super vertex $v$.
		\end{itemize}
		Next, we describe the edge set of $G$.
		\begin{itemize}
			\item For each variable $x_i$ in $\phi$ we add the following five edges: 
			$\{x_i, x_i^T\}$, $\{x_i, x_i^F\}$, $\{x_i^T, x_i^F\}$, $\{x_i^T, v\}$, and $\{x_i^F,v\}$.
			\item For each pair of variables $x_i,x_j$ in $\phi$ with $i \neq j$ we add the following four edges: 
			$\{x_i^T, x_j^T\}$, $\{x_i^T,x_j^F\}$, $\{x_i^F,x_j^T\}$, and $\{x_i^F,x_j^F\}$.
			\item For each clause $c$ in $\phi$ we add one edge for each literal. Let $x_i$ appear in $c$. If $x_i$ appears non-negated in $c$ we add edge $\{c, x_i^T\}$. 
			If $x_i$ appears negated in $c$ we add edge $\{c, x_i^F\}$.
		\end{itemize}
		This finishes the construction of $G$.
		For an illustration see~\cref{fig:NP-example}.

		\begin{figure}[t]
			\noindent
			\makebox[\textwidth]{
				\centering
				\includegraphics{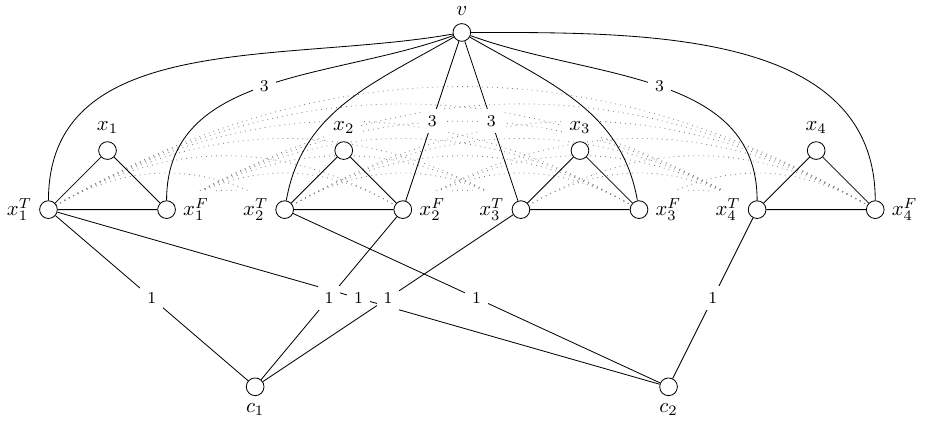}
			}
			\caption{Illustration of the temporal graph $(G,\lambda)$ from the NP-hardness reduction, 
				where the NAE 3-SAT formula $\phi$ is of the form $\phi = \text{NAE}(x_1, \overline{x}_2, x_3) \wedge \text{NAE}(x_1, x_2, x_4)$.
				To improve the readability, we draw edges between vertices $x_i^T$ and $x_j^F$ (where $i \neq j$) with gray dotted lines.
				Presented is the labeling of $G$ corresponding to the assignment $x_1=x_2=\textsc{true}$ and $x_3,x_4=\textsc{false}$,
				where all unlabeled edges get the label $2$.
			}\label{fig:NP-example}
		\end{figure}
		We set $\Delta$ to some constant larger than two, that is, $\Delta\ge 3$. Next, we specify the durations in the matrix $D$ between all vertex pairs.
		For the sake of simplicity we write $D_{u,v}$ as $d(u,v)$,
		where $u,v$ are two vertices of $G$. 
		We start by setting the value of $d(u,v) = 1$ where $u$ and $v$ are two adjacent vertices in $G$.
		\begin{itemize}
			\item For each variable $x_i$ in $\phi$ and the super vertex $v$
			we specify the following durations:
			$d(x_i,v)=2$ and $d(v,x_i)= \Delta $. 
			\item For each clause $c$ in $\phi$ and the super vertex $v$
			we specify the following durations:
			$d(c,v)=2$ and $d(v,c)= \Delta - 1$.
			\item Let $x_i$ be a variable that appears in clause $c$, then  we specify the following durations:
			$d(c,x_i)=2$ and $d(x_i,c)=\Delta$.
			If $x_i$ appears non-negated in $c$ we specify the following durations:
			$d(c,x_i^F)=2$ and $d(x_i^F,c)=\Delta$.
			If $x_i$ appears negated in $c$ we specify the following duratios:
			$d(c,x_i^T)=2$ and $d(x_i^T,c)=\Delta$.
			\item Let $x_i$ be a variable that does \emph{not} appear in clause $c$, then we specify the following duratios:
			$d(x_i,c)=2 \Delta$, $d(c,x_i)=\Delta + 2$
			and
			$d(c,x_i^T)=d(c,x_i^F)=2$, $d(x_i^T,c)=d(x_i^F,c)=\Delta$. 
			\item For each pair of variables $x_i \neq x_j$ in $\phi$ we specify the following duratios:
			$d(x_i,x_j)=2\Delta +1$ and
			$d(x_i,x_j^T)=d(x_i,x_j^F)=\Delta + 1$.
			\item For each pair of clauses $c_i \neq c_j$ in $\phi$ we specify the following duratios:
			$d(c_i,c_j)= \Delta + 1$.
		\end{itemize}
		This finishes the construction of the instance $(D,\Delta)$ of \deltaExact\, which can clearly be done in polynomial time. In the remainder we show that $(D,\Delta)$ is a \textsc{Yes}-instance of \deltaExact\ if and only if NAE 3-SAT formula $\phi$ is satisfiable.
		
		$(\Rightarrow)$: Assume the constructed instance $(D,\Delta)$ of \deltaExact\ is a \textsc{Yes}-instance. 
		Then there exist a label $\lambda(e)$ for each edge $e\in E(G)$ such that for each vertex pair $u,w$ in the temporal graph $(G,\lambda,\Delta)$ we have that a fastest temporal path from $u$ to $w$ is of duration $d(u,w)$. 
		
		We construct a satisfying assignment for $\phi$ as follows. For each variable $x_i$, 
		if $\lambda(\{x_i, x_i^T\})=\lambda(\{x_i^T, v\})$, then we set $x_i$ to \textsc{true}, otherwise we set $x_i$ to \textsc{false}.
		
		To show that this yields a satisfying assignment, we need to prove some properties of the labeling $\lambda$.
		First, observe that adding an integer $t$ to all time labels does not change the duration of any temporal paths. 
		Second, observe that if for two vertices $u,w$ we have that $d(u,w)$ equals the distance between $u$ and $w$ in $G$
		(\ie the duration of the fastest temporal path from $u$ to $w$ equals the distance of the shortest path between $u$ and $w$), 
		then there is a shortest path $P$ from $u$ to $w$ in $G$ such that 
		the labeling $\lambda$ assigns consecutive time labels to the edges of $P$.

		Let $\lambda(\{x_i, x_i^T\})=t$ and $\lambda(\{x_i, x_i^F\})=t'$, for an arbitrary variable $x_i$. 
		If both $\lambda(\{x_i^T, v\})\neq t+1$ and $\lambda(\{x_i^F, v\})\neq t'+1$, then $d(x_i,v)>2$, which is a contradiction. 
		Thus, for every variable $x_i$, we have that $\lambda(\{x_i^T, v\})= t+1$ or $\lambda(\{x_i^F, v\})= t'+1$ (or both). 
		In particular, this means that if $\lambda(\{x_i, x_i^F\})=\lambda(\{x_i^F, v\})$, then we set $x_i$ to \textsc{false}, since in this case $\lambda(\{x_i, x_i^T\})\neq\lambda(\{x_i^T, v\})$.
		
		
		Now assume for a contradiction that the described assignment is not satisfying. Then there exists a clause $c$ that is not satisfied. 
		Suppose that $x_1, x_2, x_3$ are three variables that appear in $c$.
		Recall that we require $d(c,v)=2$ and $d(v,c)=\Delta -1$. 
		The fact that $d(c,v)=2$ implies that we must have a temporal path consisting of two edges from $c$ to $v$, 
		such that the two edges have consecutive labels. 
		By construction of $G$ there are three candidates for such a path, one for each literal of $c$. 
		Assume w.l.o.g.\ that $x_1$ appears in $c$ non-negated (the case of a negated appearance of $x_1$ is symmetrical) and that the temporal path realizing $d(c,v)=2$ goes through vertex $x_1^T$. 
		Let us denote with $t = \lambda(\{x_1^T, v\})$.
		It follows that $\lambda(\{x_1^T, c\})=\lambda(\{x_1^T, v\})-1 = t - 1$.
		Furthermore, since $d(c,x_1)=2$ we also have that $\lambda(\{x_1^T, c\})=\lambda(\{x_1, x_1^T\})-1$. 
		Therefore $\lambda(\{x_1, x_1^T\})=\lambda(\{x_1^T, v\}) = t$. 
		Which implies that $x_1$ is set to \textsc{true}.
		Let us observe paths from $v$ to $c$.
		We know that $d(v,c)=\Delta -1$.
		The underlying path of the fastest temporal path from $v$ to $c$, that goes through $x_1^T$ is the path $P = (v,x_1^T,c)$.
		Since $\lambda(\{x_1^T,c\}) > \lambda(\{x_1^T, v\})$ we get that the duration of the temporal path $(P,\lambda)$ is equal to 
		$d(P,\lambda)= (\Delta + t-1) - t + 1 = \Delta$.
		This implies that
		the fastest temporal path from $v$ to $c$ is not $(P,\lambda)$ and therefore does not pass through $x_1^T$.
		Since there are only two other vertices connected to $c$, 
		we have only two other edges incident to $c$, that can be used on a fastest temporal path
		$v$ to $c$.
		Suppose now w.l.o.g.\ that also $x_2$ appears in $c$ non-negated (the case of a negated appearance of $x_2$ is symmetrical) and that the temporal path realizing $d(v,c)=\Delta-1$ goes through vertex $x_2^T$.
		Let us denote with $t' = \lambda(\{x_2^T, v\})$.
		Since the fastest temporal path from $v$ to $c$ is of the duration $\Delta - 1$,
		and the edge $x_2^T c$ is the only edge incident to vertex $c$ and edge $\{x_2^T, v\}$,
		it follows that $\lambda(\{x_2^T, c\}) \geq \lambda(\{x_2^T, v\}) - 2 = t' - 2$.
		Since $d(x_2,v) = 2$ it follows that $\lambda(\{x_2, x_2^T\}) = \lambda (\{x_2^T, v\}) - 1 = t' - 1$.
		Knowing this and the fact that $d(x_2,c)=2$, we get that $\lambda(\{x_2^T, c\})$ must be equal to $t'-2$.
		Therefore the fastest temporal path from $v$ to $c$ passes through edges $\{x_2^T,v\}$ and $\{x_2^T,c\}$.
		In the above we have also determined that $\lambda(\{x_2,x_2^T\}) \neq \lambda ( \{x_2^T,v\})$,
		which implies that $x_2$ is set to \textsc{false}.
		But now we have that $x_1,x_2$ both appear in $c$ non-negated, where one of them is \textsc{true}, while the other is \textsc{false},
		which implies that the clause $c$ is satisfied, a contradiction.

		
		$(\Leftarrow)$: Assume that $\phi$ is satisfiable. Then there exists a satisfying assignment for the variables in $\phi$.
		
		We construct a labeling $\lambda$ as follows.
		\begin{itemize}
			\item All edges incident with a clause vertex $c$ obtain label one.
			\item If variable $x_i$ is set to \textsc{true}, we set $\lambda(\{x_i^F, v\})=3$.
			\item If variable $x_i$ is set to \textsc{false}, we set $\lambda(\{x_i^T, v\})=3$.
			\item We set the labels of all other edges to two.
		\end{itemize}
		For an example of the constructed temporal graph see~\cref{fig:NP-example}.
		We now verify that all duratios are realized.
		\begin{itemize}
			\item For each variable $x_i$ in $\phi$ we have to check that $d(x_i,v)=2$
			and $d(v,x_i)=\Delta$. 
			
			If $x_i$ is set to \textsc{true}, then there is a temporal path from $x_i$ to $v$ via 
			$x_i^F$ of duration $2$, since
			$\lambda(\{x_i, x_i^F\})=2$ and $\lambda(\{x_i^F, v\})=3$.
			For a temporal path from $v$ to $x_i$ we observe the following.
			The only possible labels to leave the vertex $v$ are $2$ and $3$, which take us from $v$ to $x_j^T$ or $x_j^F$ of some variable $x_j$.
			The only two edges incident to $x_i$ have labels $2$, therefore the fastest path from $v$ to $x_i$
			cannot finish before the time $\Delta + 2$.
			The fastest way to leave $v$ and enter to $x_i$ would then be to leave $v$ at edge $\{x_i^F,v\}$ with label $3$,
			and continue to $x_i$ at time $\Delta + 2$,
			which gives us the desired duration $\Delta$.
			
			If $x_i$ is set to \textsc{false}, then, by similar arguing, 
			there is a temporal path from $x_i$ to $v$ via $x_i^T$ of duration $2$,
			and a temporal path from $v$ to $x_i$, through $x_i^F$ of duration $\Delta$.
			
			\item For each clause $c$ in $\phi$ we have to check that $d(c,v)=2$
			and $d(v,c)=\Delta - 1$:
			
			Suppose $x_i,x_j,x_k$ appear in $c$.
			Since we have a satisfying assignment at least one of the literals in $c$ is set to \textsc{true} and at least one to \textsc{false}. 
			Suppose $x_i$ is the variable of the literal that is \textsc{true} in $c$,
			and $x_j$ is the variable of the literal that is \textsc{false} in $c$.
			Let $x_i$ appear non-negated in $c$ and is therefore set to \textsc{true} (the case when $x_i$ appears negated in $c$ and is set to \textsc{false} is symmetric).
			Then there is a temporal path from $c$ to $v$ through $x_i^T$ such that $\lambda(\{x_i^T, c\})=1$ and $\lambda(\{x_i^T, v\})=2$. 
			Let $x_j$ appear non-negated in $c$ and is therefore set to \textsc{false} 
			(the case when $x_j$ appears negated in $c$ and is set to \textsc{true} is symmetric).
			Then there is a temporal path from $v$ to $c$ through $x_j^T$ such that $\lambda(\{x_j^T, v\})=3$ and $\lambda(\{x_j^T, c\})=1$,
			which results in a temporal path from $v$ to $c$ of duration $\Delta -1$.
			
			\item Let $x_i$ be a variable that appears in clause $c$.
			If $x_i$ appears non-negated in $c$ we have to check that $d(c,x_i)=d(c,x_i^F)=2$
			and $d(x_i,c)=d(x_i^F,c)= \Delta$.
			
			There is a temporal path from $c$ to $x_i$ via $x_i^T$ and also a temporal path from $c$ to $x_i^F$ via $x_i^T$ such that $\lambda(\{x_i^T, c\})=1$ and $\lambda(\{x_i, x_i^T\})=\lambda(\{x_i^T, x_i^F\})=2$,
			which proves the first equality.
			There are also the following two temporal paths,
			first, from $x_i$ to $c$ through $x_i^T$ and
			second, from $x_i^F$ to $c$ through $x_i^T$.
			Both of the temporal paths start on the edge with label $2$, as 
			$\lambda(\{x_i, x_i^T\}) = \lambda (\{x_i^T, x_i^F\})=2$ and 
			finish on the edge with label $1$, as $\lambda(\{x_i^T, c\}) = 1$.
			
			If $x$ appears negated in $c$ we have to check that 
			$d(c,x_i)=d(c,x_i^T)=2$
			and $d(x_i,c)=d(x_i^T,c)= \Delta$.
			
			There is a temporal path from $c$ to $x$ via $x^F$ and also a temporal path from $c$ to $x^T$ via $x^F$ such that $\lambda(\{c, x^F\})=1$ and $\lambda(\{x, x^F\})=\lambda(\{x^T, x^F\})=2$,
			which proves the first inequality.
			There are also the following two temporal paths,
			first, from $x_i$ to $c$ through $x_i^F$ and
			second, from $x_i^T$ to $c$ through $x_i^F$.
			Both of the temporal paths start on the edge with label $2$, as 
			$\lambda(\{x_i, x_i^F\}) = \lambda (\{x_i^T, x_i^F\})=2$ and 
			finish on the edge with label $1$, as $\lambda(\{x_i^F, c\}) = 1$.
			Which proves the second equality.
			
			\item Let $x_i$ be a variable that does \emph{not} appear in clause $c$, then we have to check that
			first,
			$d(c,x_i^T)=d(c,x_i^F)=2$,
			second, $d(x_i^T, c) = d(x_i^F,c)=\Delta$,
			third, $d(c,x_i)=\Delta + 2$,
			and fourth $d(x_i,c)=2 \Delta$.
			
			Let $x_j$ be a variable that appears non-negated in $c$ (the case where $x_j$ appears negated is symmetric). 
			Then there is a temporal path 
			from $c$ to $x_i^T$ via $x_j^T$ and also a temporal path 
			from $c$ to $x_i^F$ via $x_j^T$ such that $\lambda(\{x_j^T, c\})=1$ and $\lambda(\{x_j^T, x_i^T\})=\lambda(\{x_j^T, x_i^F\})=2$,
			which proves the first equality.
			Using the same temporal path in the opposite direction,
			\ie first the edge $x_j^T c$ and then one of the edges $\{x_j^T, x_i^F\}$ or $\{x_j^T, x_i^T\}$ at times $2$ and $\Delta + 1$, respectively,
			yields the second equality.
			For a temporal path from $c$ to $x_i$ we traverse the following three edges 
			$\{x_j^T, c\}$, $\{x_j^T, x_i^F\}$, and $\{x_i^F,x_i\}$,
			with labels $1$, $2$, and $2$ respectively (\ie the path traverses them at time $1,2$ and $\Delta + 2$, respectively), which proves the third equality.
			Now for the case of a temporal path from $x_i$ to $c$,
			we use the same three edges, but in the opposite direction,
			namely $\{x_i^F,x_i\}$, $\{x_j^T, x_i^F\}$, and $\{x_j^T, c\}$,
			again at times $2$, $\Delta + 2$, and $2\Delta + 1$, respectively,
			which proves the last equality.
			Note that all of the above temporal paths are also the shortest possible, 
			and since the labels of first and last edges (of these paths) are unique,
			it follows that we cannot find faster temporal paths.
			
			\item For each pair of variables $x_i \neq x_j$ in $\phi$ we have to check that
			$d(x_i,x_j)=2\Delta +1$ and
			$d(x_i,x_j^T)=d(x_i,x_j^F)=\Delta + 1$.
			
			There is a path from $x_i$ to $x_j$ that passes first through one of the vertices 
			$x_i^T$ or $x_i^F$, 
			and then through one of the vertices $x_j^T$ or $x_j^F$.
			This temporal path is of length $3$, where all of the edges have label $2$,
			which proves the first equality.
			Now, a temporal path from $x_i$ to $x_j^T$ (resp.~$x_j^F$),
			passes through one of the vertices $x_i^T$ or $x_i^F$.
			This path is of length two, where all of the edges have label $2$,
			which proves the second equality.
			Note that all of the above temporal paths are also the shortest possible, 
			and since the labels of first and last edges (of these paths) are unique,
			it follows that we cannot find faster temporal paths.
			
			\item For each pair of clauses $c_i \neq c_j$ in $\phi$ we have to check that
			$d(c_i,c_j)= \Delta + 1$.
			
			Let $x_k$ be a variable that appears non-negated in $c_i$
			and $x_\ell$ the variable that appears non-negated in $c_j$
			(all other cases are symmetric).
			There is a path of length three from $c_i$ to $c_j$
			that passes first through vertex $x_k^T$ and then through vertex $x_\ell^T$.
			Therefore the temporal path from $c_i$ to $c_j$
			uses the edges $\{x_k^T, c_i\}$, $\{x_\ell^T, c_j\}$, and $\{x_k^T, x_\ell ^T\}$, with labels $1$,~$2$, and $1$ (at times $1$,~$2$, and $\delta + 1$), respectively,
			which proves the desired equality.
			Note also that this is the shortest path between $c_i$ and $c_j$, and that the first and the last edge must have the label $1$,
			therefore it follows that this is the fastest temporal path.
		\end{itemize}
		Lastly, observe that the above constructed labeling $\lambda$ uses values $\{1,2,3\}\subseteq[\Delta]$, therefore $\Delta\ge 3$.
	\end{proof}

	\subsection{Parameterized hardness of \deltaExact}\label{sec:w1hardness}
	
	In this section, we investigate the parameterized hardness of \deltaExact\ with respect to structural parameters of the underlying graph. We show that the problem is W[1]-hard when parameterized by the feedback vertex number of the underlying graph. The \emph{feedback vertex number} of a graph $G$ is the cardinality of a minimum vertex set $X\subseteq V(G)$ such that $G-X$ is a forest. The set $X$ is called a \emph{feedback vertex set}. 
	Note that, in contrast to the result of the previous section (\cref{thm:NPhardness}), the reduction we use to obtain the following result does not produce instances with a constant $\Delta$.
	
	\begin{theorem}\label{thm:W1wrtFVS}
		\deltaExact\ is W[1]-hard when parameterized by the feedback vertex number of the underlying graph.
	\end{theorem}
	\begin{proof}
		We present a parameterized reduction from the W[1]-hard problem \textsc{Multicolored Clique} parameterized by the number of colors~\cite{fellows2009multipleinterval}.  Here, given a $k$-partite graph $H=(W_1\uplus W_2 \uplus\ldots\uplus W_k, F)$, we are asked whether $H$ contains a clique of size $k$. If $w\in W_i$, then we say that $w$ has \emph{color} $i$. W.l.o.g.\ we assume that $|W_1|=|W_2|=\ldots=|W_k|=n$. 
		Furthermore, for all $i\in[k]$, we assume the vertices in $W_i$ are ordered in some arbitrary but fixed way, that is, $W_i=\{w^i_1,w^i_2,\ldots,w^i_n\}$.
		Let $F_{i,j}$ with $i<j$ denote the set of all edges between vertices from $W_i$ and $W_j$. We assume w.l.o.g.\ that $|F_{i,j}|=m$ for all $i< j$  (if not we can add $k \max_{i,j}|F_{i,j}|$ vertices to each $W_i$ and use those to add up to $\max_{i,j}|F_{i,j}|$ additional isolated edges to each $F_{i,j}$).
		Furthermore, for all $i<j$ we assume that the edges in $F_{i,j}$ are ordered in some arbitrary but fixed way, that is, $F_{i,j}=\{e^{i,j}_1,e^{i,j}_2,\ldots,e^{i,j}_m\}$.
		
		We give a reduction to a variant of \deltaExact\ where the period $\Delta$ is infinite (that is, the sought temporal graph is not periodic and the labeling function $\lambda : E\rightarrow \mathbb{N}$ maps to the natural numbers) and we allow $D$ to have infinity entries, meaning that the two respective vertices are not temporally connected. 
		Note that, given the matrix $D$, we can easily compute the underlying graph $G$, as follows. Two vertices $v,v'$ are adjacent in $G$ if and only if $D_{v,v'}=1$, as having an edge between $v$ and $v'$ is the only way that there exists a temporal path from $v$ to $v'$ with duration 1. 
		For simplicity of the presentation of the reduction, we describe the underlying graph $G$ (which directly implies the entries of $D$ where $D_{v,v'}=1$) and then we provide the remaining entries of $D$. 
		At the end of the proof we show how to obtain the result for a finite $\Delta$ and a matrix $D$ of durations of fastest paths, that only has finite entries.

		In the following, we give an informal description of the main ideas of the reduction. The construction uses several gadgets, where the main ones are an ``edge selection gadget'' and a ``verification gadget''.

		Every \emph{edge selection gadget} is associated with a color combination $i,j$ in the \textsc{Multicolored Clique} instance, and its main purpose is to ``select'' an edge connecting a vertex from color $i$ with a vertex from color $j$.
		Roughly speaking, the edge selection gadget consists of $m$ paths, one for every edge in $F_{i,j}$ (see \cref{fig:hardness1} for reference). The distance matrix $D$ will enforce that the labels on those paths effectively order them temporally, that is, in particular, the labels on one of the paths will be smaller than the labels on all other paths. The edge corresponding to this path is selected.

		We have a \emph{verification gadget} for every color $i$. They interact with the edge selection gadgets as follows.
		The verification gadget for color $i$ is connected to all edge selection gadgets that involve color $i$. More specifically, this is connected to every path corresponding to an edge at a position in the path that encodes the endpoint of color $i$ of that edge (again, see \cref{fig:hardness1} for reference). Intuitively, the distances in the verification gadget are only realizable if the selected edges all have the same endpoint of color $i$.
		Hence, the distances of all verification gadgets can be realized if and only if the selected edges form a clique. 
		
		Furthermore, we use an \emph{alignment gadget} which, intuitively, ensures that the labelings of all gadgets use the same range of time labels. Finally, we use \emph{connector gadgets} which create shortcuts between all vertex pairs that are irrelevant for the functionality of the other gadgets. This allows us to easily fill in the distance matrix with the corresponding values.
		We ensure that all our gadgets have a constant feedback vertex number, hence the overall feedback vertex number is quadratic in the number of colors of the \textsc{Multicolored Clique} instance and we get the parameterized hardness result.
		
		In the following, for every gadget, we first give a formal description of the underlying graph of this gadget (\ie not the complete distance sub-matrix of the gadget). Afterwards, we define the corresponding entries in the distance matrix $D$.
		
		Given an instance $H$ of \textsc{Multicolored Clique}, we construct an instance $D$ of \deltaExact\ (with infinity entries and no periods) as follows. 
		
		\subparagraph{Edge selection gadget.} We first introduce an \emph{edge selection gadget $G_{i,j}$ for color combination $i,j$} with $i<j$. We start with describing the vertex set of the gadget.
		\begin{itemize}
			\item A set $X_{i,j}$ of vertices $x_1, x_2, \ldots, x_m$.
			\item Vertex sets $U_1, U_2, \ldots, U_m$ with $4n+1$ vertices each, that is, $U_\ell=\{u^\ell_0, u^\ell_1, u^\ell_2,\ldots, 
			u^\ell_{4n}\}$ for all $\ell\in[m]$. 
			\item Two special vertices $v_{i,j}^\star,v_{i,j}^{\star\star}$.
		\end{itemize}
		The gadget has the following edges.
		\begin{itemize}
			\item For all $\ell\in [m]$ we have edge $\{x_\ell,v_{i,j}^\star\}$, $\{v_{i,j}^\star,u^\ell_0\}$, and $\{u^\ell_{4n},v_{i,j}^{\star\star}\}$. 
			
			\item For all $\ell\in [m]$ and $\ell'\in [4n]$, we have edge $\{u^\ell_{\ell'-1},u^\ell_{\ell'}\}$.
		\end{itemize}

		\subparagraph{Verification gadget.} For each color $i$, we introduce the following vertices. What we describe in the following will be used as a \emph{verification gadget for color $i$}.
		\begin{itemize}
			\item We have one vertex $y^i$ and $k+1$ vertices $v^i_\ell$ for $0\le \ell\le k$.
			\item For every $\ell\in[m]$ and every $j\in[k]\setminus\{i\}$ we have $5n$ vertices $a^{i,j,\ell}_1,a^{i,j,\ell}_2,\ldots,a^{i,j,\ell}_{5n}$ and $5n$ vertices $b^{i,j,\ell}_1,b^{i,j,\ell}_2,\ldots,b^{i,j,\ell}_{5n}$.
			\item We have a set $\hat{U}_i$ of $13n+1$ vertices $\hat{u}^i_1,\hat{u}^i_2,\ldots,\hat{u}^i_{13n+1}$.
		\end{itemize}
		We add the following edges. We add edge $\{y^i,v^i_0\}$.
		For every $\ell\in[m]$, every $j\in[k]\setminus\{i\}$, and every $\ell'\in[5n-1]$ we add edge $\{a^{i,j,\ell}_{\ell'},a^{i,j,\ell}_{\ell'+1}\}$ and we add edge $\{b^{i,j,\ell}_{\ell'},b^{i,j,\ell}_{\ell'+1}\}$.
		
		Let $1\le j<i$ (skip if $i=1$), let $e_\ell^{j,i}\in F_{j,i}$, and let $w^i_{\ell'}\in W_i$ be incident with $e_\ell^{j,i}$. Then we add edge $\{v_{j-1}^i,a^{i,j,\ell}_{1}\}$ and we add edge $\{a^{i,j,\ell}_{5n},u^\ell_{\ell'-1}\}$ between $a^{i,j,\ell}_{5n}$ and the vertex $u^\ell_{\ell'-1}$ of the edge selection gadget of color combination $j,i$.
		Furthermore, we add edge $\{v_{j}^i,b^{i,j,\ell}_{1}\}$ and edge $\{b^{i,j,\ell}_{5n},u^\ell_{\ell'}\}$ between $b^{i,j,\ell}_{5n}$ and the vertex $u^\ell_{\ell'}$ of the edge selection gadget of color combination $j,i$.
		
		We add edge $\{v^i_{i-1},\hat{u}^i_1\}$ and for all $\ell''\in[13n]$ we add edge $\{\hat{u}^i_{\ell''},\hat{u}^i_{\ell''+1}\}$. Furthermore, we add edge $\{\hat{u}^i_{13n+1},v^i_i\}$. 
		
		Let $i<j\le k$ (skip if $i=k$), let $e_\ell^{i,j}\in F_{i,j}$, and let $w^i_{\ell'}\in W_i$ be incident with $e_\ell^{i,j}$. Then we add edge $\{v_{j-1}^i,a^{i,j,\ell}_{1}\}$ and edge $\{a^{i,j,\ell}_{5n},u^\ell_{3n+\ell'-1}\}$ between $a^{i,j,\ell}_{5n}$ and the vertex $u^\ell_{3n+\ell'-1}$ of the edge selection gadget of color combination $i,j$.
		Furthermore, we add edge $\{v_{j}^i,b^{i,j,\ell}_{1}\}$ and edge $\{b^{i,j,\ell}_{5n},u^\ell_{3n+\ell'}\}$ between $b^{i,j,\ell}_{5n}$ and the vertex $u^\ell_{3n+\ell'}$ of the edge selection gadget of color combination $i,j$.

		\subparagraph{Connector gadget.} Next, we describe \emph{connector gadgets}. Intuitively, these gadgets will be used to connect many vertex pairs by fast paths, which will make arguing about possible labelings in \textsc{Yes}-instances much easier. Connector gadgets consist of six vertices $\hat{v}_0,\hat{v}_0',\hat{v}_1,\hat{v}_2,\hat{v}_3,\hat{v}_3'$. 
		Each connector gadget is associated with two sets $A,B$ with $B\subseteq A$ containing vertices of other gadgets. 
		Let $V^\star$ denote the set of all vertices from all edge selection gadgets and all verification gadgets.
		The sets $A$ and $B$ will only play a role when defining the matrix $D$ later. Informally speaking, vertices in $A$ should reach all vertices in $V^\star$ quickly through the gadget, except the ones in $B$.
		We have the following edges. 
		\begin{itemize}
			\item Edges $\{\hat{v}_0,\hat{v}_1\},\{\hat{v}_0',\hat{v}_1\},\{\hat{v}_1,\hat{v}_2\},\{\hat{v}_2,\hat{v}_3\},\{\hat{v}_2,\hat{v}_3'\}$.
			\item An edge between $\hat{v}_1$ and each vertex in $V^\star$.
			\item An edge between $\hat{v}_2$ and each vertex in $V^\star$.
		\end{itemize}
		We add two connector gadgets for each edge selection gadget and two connector gadgets for each verification gadget.
		
		The \emph{first connector gadget for the edge selection gadget of color combination $i,j$} with $i<j$ has the following sets.
		\begin{itemize}
			\item Sets $A$ and $B$ contain all vertices in $X_{i,j}$ and vertex $v_{i,j}^{\star\star}$.
		\end{itemize}
		The \emph{second connector gadget for the edge selection gadget of color combination $i,j$} with $i<j$ has the following sets.
		\begin{itemize}
			\item Set $A$ contains all vertices from the edge selection gadget $G_{i,j}$ except vertices in $X_{i,j}$.
			\item Set $B$ is empty.
		\end{itemize}
		The \emph{first connector gadget for the verification gadget of color $i$} has the following sets.
		\begin{itemize}
			\item Sets $A$ and $B$ contain all vertices $v^i_\ell$ with $0\le \ell\le k$.
		\end{itemize}
		The \emph{second connector gadget for the verification gadget of color $i$} has the following sets.
		\begin{itemize}
			\item Set $A$ contains all vertices of the verification gadget except vertices $v^i_\ell$ with $0\le \ell \le k$.
			\item Set $B$ is empty.
		\end{itemize}
		
		\subparagraph{Alignment gadget.} Lastly, we introduce an \emph{alignment gadget}. It consists of one vertex $w^\star$ and a set of vertices $\hat{W}$ containing one vertex for each edge selection gadget, one vertex for each verification gadget, and one vertex for each connector gadget. Vertex $w^\star$ is connected to each vertex in $\hat{W}$.
		The vertex $x_1$ of each edge selection gadget, the vertex $y^i$ of each verification gadget, and the vertex $\hat{v}_1$ of each connector gadget are each connected to  one vertex in $\hat{W}$ such that all vertices in $\hat{W}$ have degree two. Intuitively, this gadget is used to relate labels of different gadgets to each other. 
		
		\subparagraph{Feedback vertex number.} This finished the description of the underlying graph $G$. For an illustration see \cref{fig:hardness1}. We can observe that the vertex set containing
		\begin{itemize}
			\item vertices $v_{i,j}^\star$ and $v_{i,j}^{\star\star}$ of each edge selection gadget,
			\item vertices $v^i_\ell$ with $0\le \ell\le k$ of each verification gadget,
			\item vertices $\hat{v}_1$ and $\hat{v}_2$ of each connector gadget, and
			\item vertex $w^\star$ of the alignment gadget
		\end{itemize}
		forms a feedback vertex set in $G$ with size $O(k^2)$.
		
		\subparagraph{\boldmath Duration matrix $D$.} We proceed with describing the matrix $D$ of durations of fastest paths. For a more convenient presentation, we use the notation $d(v,v'):= D_{v,v'}$. For all vertices $v,v'$ that are neighbors in $G$ we have that $d(v,v')=1$ and $d(v',v)=1$.
		
		Next, consider a connector gadget consisting of vertices $\hat{v}_0,\hat{v}_0',\hat{v}_1,\hat{v}_2,\hat{v}_3,\hat{v}_3'$ and with sets $A$ and $B$. Informally, the connector gadget makes sure that all vertices in $A$ can reach all other vertices (of edge selection gadgets and verification gadgets) except the ones in $B$. We set the following durations. Recall that $V^\star$ denotes the set of all vertices from all edge selection gadgets and all verification gadgets.
		\begin{itemize}
			\item We set $d(\hat{v}_0,\hat{v}_2)=d(\hat{v}_3,\hat{v}_1)=d(\hat{v}_2,\hat{v}_0')=d(\hat{v}_1,\hat{v}_3')=2$, and $d(\hat{v}_0,\hat{v}_0')=d(\hat{v}_3,\hat{v}_3')=d(\hat{v}_0,\hat{v}_3')=d(\hat{v}_3,\hat{v}_0')=3$.
			\item Let $v\in A$, then we set $d(v,\hat{v}_0')=3$ and $d(v,\hat{v}_3')=3$.
			\item Let $v\in V^\star\setminus B$, then we set $d(\hat{v}_0,v)=3$ and $d(\hat{v}_3,v)=3$.
			\item Let $v\in A$ and $v'\in V^\star\setminus B$ such that $v$ and $v'$ are not neighbors, then we set $d(v,v')=3$.
		\end{itemize}
		Now consider two connector gadgets, one with vertices $\hat{v}_0,\hat{v}_0',\hat{v}_1,\hat{v}_2,\hat{v}_3,\hat{v}_3'$ and with sets $A$ and $B$, and one with vertices $\hat{v}_0',\hat{v}_0'',\hat{v}_1',\hat{v}_2',\hat{v}_3',\hat{v}_3''$ and with sets $A'$ and $B'$.
		\begin{itemize}
			\item If there is a vertex $v\in A$ with $v\notin A'$, then we set $d(\hat{v}_1,\hat{v}_1')=3$.
			\item If there is a vertex $v\in A$ with $v\in A'\setminus B'$, then we set $d(\hat{v}_1,\hat{v}_2')=3$.
			\item If there is a vertex $v\in V^\star\setminus (A\setminus B)$ with $v\notin A'$, then we set $d(\hat{v}_2,\hat{v}_1')=3$.    
			\item If there is a vertex $v\in V^\star\setminus (A\setminus B)$ with $v\in A'\setminus B'$, then we set $d(\hat{v}_2,\hat{v}_2')=3$.
		\end{itemize}
		
		Next, consider the edge selection gadget for color combination $i,j$ with $i<j$.
		\begin{itemize}
			\item Let $1\le \ell<\ell'\le m$. We set $d(x_\ell,x_{\ell'})=2n\cdot (i+j)\cdot((\ell')^2-\ell^2)+1$.
			\item For all $\ell\in[m]$ we set $d(x_\ell,v_{i,j}^{\star\star})=8n+5$.
		\end{itemize}
		
		Next, consider the verification gadget for color $i$. 
		%
		For all $0\le j<j'<i$ and all $i\le j<j'\le k$ we set the following.
		\begin{itemize}
			\item We set $d(v^i_j,v^i_{j'})=(20n+6)(j'-j)-1$.
		\end{itemize}
		For all $0\le j<i$ and all $i\le j'\le k$ we set the following.
		\begin{itemize}
			\item We set $d(v^i_j,v^i_{j'})=(20n+6)(j'-j)+6n-1$.
		\end{itemize}
		
		Finally, we consider the alignment gadget. Let $x_1$ belong to the edge selection gadget of color combination $i,j$ and let $w\in \hat{W}$ denote the neighbor of $x_1$ in the alignment gadget. Let $\hat{v}_1$ and $\hat{v}_2$ belong to the first connector gadget of the edge selection gadget for color combination $i,j$. Let $\hat{V}$ contain all vertices $\hat{v}_1$ and $\hat{v}_2$ belonging to the other connector gadgets (different from the first one of the edge selection gadget for color combination $i,j$). 
		\begin{itemize}
			\item We set $d(w^\star,x_1)=(20n+6)(i+j)$.
			\item We set $d(w^\star,\hat{v}_1)=n^9$, $d(w,\hat{v}_2)=n^9$, $d(w,\hat{v}_1)=n^9-(20n+6)(i+j)+1$, and $d(w,\hat{v}_2)=n^9-(20n+6)(i+j)+1$.
			\item For each vertex $v\in (V^\star\cup \hat{V})\setminus (X_{i,j}\cup \{v_{i,j}^{\star\star}\})$ we set $d(w^\star,v)=n^9+2$ and $d(w,v)=n^9-(20n+6)(i+j)+3$.
		\end{itemize}
		
		Let $y^i$ belong to the verification gadget of color $i$ and let $w'\in \hat{W}$ denote the neighbor of $y^i$ in the alignment gadget. Let $\hat{v}_1$ and $\hat{v}_2$ belong to the connector gadget of the verification gadget for color $i$. Let $\hat{V}$ contain all vertices $\hat{v}_1$ and $\hat{v}_2$ belonging to the other connector gadgets (different from the one of the verification gadget for color $i$). Let $V_i$ denote the set of all vertices of the verification gadget of color $i$.
		\begin{itemize}
			\item We set $d(w^\star,y^i)=n^8-1$, $d(w',v^i_0)=2$, and $d(w^\star,v^i_0)=n^8$.
			\item We set $d(w^\star,\hat{v}_1)=n^9$, $d(w^\star,\hat{v}_2)=n^9$, $d(w',\hat{v}_1)=n^9-n^8$, and $d(w',\hat{v}_2)=n^9-n^8$.
			\item For each vertex $v\in (V^\star\cup \hat{V})\setminus V_i$ we set $d(w^\star,v)=n^9+1$, $d(w,v)=n^9-n^8+2$, and $d(y^i,v)=n^9-n^8+2$.
		\end{itemize}
		Let $\hat{v}_1$ belong to some connector gadget. Then we set $d(w^\star,\hat{v}_1)=n^9$.
		
		All fastest path durations between non-adjacent vertex pairs that are not specified above are set to infinity.

		\begin{figure}
			\noindent\makebox[\textwidth]{
				\centering
				\includegraphics{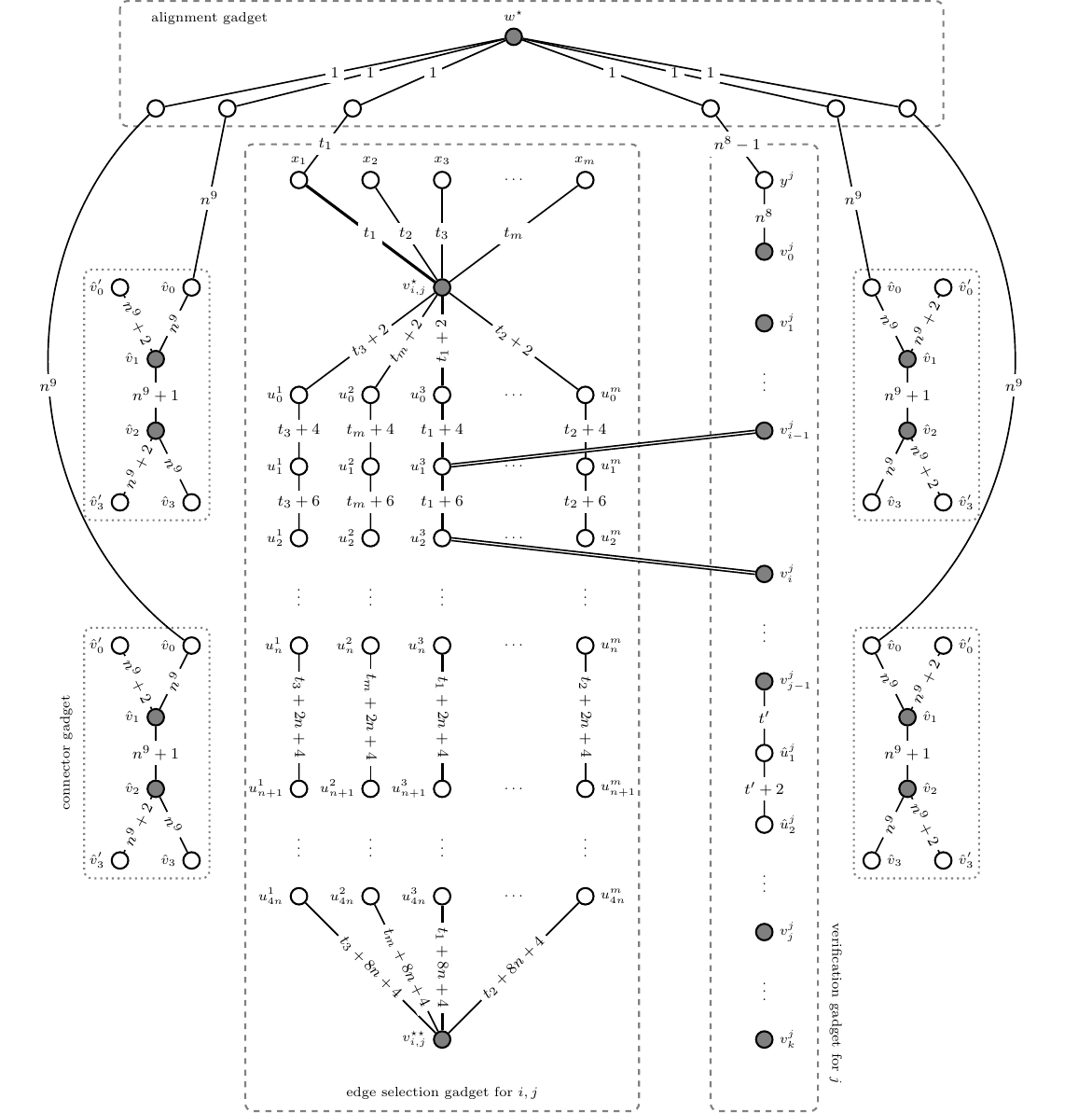}
			}
			\caption{Illustration of part of the underlying graph $G$ and a possible labeling. Edges incident with vertices $\hat{v}_1,\hat{v}_2$ of connector gadgets are omitted. Gray vertices form a feedback vertex set.
				The double line connections, between a vertex $v_{i-1}^j$ in the verification gadget, and $u_1^3$ in the edge selection gadget, 
				and, between a vertex $u_2^3$ in the edge selection gadget, and $v_{i}^j$ in the verification gadget,
				consist of $5n$ vertices $a_1^{j,i,3},a_2^{j,i,3},\dots,a_{5n}^{j,i,3}$ 
				and $b_1^{j,i,3},b_2^{j,i,3},\dots,b_{5n}^{j,i,3}$, respectively.
			}\label{fig:hardness1}
		\end{figure}

		\subparagraph{Correctness.} This finishes the construction of \deltaExactLong\ instance $D$, which can clearly be computed in polynomial time. For an illustration see \cref{fig:hardness1}. As discussed earlier, we have that the vertex cover number of the underlying graph of the instance is in $O(k^2)$.
		
		In the remainder we prove that $D$ is a \textsc{Yes}-instance of \deltaExactLong\ if and only if the $H$ is a \textsc{Yes}-instance of \textsc{Multicolored Clique}.
		
		\subparagraph{$(\Rightarrow)$:} Assume $D$ is a \textsc{Yes}-instance of \deltaExactLong\ and let $(G,\lambda)$ be a solution. We have that the underlying graph $G$ is uniquely defined by $D$. We first prove a number of properties of $\lambda$ that we need to define a set of vertices in $H$ which we claim to be a multicolored clique.
		
		To start, consider the alignment gadget. We can observe that all edges incident with $w^\star$ have the same label.
		\begin{claim}\label{claim:1}
			For all $w\in \hat{W}$ we have that $\lambda(\{w^\star,w\})=t$ for some $t\in\mathbb{N}$.
		\end{claim}
		\begin{claimproof}
			Assume for contradiction that there are $w,w'\in \hat{W}$ such that $\lambda(\{w^\star,w\})=t$ and $\lambda(\{w^\star,w'\})=t'$ with $t\neq t'$. Let w.l.o.g.\ $t<t'$. Then $w$ can reach $w'$, however we have that $d(w,w')=\infty$, a contradiction.
		\end{claimproof}
		\cref{claim:1} allows us to assume w.l.o.g.\ that all edges incident with vertex $w^\star$ of the alignment gadget have label $1$. From now we will assume that this is the case.
		
		Next, we analyse the labelings of connector gadgets. We show that all edges incident with vertices of connector gadgets have labels of at least $n^9$ and at most $n^9+2$. More precisely, we show the following.
		\begin{claim}\label{claim:2}
			Let $\hat{v}_0,\hat{v}_0',\hat{v}_1,\hat{v}_2,\hat{v}_3,\hat{v}_3'$ be the vertices of a connector gadget with sets $A$ and $B$. Then we have that 
			$\lambda(\{\hat{v}_0,\hat{v}_1\})=n^9$,  $\lambda(\{\hat{v}_0',\hat{v}_1\})=n^9+2$,
			$\lambda(\{\hat{v}_1,\hat{v}_2\})=n^9+1$, $\lambda(\{\hat{v}_2,\hat{v}_3\})=n^9$, and $\lambda(\{\hat{v}_2,\hat{v}_3'\})=n^9+2$. 
			Furthermore, for all $v\in V^\star$ we have $n^9\le \lambda(\{\hat{v}_1,v\})\le n^9+2$ and $n^9\le \lambda(\{\hat{v}_2,v\}) \le n^9+2$.
		\end{claim}
		\begin{claimproof}
			Let $w\in \hat{W}$ denote the vertex of the alignment gadget that is neighbor of $w^\star$ and $\hat{v}_0$. We have $d(w^\star,\hat{v}_0)=n^9$. It follows that $\lambda(\{w,\hat{v}_0\})=n^9$. Since $d(\hat{v}_1,w)=\infty$ and $d(w,\hat{v}_1)=\infty$, we have that $\lambda(\{\hat{v}_0,\hat{v}_1\})=n^9$.
			Note that $\hat{v}_1$ is the only common neighbor of $\hat{v}_0$ and $\hat{v}_2$ and the only common neighbor of $\hat{v}_0$ and $\hat{v}_0'$. Since $d(\hat{v}_0,\hat{v}_2)=2$ and $d(\hat{v}_0,\hat{v}_0')=3$ we have that $\lambda(\{\hat{v}_1,\hat{v}_2\})=n^9+1$ and $\lambda(\{\hat{v}_0',\hat{v}_1\})=n^9+2$. Similarly, we have that $\hat{v}_2$ is the only common neighbor of $\hat{v}_3$ and $\hat{v}_1$ and the only common neighbor of $\hat{v}_3$ and $\hat{v}_3'$. Since $d(\hat{v}_3,\hat{v}_1)=2$ and $d(\hat{v}_3,\hat{v}_3')=3$ we have that $\lambda(\{\hat{v}_2,\hat{v}_3\})=n^9$ and $\lambda(\{\hat{v}_2,\hat{v}_3'\})=n^9+2$.

			Let $v\in V^\star$. Note that $d(v,\hat{v}_0)=\infty$ and $d(v,\hat{v}_3)=\infty$. It follows that $\lambda(\{\hat{v}_1,v\})\ge n^9$ and $\lambda(\{\hat{v}_2,v\})\ge n^9$. Otherwise, there would be a temporal path from $v$ to $\hat{v}_0$ via $\hat{v}_1$ or a temporal path from $v$ to $\hat{v}_3$ via $\hat{v}_2$, a contradiction.
			Furthermore, note that $d(\hat{v}_0',v)=\infty$ and $d(\hat{v}_3',v)=\infty$. It follows that $\lambda(\{\hat{v}_1,v\})\le n^9+2$ and $\lambda(\{\hat{v}_2,v\})\le n^9+2$. Otherwise, there would be a temporal path from $\hat{v}_0'$ to $v$ via $\hat{v}_1$ or a temporal path from $\hat{v}_3$ to $v$ via $\hat{v}_2$, a contradiction.
		\end{claimproof}
		
		Now we take a closer look at the edge selection gadgets. We make a number of observations that will allow us to define a set of vertices in $H$ that we claim to be a multicolored clique.
		
		\begin{claim}\label{claim:3}
			For all $1\le i<j\le k$ and $\ell\in[m]$ we have that $\lambda(\{u^\ell_{4n},v_{i,j}^{\star\star}\})\le n^9+2$, where $u^\ell_{4n}$ belongs to the edge selection gadget for $i,j$.
		\end{claim}
		\begin{claimproof}
			Consider the first connector gadget of the edge selection gadget for $i,j$ with vertices $\hat{v}_0,\hat{v}_0',\hat{v}_1,\hat{v}_2,\hat{v}_3,\hat{v}_3'$ and sets $A,B$. Recall that $v_{i,j}^{\star\star}\in B$ and hence we have that $d(\hat{v}_0,v_{i,j}^{\star\star})=\infty$. Furthermore, we have that $u^\ell_{4n}\notin B$ and hence $d(\hat{v}_0,u^\ell_{4n})=3$. By \cref{claim:2} and the fact that $d(w^\star,\hat{v}_0)=n^9$ we have that both edges incident with $\hat{v}_0$ have label $n^9$. It follows that a fastest temporal path from $\hat{v}_0$ to $u^\ell_{4n}$ arrives at $u^\ell_{4n}$ at time $n^9+2$. Now assume for contradiction that $\lambda(\{u^\ell_{4n},v_{i,j}^{\star\star}\})> n^9+2$. Then there exists a temporal walk from $\hat{v}_0$ to $v_{i,j}^{\star\star}$ via $u^\ell_{4n}$, a contradiction to $d(\hat{v}_0,v_{i,j}^{\star\star})=\infty$.
		\end{claimproof}
		
		\begin{claim}\label{claim:4}
			For all $1\le i<j\le k$ and $\ell\in[m]$ we have that $\lambda(\{x_\ell,v_{i,j}^{\star}\})=(i+j)\cdot (2n\ell^2 +18n+6)$, where $x_\ell$ belongs to the edge selection gadget for $i,j$.
		\end{claim}
		\begin{claimproof}
			We first determine the label of $\{x_1,v_{i,j}^{\star}\}$, where $x_1$ belongs to the edge selection gadget for $i,j$. Note that $x_1$ is connected to the alignment gadget. Let $w\in \hat{W}$ be the vertex of the alignment gadget that is a neighbor of $x_1$. Since $d(w^\star,x_1)=(20n+6)(i+j)$ we have that $\lambda(\{w,x_1\})=(20n+6)(i+j)$. 
			
			First, assume that $\lambda(\{x_1,v_{i,j}^{\star}\})<(20n+6)(i+j)$. Then there is a temporal path from $v_{i,j}^\star$ to $w$ via $x_1$. However, we have that $d(x_{i,j^\star},w)=\infty$, a contradiction.
			Next, assume that $(20n+6)(i+j)<\lambda(\{x_1,v_{i,j}^{\star}\})<n^9+2$. Then there is a temporal path from $w$ to $v_{i,j}$ via $x_1$ with duration strictly less than $n^9-(20n+6)(i+j)+3$. However, we have that $d(w,v_{i,j}^\star)=n^9-(20n+6)(i+j)+3$, a contradiction.
			Finally, assume that $\lambda(\{x_1,v_{i,j}^{\star}\})\ge n^9+2$. Consider a fastest temporal path from $x_1$ to $v_{i,j}^{\star\star}$. This temporal path cannot visit $w$ as its first vertex, since from there it cannot continue. 
			From this assumption and \cref{claim:2} it follows, that the first edge of the temporal path has a label with value at least $n^9$. However, by \cref{claim:2,claim:3} we have that all edges incident with $v_{i,j}^{\star\star}$ have a label with value at most $n^9+2$. It follows that $d(x_1,v_{i,j}^{\star\star})\le 3$, a contradiction.
			
			We can conclude that $\lambda(\{x_1,v_{i,j}^{\star}\})=(20n+6)(i+j)$. Now let $1<\ell\le m$. We have that $d(x_1,x_\ell)=2n\cdot (i+j)\cdot(\ell^2-1)+1$ which implies that $\lambda(\{x_\ell,v_{i,j}^{\star}\})\ge (i+j)\cdot (2n\ell^2 +18n+6)$. Assume that $(i+j)\cdot (2n\ell^2 +18n+6) <\lambda(\{x_\ell,v_{i,j}^{\star}\})\le n^9+2$. Then the temporal path from $x_1$ to $x_\ell$ via $v_{i,j}^\star$ is not a fastest temporal path from $x_1$ to $x_\ell$. Again, we have that a fastest temporal path from $x_1$ to $x_\ell$ cannot visit $w$ as its first vertex, since from there it cannot continue. By \cref{claim:2}, all other edges incident with $x_1$ (that is, all different from the one to $v_{i,j}^\star$ and the one to $w$) have a label of at least $n^9$ and at most $n^9+2$. Similarly, by \cref{claim:2} we have that all other edges incident with $x_\ell$ (that is, all different from the one to $v_{i,j}^\star$) have a label of at least $n^9$ and at most $n^9+2$. It follows that any temporal path from $x_1$ to $x_\ell$ that visits $v_{i,j}^\star$ as its first vertex has a duration strictly larger than $2n\cdot (i+j)\cdot(\ell^2-1)+1$. Any temporal path from $x_1$ to $x_\ell$ that visits a vertex different from $v_{i,j}^\star$ as its first vertex has duration of at most $3$. In both cases we have a contradiction. 
			Lastly, assume that $\lambda(\{x_\ell,v_{i,j}^{\star}\})> n^9+2$. Consider a fastest temporal path from $x_\ell$ to $v_{i,j}^{\star\star}$. Now this temporal path has duration at most 3 since by \cref{claim:2} and the just made assumption all edges incident with $x_\ell$ have label at least $n^9$ whereas by \cref{claim:2,claim:3} all edges incident with $v_{i,j}^{\star\star}$ have label at most $n^9+2$, a contradiction. 
		\end{claimproof}
		
		\begin{claim}\label{claim:5}
			For all $1\le i<j\le k$ there exist a permutation $\sigma_{i,j}:[m]\rightarrow [m]$ such that for all $\ell\in[m]$ we have that $\lambda(\{u^\ell_{4n},v_{i,j}^{\star\star}\})=(i+j)\cdot (2n\cdot (\sigma_{i,j}(\ell))^2 +18n+6)+8n+4$, where $u^\ell_{4n}$ belongs to the edge selection gadget for $i,j$.
			
			Furthermore, a fastest temporal path from $x_\ell$ (of the edge selection gadget for $i,j$) to $v_{i,j}^{\star\star}$ visits $v_{i,j}^\star$ as its second vertex, and $u^{\ell'}_{4n}$ with $\sigma_{i,j}(\ell')=\ell$ (of the edge selection gadget for $i,j$) as its second last vertex.
		\end{claim}
		\begin{claimproof}
			For every $\ell\in[m]$ we have that $d(x_\ell,v_{i,j}^{\star\star})=8n+5$, where $x_\ell$ belongs to the edge selection gadget for $i,j$. From \cref{claim:2,claim:4} follows that all edges incident with $x_\ell$ have a label of at least $n^9$ except the one to $v_{i,j}^\star$ and, if $\ell=1$, the edge connecting $x_1$ to the alignment gadget. In the latter case, no temporal path from $x_1$ from $v_{i,j}^{\star\star}$ can continue to the neighbor of $x_1$ in the alignment gadget, since it cannot continue from there.
			
			Now consider $v_{i,j}^{\star\star}$. By \cref{claim:2,claim:3} we have that all edges incident with $v_{i,j}^{\star\star}$ have a label of at most $n^9+2$. It follows that a fastest temporal path $P$ from $x_\ell$ to $v_{i,j}^{\star\star}$ has to visit $v_{i,j}^\star$ after $x_\ell$, since otherwise we have $d(x_\ell,v_{i,j}^{\star\star})\le 2$, a contradiction.
			
			Furthermore, we have by \cref{claim:2} that all edges incident with $v_{i,j}^{\star\star}$ have a label of at least $n^9$ except the ones incident to $u_{\ell'}^{2n}$ for $\ell'\in[m]$. By \cref{claim:4} we have that $\lambda(\{x_\ell,v_{i,j}^{\star}\})\le 4n^6$.
			It follows that a fastest temporal path from $x_\ell$ to $v_{i,j}^{\star\star}$ has to visit $u^{\ell'}_{4n}$ for some $\ell'\in[m]$ as its second last vertex. Otherwise, we have $d(x_\ell,v_{i,j}^{\star\star})> 8n+5$ (for sufficiently large $n$), a contradiction.
			
			We can conclude that a fastest temporal path from $x_\ell$ to $v_{i,j}^{\star\star}$ has to visit $v_{i,j}^\star$ as its second vertex and $u^{\ell'}_{4n}$ for some $\ell'\in[m]$ as its second last vertex. Recall that in a  temporal path, the difference between the labels of the first and last edge determine its duration (minus one). Hence, we have that $\lambda(\{u^{\ell'}_{4n},v_{i,j}^{\star\star}\})-\lambda(\{x_\ell,v_{i,j}^{\star}\})+1=8n+5$.
			By \cref{claim:4} we have that $\lambda(\{x_\ell,v_{i,j}^{\star}\})=(i+j)\cdot (2n\ell^2+18n +2)$. It follows that $\lambda(\{u^{\ell'}_{4n},v_{i,j}^{\star\star}\})=(i+j)\cdot (2n\ell^2 +18n+6)+8n+4$. We set $\sigma_{i,j}(\ell')=\ell$.
			
			Finally, we show that $\sigma_{i,j}$ is a permutation on $[m]$. Assume for contradiction that there are $\ell,\ell'\in[m]$ with $\ell\neq \ell'$ such that $\sigma_{i,j}(\ell)=\sigma_{i,j}(\ell')$. Then we have that $\lambda(\{u^{\ell}_{4n},v_{i,j}^{\star\star}\})=\lambda(\{u^{\ell'}_{4n},v_{i,j}^{\star\star}\})$. However, by \cref{claim:4} we have that all edges from $v_{i,j}^\star$ to a vertex in $X_{i,j}$ have distinct labels. Furthermore, we argued above that every fastest path from a vertex in $X_{i,j}$ to $v_{i,j}^{\star\star}$ visits $v_{i,j}^\star$ as its second vertex and a vertex from the set $\{u^{\ell''}_{4n} :  \ell''\in[m]\}$ as its second last vertex. Since for all $x_{\ell''}$ with $\ell''\in[m]$ we have that $d(x_{\ell''},v_{i,j}^{\star\star})=8n+5$, we must have that all edges from vertices in $\{u^{\ell''}_{4n} :  \ell''\in[m]\}$ to $v_{i,j}^{\star\star}$ must have distinct labels. Hence, we have a contradiction and can conclude that $\sigma_{i,j}$ is indeed a permutation.
		\end{claimproof}
		
		For all $1\le i<j\le k$, let $\sigma_{i,j}$ be the permutation on $[m]$ as defined in \cref{claim:5}. We call $\sigma_{i,j}$ the \emph{permutation of color combination $i,j$}. 
		Now we have enough information to define a set of vertices of $H$ that form a multicolored clique.
		To this end, consider the following set $X$ of edges from $H$.
		\[
		X=\{e_\ell^{i,j}\in F_{i,j} :  \sigma_{i,j}(\ell)=1\}
		\]
		We claim that $\bigcup_{e\in X}e$ forms a multicolored clique in $H$.
		From now on, denote $\{e_{i,j}\}=X\cap F_{i,j}$. We show that for all $i\in [k]$ we have that $|(\bigcap_{1\le j<i} e_{j,i}) \cap (\bigcap_{i<j\le k} e_{i,j})|=1$, that is, for every color~$i$, all edges of a color combination involving $i$ have the same vertex of color $i$ as endpoint. This implies that $\bigcup_{e\in X}e$ is a multicolored clique in $H$.
		
		Before we proceed, we show some further properties of $\lambda$. First, let us focus on the labels on edges of the edge selection gadgets.
		
		\begin{claim}\label{claim:6}
			For all $1\le i<j\le k$,  $\ell\in[m]$, and $\ell'\in[4n]$ we have that $\lambda(\{u^\ell_{\ell'-1},u^\ell_{\ell'}\})=(i+j)\cdot (2n\cdot (\sigma_{i,j}(\ell))^2 +18n+6)+2\ell'+2$, where $u^\ell_{\ell'-1}$ and $u^\ell_{\ell'}$ belong to the edge selection gadget for $i,j$ and $\sigma_{i,j}$ is the permutation of color combination $i,j$.
		\end{claim}
		\begin{claimproof}
			Let $1\le i<j\le k$ and $\ell\in[m]$.
			By \cref{claim:5} we know that a fastest temporal path from $x_{\sigma_{i,j}(\ell)}$ (of the edge selection gadget for $i,j$) to $v_{i,j}^{\star\star}$ visits $v_{i,j}^\star$ as its second vertex, and $u^{\ell}_{4n}$ (of the edge selection gadget for $i,j$) as its second last vertex. 
			Furthermore, by \cref{claim:4} we have that $\lambda(\{x_{\sigma_{i,j}(\ell)},v_{i,j}^{\star}\})=(i+j)\cdot (2n\cdot {(\sigma_{i,j}(\ell))^2}+18n +2)$ and by \cref{claim:5} we have that $\lambda(\{u^\ell_{4n},v_{i,j}^{\star\star}\})=(i+j)\cdot (2n\cdot (\sigma_{i,j}(\ell))^2+18n +2)+8n+4$. It follows that there exist a temporal path $P$ from $v_{i,j}^\star$ to $u^\ell_{4n}$ that starts at $v_{i,j}^\star$ later than $(i+j)\cdot (2n\cdot {(\sigma_{i,j}(\ell))^2} +18n+6)$ and arrives at $u^\ell_{4n}$ earlier than $(i+j)\cdot (2n\cdot (\sigma_{i,j}(\ell))^2 +18n+6)+8n+4$. Hence, the temporal path $P$ has duration at most $8n+3$.
			
			We investigate the temporal path $P$ from its destination $u^\ell_{4n}$ back to its start vertex $v_{i,j}^\star$. Consider the neighbors of $u^\ell_{4n}$ that are different from $v_{i,j}^{\star\star}$. By \cref{claim:2} we have that all edges from $u^\ell_{4n}$ to neighbors of $u^\ell_{4n}$ that are vertices of connector gadgets have a label of at least $n^9$. Hence, $P$ does not visit any of those neighbors. Next, consider neighbors of $u^\ell_{4n}$ in verification gadgets. Assume $u^\ell_{4n}$ has a neighbor in the verification gadget of color $i'$ for some $i'\in[k]$. Then this neighbor is vertex $b^{i',j,\ell}_{5n}$. Note that if $P$ visits $b^{i',j,\ell}_{5n}$, then it also visits all of $\{b^{i',j,\ell}_{\ell'} :  \ell'\in[5n]\}$, since all these vertices have degree two. 
			Now consider the second connector gadget of a verification gadget $i'$ with sets $A,B$, we have that all vertices $\{b^{i',j,\ell}_{\ell'} :  \ell'\in[5n]\}$ are contained in $A$ and are not contained in $B$. Hence, we have that all non-adjacent pairs of vertices in $\{b^{i',j,\ell}_{\ell'} :  \ell'\in[5n]\}$ are on duration $3$ apart, according to $D$, and that $|\lambda(\{b^{i',j,\ell}_{\ell'},b^{i',j,\ell}_{\ell'+1}\})-\lambda(\{b^{i',j,\ell}_{\ell'+1},b^{i',j,\ell}_{\ell'+2}\})|\ge 2$ for all $\ell'\in[5n-2]$.  
			It follows that $P$ would have a duration larger than $8n+3$. We can conclude that $P$ does not visit $b^{i',j,\ell}_{5n}$. It follows that $P$ visits $u^\ell_{4n-1}$.
			Here, we can make an analogous investigation. Additionally, we have to consider the case that $P$ visits a neighbor of $u^\ell_{4n-1}$ in verification gadget of color $i'$ for some $i'\in[k]$ that is vertex $a^{i',j,\ell}_{5n}$. However, we can exclude this by a similar argument as above. 
			
			By repeating the above arguments, we can conclude that $P$ visits (exactly) all vertices in $\{u^\ell_{\ell'} :  0\le \ell'\le 4n\}$ and $v_{i,j}^\star$.
			Consider the second connector gadget of the edge selection gadget of $i,j$ with set $A$ and $B$. Note that all vertices visited by $P$ are contained in $A\setminus B$. It follows that all pairs of non-adjacent vertices visited by $P$ are on duration $3$ apart, according to $D$. In particular, we have $d(u^\ell_{\ell'-1},u^\ell_{\ell'+1})=3$ for all $\ell'\in[4n-1]$ and $d(v_{i,j}^\star,u^\ell_{1})=3$. If follows that for every $\ell'\in[4n-1]$ we have that $\lambda(\{u^\ell_{\ell'},u^\ell_{\ell'+1}\})-\lambda(\{u^\ell_{\ell'-1},u^\ell_{\ell'}\})\ge 2$ and $\lambda(\{u^\ell_{1},u^\ell_{2}\})-\lambda(\{v_{i,j}^\star,u^\ell_{1}\})\ge 2$. 
			
			By investigating the sets $A,B$ of the first connector gadget of the edge selection gadget of $i,j$, we get that $d(x_{\sigma_{i,j}(\ell)},u^\ell_{1})=3$ and hence $\lambda(\{v_{i,j}^\star,u^\ell_{1}\})-\lambda(\{x_{\sigma_{i,j}(\ell)},v_{i,j}^\star\})\ge 2$. Furthermore, we get that $d(u^\ell_{4n-1},v_{i,j}^{\star\star})=3$ and hence $\lambda(\{v_{i,j}^{\star\star},u^\ell_{4n}\})-\lambda(\{u^\ell_{4n-1},u^\ell_{4n}\})\ge 2$. Considering that $P$ visits $4n+2$ vertices, we have that all mentioned inequalities of differences of labels have to be equalities, otherwise $P$ has a duration larger than $8n+3$ or we have that $\lambda(\{v_{i,j}^\star,u^\ell_{1}\})-\lambda(\{x_{\sigma_{i,j}(\ell)},v_{i,j}^\star\})< 2$ or $\lambda(\{v_{i,j}^{\star\star},u^\ell_{4n}\})-\lambda(\{u^\ell_{4n-1},u^\ell_{4n}\})< 2$. Since by \cref{claim:4,claim:5} the labels $\lambda(\{x_{\sigma_{i,j}(\ell)},v_{i,j}^\star\})$ and $\lambda(\{v_{i,j}^{\star\star},u^\ell_{4n}\})$ are determined, then also all labels of edges traversed by $P$ are determined and the claim follows.
		\end{claimproof}
		
		Next, we investigate the labels of the verification gadgets.
		
		\begin{claim}\label{claim:7}
			For all $i\in[k]$ we have that $\lambda(\{y^i,v_0^i\})=n^8$.
		\end{claim}
		\begin{claimproof}
			Let $w\in \hat{W}$ denote the neighbor of $y^i$ in the alignment gadget. Note that we have $d(w^\star,y^i)=n^8-1$. It follows that $\lambda(\{w,y^i\})=n^8-1$. Furthermore, we have that $d(w,v_0^i)=2$ and note that $y^i$ has degree 2. It follows that $\lambda(\{y^i,v_0^i\})=n^8$.
		\end{claimproof}
		
		\begin{claim}\label{claim:8}
			For all $1<i\le k$ and all $\ell\in[m]$ we have that $\lambda(\{v_0^i,a^{i,1,\ell}_{1}\})\le n^8$ or $\lambda(\{v_0^i,a^{i,1,\ell}_{1}\})\ge n^9+2$.
			For $i=1$ we have that $\lambda(\{v_0^i,\hat{u}^i_1\})\le n^8$ or $\lambda(\{v_0^i,\hat{u}^i_1\})\ge n^9+2$.
		\end{claim}
		\begin{claimproof}
			Let $1<i\le k$ and $\ell\in[m]$. Assume that $n^8<\lambda(\{v_0^i,a^{i,1,\ell}_{1}\})< n^9+2$. Then, since by \cref{claim:7} we have $\lambda(\{y^i,v_0^i\})=n^8$, there is a temporal path from $w^\star$ to $a^{i,1,\ell}_{1}$ via $v_0^i$ that arrives at $a^{i,1,\ell}_{1}$ strictly earlier than $n^9+2$. However, we have $d(w^\star,a^{i,1,\ell}_{1})=n^9+2$, a contradiction. The argument for case where $i=1$ is analogous.
		\end{claimproof}
		
		\begin{claim}\label{claim:9}
			For all $1\le i< k$ and all $\ell\in[m]$ we have that $\lambda(\{v_k^i,b^{i,k,\ell}_{1}\})\le n^9+2$.
			For $i=k$ we have that $\lambda(\{v_{k}^i,\hat{u}^i_{13n+1}\})\le n^9+2$.
			%
		\end{claim}
		\begin{claimproof}
			Let $1\le i< k$ and $\ell\in[m]$.
			Consider the first connector gadget of verification gadget for color $i$ with vertices $\hat{v}_0,\hat{v}_0',\hat{v}_1,\hat{v}_2,\hat{v}_3,\hat{v}_3'$ and sets $A,B$. Recall that $v_k^i\in B$ and hence we have that $d(\hat{v}_0,v_k^i)=\infty$. Furthermore, we have that $b^{i,k,\ell}_{1}\notin B$ and hence $d(\hat{v}_0,b^{i,k,\ell}_{1})=3$. By \cref{claim:2} and the fact that $d(w^\star,\hat{v}_0)=n^9$ we have that both edges incident with $\hat{v}_0$ have label $n^9$. It follows that a fastest temporal path from $\hat{v}_0$ to $b^{i,k,\ell}_{1}$ arrives at $b^{i,k,\ell}_{1}$ at time $n^9+2$. Now assume for contradiction that $\lambda(\{v_k^i,b^{i,k,\ell}_{1}\})> n^9+2$. Then there exists a temporal walk from $\hat{v}_0$ to $v_k^i$ via $b^{i,k,\ell}_{1}$, a contradiction to $d(\hat{v}_0,v_k^i)=\infty$.
			The argument for case where $i=k$ is analogous.
		\end{claimproof}
		
		Now we are ready to prove for all $i\in [k]$ that $|(\bigcap_{1\le j<i} e_{j,i}) \cap (\bigcap_{i<j\le k} e_{i,j})|=1$. Assume for contradiction that for some color $i\in[k]$ we have that $|(\bigcap_{1\le j<i} e_{j,i}) \cap (\bigcap_{i<j\le k} e_{i,j})|\neq 1$. 
		Consider the verification gadget for color $i$. Recall that $d(v_0^i,v_k^i)=k(20n+6)+6n-1$. Let $P$ be a fastest temporal path from $v_0^i$ to $v_k^i$.
		We first argue that $P$ cannot visit any vertex of a connector gadget or the alignment gadget. 
		
		\begin{claim}\label{claim:10}
			Let $i\in[k]$. Let $P$ be a fastest temporal path from $v_0^i$ to $v_k^i$. Then $P$ does not visit any vertex of a connector gadget.
		\end{claim}
		\begin{claimproof}
			Assume for contradiction that $P$ visits a vertex of a connector gadget. Then by \cref{claim:2} we have that the arrival time of $P$ is at least $n^9$. By \cref{claim:2} and \cref{claim:9} we have that the arrival time of $P$ is at most $n^9+2$. This means that the second vertex visited by $P$ cannot be a vertex from a connector gadget, because by \cref{claim:2} this would imply $d(v_0^i,v_k^i)\le 2$. Now we can deduce with \cref{claim:8} that $P$ must have a starting time of at most $n^8$. It follows that the arrival time of $P$ must be smaller than $n^9$, a contradiction to the assumption that $P$ visits a vertex of a connector gadget. 
		\end{claimproof}
		
		\begin{claim}\label{claim:11}
			Let $i\in[k]$. Let $P$ be a fastest temporal path from $v_0^i$ to $v_k^i$. Then $P$ does not visit any vertex of the alignment gadget.
		\end{claim}
		\begin{claimproof}
			Note that $P$ starts outside the alignment gadget. This means that if $P$ visits a vertex of the alignment gadget, then the first vertex of the alignment gadget visited by $P$ is a neighbor of $w^\star$. However, these vertices have degree two and the edge to $w^\star$ has label one. It follows that $P$ cannot continue from the vertex of the alignment gadget, a contradiction.
		\end{claimproof}
		
		It follows that the second vertex visited by $P$ is a vertex $a^{i,1,\ell}_{1}$ for some $\ell\in[m]$ or vertex $\hat{u}_1^i$ if $i=1$. In the former case, $P$ has to follow the path segment consisting of vertices in $\{a^{i,1,\ell}_{\ell'} :  \ell'\in[5n]\}$ until it reaches the edge selection gadget of color combination $1,i$. From there it can reach vertex $v_1^i$ by traversing some path segment consisting of vertices $\{b^{i,1,\ell'}_{\ell''} :  \ell''\in[5n]\}$ for some $\ell'\in[m]$. Alternatively, it can reach vertex $v^1_{i-1}$ or $v^1_{i}$ by traversing some path segment consisting of vertices $\{a^{1,i,\ell'}_{\ell''} :  \ell''\in[5n]\}$ for some $\ell'\in[m]$ or $\{b^{1,i,\ell'}_{\ell''} :  \ell''\in[5n]\}$ for some $\ell'\in[m]$, respectively. In the latter case ($i=1$), the temporal path $P$ has to follow the path segment consisting of vertices in $\{\hat{u}^i_\ell :  \ell\in[13n+1]\}$ until it reaches $v^i_1$. More generally, we can make the following observation.
		\begin{claim}\label{claim:12}
			Let $i,j\in\{0,1,\ldots,k\}$. Let $P$ be a temporal path starting at $v^i_j$ and visiting at most $13n+1$ vertices and no vertex of a connector gadget or the alignment gadget. Then $P$ cannot visit vertices in $\{v^{i'}_{j'} :  i',j'\in \{0,1,\ldots,k\}\}\setminus\{v^i_{j-1},v^i_j,v^i_{j+1},v^j_{i-1},v^j_i,v^{j+1}_{i-1},v^{j+1}_i\}$.
		\end{claim}
		\begin{claimproof}
			Consider the edge selection gadget of color combination $i',j'$ for some $i',j'\in[k]$ and let $u^\ell_{\ell'}$ be a vertex of that gadget. Disregarding connections via connector gadgets and the alignment gadget, we have that $u^\ell_{\ell'}$ is (potentially) connected to the verification gadget for color $i'$ and the verification gadget of color $j'$.
			More specifically, by construction of $G$, we have that $u^\ell_{\ell'}$ is potentially connected to 
			\begin{itemize}
				\item vertex $v^{i'}_{j'-1}$ by a path along vertices $\{a^{i',j',\ell}_{\ell''} :  \ell''\in[5n]\}$,
				\item vertex $v^{i'}_{j'}$ by a path along vertices $\{b^{i',j',\ell}_{\ell''} :  \ell''\in[5n]\}$, 
				\item vertex $v^{j'}_{i'-1}$ by a path along vertices $\{a^{j',i',\ell}_{\ell''} :  \ell''\in[5n]\}$, and
				\item vertex $v^{j'}_{i'}$ by a path along vertices $\{b^{j',i',\ell}_{\ell''} :  \ell''\in[5n]\}$.
			\end{itemize}
			Furthermore, by construction of $G$, we have that the duration of a fastest path from $u^\ell_{\ell'}$ to any $v^{i''}_{j''}$ with $i'',j''\in\{0,1,\ldots,k\}$ not mentioned above is at least $10n$ (disregarding edges incident with connector gadgets or the alignment gadget).
			
			Now consider $v^i_j$ and assume $i<j$ ($i>j$). This vertex is (if $j\neq i-1$ and $j\neq k$) connected to some vertex $u^\ell_{\ell'}$ in the edge selection gadget for color combination $i,j+1$ ($j+1,i$) via a path along vertices $\{a^{i,j,\ell}_{\ell''} :  \ell''\in[5n]\}$. Furthermore, $v^i_j$ is (if $j\neq 0$ and $j\neq i$) connected to some vertex $u^{\ell'}_{\ell'''}$ in the edge selection gadget for color combination $i,j$ ($j,i$) via a path along vertices $\{b^{i,j,\ell'}_{\ell''} :  \ell''\in[5n]\}$.
			
			We can conclude that $v^i_j$ can reach a vertex $u^\ell_{\ell'}$ of the edge selection gadget for $i,j+1$ (or $j+1,i$) and a vertex $u^{\ell'}_{\ell'''}$ of the edge selection gadget for color combination $i,j$ (or $j,i$), each along paths of length at least $5n$. From $u^\ell_{\ell'}$ and $u^{\ell'}_{\ell'''}$ we have that any other vertex of the edge selection gadget for $i,j+1$ (or $j+1,i$) and the edge selection gadget for color combination $i,j$ (or $j,i$), respectively, can be reached by a path of length at most $3n$. Together with the observation made in the beginning of the proof, we can conclude that $v^i_j$ can potentially reach any vertex in $\{v^i_{j-1},v^i_j,v^i_{j+1},v^j_{i-1},v^j_i,v^{j+1}_{i-1},v^{j+1}_i\}$ by a path that visits at most $13n+1$ vertices. 
			
			Lastly, consider the case that $j=i-1$ or $j=i$. Then we have that $v^i_{i-1}$ and $v^i_i$ are connected via a path inside the verification gadget for color $i$, visiting the $13n+1$ vertices in $\{\hat{u}^i_\ell :  \ell\in[13n+1]\}$. The claim follows.
		\end{claimproof}
		Furthermore, we can make the following observation on the duration of the temporal paths characterized in \cref{claim:12}.
		\begin{claim}\label{claim:13}
			Let $i,j\in\{0,1,\ldots,k\}$. Let $P$ be a temporal path from $v^i_j$ to a vertex in $\{v^i_{j-1},v^i_j,v^i_{j+1},v^j_{i-1},v^j_i,v^{j+1}_{i-1},v^{j+1}_i\}$ and visiting no vertex of a connector gadget or the alignment gadget. Then $P$ has duration at least $20n$.
		\end{claim}
		\begin{claimproof}
			As argued in the proof of \cref{claim:12}, a temporal path $P$ from $v^i_j$ to a vertex in $\{v^i_{j-1},v^i_j,v^i_{j+1},v^j_{i-1},v^j_i,v^{j+1}_{i-1},v^{j+1}_i\}$ has to either traverse two segments of $5n$ vertices in $\{a^{i',j',\ell}_{\ell'} :  \ell'\in[5n]\}$ or $\{b^{i',j',\ell}_{\ell'} :  \ell'\in[5n]\}$ for some $\ell\in[m]$ and $i',j'\in\{i-1,i,j,j+1\}$ or a segment of the $13n+1$ vertices in $\{\hat{u}^i_\ell :  \ell\in[13n+1]\}$. We analyse the former case first.
			
			Consider the second connector gadget of a verification gadget $i'$ with sets $A,B$, we have that all vertices $\{a^{i',j',\ell}_{\ell'} :  \ell'\in[5n],j'\in[k]\setminus\{i'\}\}\cup\{b^{i',j',\ell}_{\ell'} :  \ell'\in[5n],j'\in[k]\setminus\{i'\}\}$ are contained in $A$ and are not contained in $B$. It follows that all non-adjacent pairs of vertices in $\{a^{i',j',\ell}_{\ell'} :  \ell'\in[5n],j'\in[k]\setminus\{i'\}\}\cup\{b^{i',j',\ell}_{\ell'} :  \ell'\in[5n],j'\in[k]\setminus\{i'\}\}$ are on duration $3$ apart, according to $D$. It follows that $|\lambda(\{a^{i',j',\ell}_{\ell'},a^{i',j',\ell}_{\ell'+1}\})-\lambda(\{a^{i',j',\ell}_{\ell'+1},a^{i',j',\ell}_{\ell'+2}\})|\ge 2$ for all $\ell'\in[5n-2]$ and $j'\in[k]\setminus\{i'\}$. Analogously, we have that $|\lambda(\{b^{i',j',\ell}_{\ell'},b^{i',j',\ell}_{\ell'+1}\})-\lambda(\{b^{i',j',\ell}_{\ell'+1},b^{i',j',\ell}_{\ell'+2}\})|\ge 2$ for all $\ell'\in[5n-2]$ and $j'\in[k]\setminus\{i'\}$. It follows that two segments of $5n$ vertices in $\{a^{i',j',\ell}_{\ell'} :  \ell'\in[5n]\}$ or $\{b^{i',j',\ell}_{\ell'} :  \ell'\in[5n]\}$ for some $\ell\in[m]$ and $i',j'\in\{i-1,i,j,j+1\}$ traversed by $P$ both have duration $10n$ and hence $P$ has duration at least $20n$.
			
			In the latter case, where $P$ traverses a segment of the $13n+1$ vertices in $\{\hat{u}^i_\ell :  \ell\in[13n+1]\}$, we can make an analogous argument, since all vertices in $\{\hat{u}^i_\ell :  \ell\in[13n+1]\}$ are contained in the set $A$ of the second connector gadget of the verification gadget of color $i$ but not in the set $B$ of that connector gadget.
		\end{claimproof}
		
		Recall that $P$ denotes a fastest temporal path from $v_0^i$ to $v_k^i$ and that $d(v_0^i,v_k^i)=k(20n+6)+6n-1$. By \cref{claim:10,claim:11,claim:12} he have that $P$ needs to visit at least one vertex in $\{v^{i'}_{j'} :  i',j'\in \{0,1,\ldots,k\}\}\setminus\{v^i_0,v^i_k\}$. Next, we analyse which vertices in this set are visited by $P$.
		
		\begin{claim}\label{claim:14}
			Let $i\in[k]$. Let $P$ be a fastest temporal path from $v_0^i$ to $v_k^i$. Then $P$ visits all vertices in $\{v^j_i :  0\le j\le k\}$ and no vertex in $\{v^{i'}_{j'} :  i',j'\in \{0,1,\ldots,k\}\}\setminus\{v^j_i :  0\le j\le k\}$.
			Furthermore, $P$ visits the vertices in order $v^i_0, v^i_1, v^i_2, \ldots, v^i_{k-1},v^i_k$.
		\end{claim}
		\begin{claimproof}
			Let $X\subseteq \{v^{i'}_{j'} :  i',j'\in \{0,1,\ldots,k\}\}$ denote the set of vertices in $\{v^{i'}_{j'} :  i',j'\in \{0,1,\ldots,k\}\}$ that are visited by $P$.
			By \cref{claim:12,claim:13} we have that $|X|\le k+1$, since otherwise the duration of $P$ would be at least $20n(k+1)>k(20n+6)+6n-1$, a contradiction.
			
			To prove the claim, we use the notion of a \emph{potential $p^i$ with respect to $i$} of a vertex $v^{i'}_j$.
			We say that the first potential of vertex $v^{i'}_j$ with respect to $i$ is $p^i(v^{i'}_j)=i'+j-i$.
			The temporal path $P$ starts at vertex $v^i_0$ with $p^i(v^i_0)=0$, and ends at vertex $v^i_k$ with $p^i(v^i_k)=k$. 
			
			Assume the path $P$ is at some vertex $v^{i'}_j$ with $p^i_1(v^{i'}_j)=i'+j-i$.
			By \cref{claim:12} we have that the next vertex in  $\{v^{i'}_{j'} :  i',j'\in \{0,1,\ldots,k\}\}$ visited by $P$ is some $v^{i''}_{j'}\in \{v^{i'}_{j-1},v^{i'}_j,v^{i'}_{j+1},v^j_{i'-1},v^j_i,v^{j+1}_{i'-1},v^{j+1}_{i'}\}$. We can observe that $|p^i(v^{i'}_j)-p^i(v^{i''}_{j'})|\le 1$, that is, the first potential changes at most by one when $P$ goes from one vertex in $\{v^{i'}_{j'} :  i',j'\in \{0,1,\ldots,k\}\}$ to the next one.
			Since $|X|\le k+1$ we and $p^i(v^i_k)-p^i(v^i_0)=k$ have that the  potential has to increase by exactly one every time $P$ goes from one vertex in $\{v^{i'}_{j'} :  i',j'\in \{0,1,\ldots,k\}\}$ to the next one. 
			We can conclude that $|X|=k+1$.
			Furthermore, we have that if the path $P$ is at some vertex $v^{i'}_j$, the next vertex in $\{v^{i'}_{j'} :  i',j'\in \{0,1,\ldots,k\}\}$ visited by $P$ is either $v^{i'}_{j+1}$ or $v^{j+1}_{i'}$.
			
			By \cref{claim:13} we have that the temporal path segments from $v^{i'}_j$ to $v^{i'}_{j+1}$ and $v^{j+1}_{i'}$, respectively, have duration at least $20n$. However, for the temporal path from $v^{i'}_j$ to $v^{j+1}_{i'}$ (with $j\neq i'-1$) we can obtain a larger lower bound.
			As argued in the proof of \cref{claim:11}, a temporal path segment from $v^{i'}_j$ to $v^{j+1}_{i'}$ has to either traverse two segments of $5n$ vertices in $\{a^{i',j',\ell}_{\ell'} :  \ell'\in[5n]\}$ or $\{b^{i',j',\ell}_{\ell'} :  \ell'\in[5n]\}$ for some $\ell\in[m]$ and $i',j'\in\{i-1,i,j,j+1\}$. 
			More precisely, the temporal path segment has to traverse part of the edge selection gadget of color combination $i',j+1$. To this end, it traverses the $5n$ vertices in $\{a^{i',j+1,\ell}_{\ell''} :  \ell''\in[5n]\}$ for some $\ell\in[m]$. Then it traverses some vertices in the edge selection gadget, and then it traverses the $5n$ vertices in $\{b^{j+1,i',\ell'}_{\ell''} :  \ell''\in[5n]\}$ for some $\ell'\in[m]$.
			
			By construction of $G$, the first vertex of the edge selection gadget visited by the path segment (after traversing vertices in $\{a^{i',j+1,\ell}_{\ell'} :  \ell''\in[5n]\}$) is some vertex $u^\ell_{\ell''}$ with $\ell''\in\{0,1,\ldots,4n\}$.
			The last vertex of the edge selection gadget visited by the path segment is (before traversing the vertices in $\{b^{j+1,i',\ell'}_{\ell''''} :  \ell''''\in[5n]\}$) some vertex $u^{\ell'}_{\ell'''}$ with $\ell'''\in\{0,1,\ldots,4n\}$. By construction of $G$, the duration of a fastest path between $u^\ell_{\ell''}$ and $u^{\ell'}_{\ell'''}$ (in $G$) is at least $3n$.
			Investigating the second connector gadget of the edge selection gadget for $i',j+1$ we can see that a temporal path from $u^\ell_{\ell''}$ and $u^{\ell'}_{\ell'''}$ has duration at least $6n$.
			
			It follows that the temporal path segment from $v^{i'}_j$ to $v^{j+1}_{i'}$ (with $j\neq i'-1$) has duration at least $26n$.
			Furthermore, recall that $P$ starts at $v^i_0$ and ends at $v^i_k$. We have that if $P$ contains a path segment from some $v^{i'}_j$ to $v^{j+1}_{i'}$ some (with $j\neq i'-1$), then $P$ visits a vertex $v^{i''}_{j'}$ with $i''\neq i$. Hence, it needs to contain at least one additional path segment from some $v^{i'}_j$ to some $v^{j+1}_{i'}$ (with $j\neq i-1$). However, then we have that the duration of $P$ is at least $20kn+12n>k(20n+6)+6n-1$, a contradiction.
			
			We can conclude that $P$ only contains temporal path segments from $v^{i}_{j-1}$ to $v^i_{j}$ for $j\in[k]$ and the claim follows.
			%
			%
		\end{claimproof}
		
		Now we have by \cref{claim:12,claim:14} that we can divide $P$ into $k$ segments, the subpaths from $v^i_{j-1}$ to $v^i_j$ for $j\in [k]$. We show that all subpaths except the one from $v_{i-1}^i$ to $v_i^i$ have duration $20n+5$. The subpath from $v_{i-1}^i$ to $v_i^i$ has duration $26n+5$.
		
		\begin{claim}\label{claim:15}
			Let $i\in[k]$ and $j\in[k]\setminus\{i\}$. Let $P$ be a temporal path from $v_{j-1}^i$ to $v_j^i$ that does not visit vertices from connector gadgets and the alignment gadget. If $P$ has duration at most $20n+5$, then it visits exactly two vertices $u^\ell_{\ell'-1},u^\ell_{\ell'}$ with $\ell\in[m]$, and $\ell'\in[4n]$ of the edge selection gadget for color combination $i,j$ (or $j,i$).
		\end{claim}
		\begin{claimproof}
			By the construction of $G$ (and as also argued in the proofs of \cref{claim:12,claim:13}), a temporal path $P$ with duration at most $20n+5$ that does not visit vertices from connector gadgets and the alignment gadget from $v_{j-1}^i$ to $v_j^i$ has to first traverse a segment of $5n$ vertices in $\{a^{i,j-1,\ell}_{\ell'} :  \ell'\in[5n]\}$ and then a segment of $5n$ vertices $\{b^{i,j,\ell}_{\ell'} :  \ell'\in[5n]\}$ for some $\ell\in[m]$. By construction of $G$, the two vertices visited in the edge selection gadget for color combination $i,j$ (or $j,i$) are $u^\ell_{\ell'-1}$ and $u^\ell_{\ell'}$ for some $\ell'\in[4n]$. By inspecting the connector gadgets in an analogous way as in the proof of \cref{claim:13} we can deduce that all consecutive edges traversed by $P$ must have labels that differ by at least 2. If follows that if all consecutive edges have labels that differ by exactly two, then $P$ has duration $20n+5$.
		\end{claimproof}
		
		\begin{claim}\label{claim:16}
			Let $i\in[k]$. Let $P$ be a temporal path from $v_{i-1}^i$ to $v_i^i$ that does not visit vertices from connector gadgets and the alignment gadget. Then $P$ has duration at least $26n+5$.
		\end{claim}
		\begin{claimproof}
			By construction of $G$ we have that $v^i_{i-1}$ and $v^i_i$ are connected via a path inside the verification gadget for color $i$, visiting the $13n+1$ vertices in $\{\hat{u}^i_\ell :  \ell\in[13n+1]\}$. Assume $P$ follows this path. By inspecting the connector gadgets of the verification gadget of color $i$, we can see that all consecutive edges traversed by $P$ must have labels that differ by at least two. It follows that $P$ has duration at least $26n+5$. By construction of $G$ we have that if $P$ does not follow the vertices in $\{\hat{u}^i_\ell :  \ell\in[13n+1]\}$ it has to visit at least three different edge selection gadgets: The one of color combination $i-1,i$, then one of $i-1,i+1$, and then the one of $i,i+1$. If follows that $P$ needs to visit at least four segments of length $5n$ composed of vertices $\{a^{i',j',\ell}_{\ell'} :  \ell'\in[5n]\}$ or $\{b^{i',j',\ell}_{\ell'} :  \ell'\in[5n]\}$ for some $\ell\in[m]$ and $i',j'\in[k]$. By inspecting the connector gadgets of the verification gadgets we know that it takes at least $10n$ time steps to traverse such a segment. Hence, the duration of $P$ is at least $40n$.
		\end{claimproof}
		
		Furthermore, we need the following observation which is relevant when we try to connect the above mentioned segments to a temporal path.
		
		\begin{claim}\label{claim:17}
			Let $i\in[k]$ and $0\le j\le k$. The absolute difference of labels of any two different edges incident with $v^i_j$ is at least two.
		\end{claim}
		\begin{claimproof}
			This follows by inspecting the connector gadgets of the verification gadget of color $i$.
		\end{claimproof}
		
		From \cref{claim:10,claim:11,claim:14,claim:15,claim:16,claim:17} we get that a fastest temporal path $P$ from $v^i_0$ to $v^i_k$ has the following properties.
		\begin{enumerate}
			\item The path $P$ can be segmented into temporal path segments $P_j$ from $v^i_{j-1}$ to $v^i_j$ for $j\in[k]\setminus\{i\}$ such that $P_j$ is a temporal path from $v^i_{j-1}$ to $v^i_j$ that does not visit vertices from connector gadgets and the alignment gadget and has duration $20n+5$.
			\item The segment of $P$ from $v^i_{i-1}$ to $v^i_i$ has duration $26n+5$.
			\item The path $P$ dwells at each vertex $v^i_j$ with $j\in[k-1]$ for exactly two time steps, that is, the absolute difference of the labels on the edges incident with $v^i_j$ that are traversed by $P$ is exactly two.
		\end{enumerate}
		If any of the properties does not hold, then we can observe that $d(v^i_0,v^i_k)>8n+5$ would follow.
		
		Now assume $i\in[k]$ and $j\in[k]\setminus\{i\}$ and consider a fastest temporal path $P_j$ from $v_{j-1}^i$ to $v_j^i$ that does not visit vertices from connector gadgets and the alignment gadget and a fastest temporal path $P_{j+1}$ from $v_{j}^i$ to $v_{j+1}^i$ that does not visit vertices from connector gadgets and the alignment gadget. 
		By \cref{claim:15} we know that $P_j$ visits vertices $u^\ell_{\ell'-1},u^\ell_{\ell'}$ with $\ell\in[m]$, and $\ell'\in[4n]$ of the edge selection gadget for color combination $i,j$. 
		By \cref{claim:6} we have that $\lambda(\{u^\ell_{\ell'-1},u^\ell_{\ell'}\})=(i+j)\cdot (2n\cdot (\sigma_{i,j}(\ell))^2 +18n+6)+2\ell'+2$, where $\sigma_{i,j}$ is the permutation of color combination $i,j$ (or $j,i$).
		Analogously, we have by \cref{claim:15} that $P_{j+1}$ visits vertices $u^{\ell''}_{\ell'''-1},u^{\ell''}_{\ell'''}$ with $\ell''\in[m]$, and $\ell'''\in[4n]$ of the edge selection gadget for color combination $i,j+1$. 
		By \cref{claim:6} we have that $\lambda(\{u^{\ell''}_{\ell'''-1},u^{\ell''}_{\ell'''}\})=(i+j+1)\cdot (2n\cdot (\sigma_{i,j+1}(\ell''))^2 +18n+6)+2\ell'''+2$, where $\sigma_{i,j+1}$ is the permutation of color combination $i,j+1$ (or $j+1,i$). We have that 
		\begin{align*}
			& \lambda(\{u^{\ell''}_{\ell'''-1},u^{\ell''}_{\ell'''}\})-\lambda(\{u^\ell_{\ell'-1},u^\ell_{\ell'}\}) = \\
			& (i+j+1)\cdot (2n\cdot (\sigma_{i,j+1}(\ell''))^2 +18n+6)+2\ell'''+2\\
			& -((i+j)\cdot (2n\cdot (\sigma_{i,j}(\ell))^2 +18n+6)+2\ell'+2) = \\
			& (i+j+1)\cdot 2n\cdot (\sigma_{i,j+1}(\ell''))^2 - (i+j)\cdot 2n\cdot (\sigma_{i,j}(\ell))^2 + 2(\ell'''-\ell') + 18n + 6
		\end{align*}
		By the arguments made before we also have that if $P_j$ and $P_{j+1}$ are both path segments of $P$, then 
		\[
		\lambda(\{u^{\ell''}_{\ell'''-1},u^{\ell''}_{\ell'''}\})-\lambda(\{u^\ell_{\ell'-1},u^\ell_{\ell'}\}) = 20n+6.
		\]
		It follows that 
		\[
		(i+j+1)\cdot 2n\cdot (\sigma_{i,j+1}(\ell''))^2 - (i+j)\cdot 2n\cdot (\sigma_{i,j}(\ell))^2 + 2(\ell'''-\ell') = 2n.
		\]
		Assume that $\sigma_{i,j}(\ell)\neq \sigma_{i,j+1}(\ell'')$, then we have that $(i+j+1)\cdot 2n\cdot (\sigma_{i,j+1}(\ell''))^2 - (i+j)\cdot 2n\cdot (\sigma_{i,j}(\ell))^2<6n$ or $(i+j+1)\cdot 2n\cdot (\sigma_{i,j+1}(\ell''))^2 - (i+j)\cdot 2n\cdot (\sigma_{i,j}(\ell))^2>10n$, since $|(\sigma_{i,j}(\ell''))^2-(\sigma_{i,j}(\ell))^2|\ge 3$ and $(i+j)\ge 3$. However, we have that $\ell',\ell'''\in[4n]$ and hence $|2(\ell'''-\ell')|< 8n$. We can conclude that $\sigma_{i,j}(\ell)= \sigma_{i,j+1}(\ell'')$.
		In this case we have that $(i+j+1)\cdot 2n\cdot (\sigma_{i,j+1}(\ell''))^2 - (i+j)\cdot 2n\cdot (\sigma_{i,j}(\ell))^2= 2n\cdot (\sigma_{i,j}(\ell'))^2$.
		It follows that $2n(\sigma_{i,j}(\ell))^2-2(\ell'''-\ell')=2n$. Again, since $|2(\ell'''-\ell')|< 8n$, we have that $\sigma_{i,j}(\ell)=1$ and in turn this implies that $\ell'=\ell'''$.
		
		Note that if $i=1$ or $i=k$ we can already conclude that $|(\bigcap_{1\le j<i} e_{j,i}) \cap (\bigcap_{i<j\le k} e_{i,j})|=1$. 
		By construction of $G$ we have that for all $j\in[k]\setminus\{i\}$ that $v^i_{j-1}$ and $v^i_j$ are connected to $u^\ell_{\ell'-1}$ and $u^\ell_{\ell'}$ of the edge selection gadget of color combination $i,j$ (or $j,i$), respectively, via paths using vertices $\{a^{i,j,\ell}_{\ell''} :  \ell''\in[5n]\}$ and $\{b^{i,j,\ell}_{\ell''} :  \ell''\in[5n]\}$, respectively, if the vertex $w^i_{\ell'}\in W_i$ (for $i=k$, or vertex $w^i_{\ell'-3n}\in W_i$ for $i=1$) is incident with edge $e^{i,j}_\ell\in F_{i,j}$. Note that since $\sigma_{i,j}(\ell)=1$ we have that $e^{i,j}_\ell\in X$.
		Since $\ell'$ is independent from $\ell$ and $j$, it follows that $(\bigcap_{1\le j<i} e_{j,i}) \cap (\bigcap_{i<j\le k} e_{i,j})=\{w^i_{\ell'}\}$ for $i=k$ and $(\bigcap_{1\le j<i} e_{j,i}) \cap (\bigcap_{i<j\le k} e_{i,j})=\{w^i_{\ell'-3n}\}$ for $i=1$.
		
		Assume now that $1\neq i\neq k$. By \cref{claim:16} we know that the duration of the path segment $P_i$ from $v^i_{i-1}$ to $v^i_i$ is $26n+5$. Consider the path segment $P^\star$ from $v^i_{i-2}$ to $v^i_{i+1}$. By the arguments above we know that $P^\star$ visits vertices $u^\ell_{\ell'-1},u^\ell_{\ell'}$ with $\sigma_{i-1,i}(\ell)=1$, and $\ell'\in[4n]$ of the edge selection gadget for color combination $i-1,i$ and afterwards $P^\star$ visits vertices $u^{\ell''}_{\ell'''-1},u^{\ell''}_{\ell'''}$ with $\sigma_{i,i+1}(\ell'')=1$, and $\ell'''\in[4n]$ of the edge selection gadget for color combination $i,i+1$.
		By analogous arguments as above and the fact that the duration of $P_i$ is $26n+5$ we get that 
		\[
		\lambda(\{u^{\ell''}_{\ell'''-1},u^{\ell''}_{\ell'''}\})-\lambda(\{u^\ell_{\ell'-1},u^\ell_{\ell'}\}) = 46n+6.
		\]
		It follows that 
		\[
		(2i+1)\cdot (20n+6)+2\ell'''+2 -((2i-1)\cdot (20n+6)+2\ell'+2) = 46n+6,
		\]
		and hence $\ell'''-\ell'=3n$.
		By construction of $G$ we have that $v^i_{i-2}$ and $v^i_{i-1}$ are connected to $u^\ell_{\ell'-1}$ and $u^\ell_{\ell'}$ of the edge selection gadget of color combination $i-1,i$, respectively, via paths using vertices $\{a^{i,i-1,\ell}_{\ell''''} :  \ell''''\in[5n]\}$ and $\{b^{i,i-1,\ell}_{\ell''''} :  \ell''''\in[5n]\}$, respectively, if the vertex $w^i_{\ell'}\in W_i$ is incident with edge $e^{i-1,i}_\ell\in F_{i-1,i}$. 
		Furthermore, we have that $v^i_{i}$ and $v^i_{i+1}$ are connected to $u^{\ell''}_{3n+\ell'-1}$ and $u^{\ell''}_{3n+\ell'}$ of the edge selection gadget of color combination $i,i+1$, respectively, via paths using vertices $\{a^{i,i+1,\ell''}_{\ell''''} :  \ell''''\in[5n]\}$ and $\{b^{i,i+1,\ell''}_{\ell''''} :  \ell''''\in[5n]\}$, respectively, if the vertex $w^i_{\ell'}\in W_i$ is incident with edge $e^{i,i+1}_{\ell''}\in F_{i,i+1}$. 
		
		Note that since $\sigma_{i-1,i}(\ell)=\sigma_{i,i+1}(\ell'')=1$ we have that $e^{i-1,i}_\ell\in X$ and $e^{i,i+1}_{\ell''}\in X$.
		Since, again, $\ell'$ is independent from $\ell$ and $j$, it follows that $e^{i-1,i}_\ell\cap e^{i,i+1}_{\ell''} =\{w^i_{\ell'}\}$. By arguments analogous to the ones above we can also deduce that $\bigcap_{1\le j<i} e_{j,i}=\{w^i_{\ell'}\}$ and $\bigcap_{i<j\le k} e_{i,j}=\{w^i_{\ell'}\}$. It follows that $(\bigcap_{1\le j<i} e_{j,i}) \cap (\bigcap_{i<j\le k} e_{i,j})=\{w^i_{\ell'}\}$.
		
		We can conclude that indeed $\bigcup_{e\in X}e$ forms a multicolored clique in $H$.
		
		\subparagraph{$(\Leftarrow)$:} Assume $H$ is a \textsc{Yes}-instance of \textsc{Multicolored Clique} and let $X$ be a solution. We construct the following labeling for the underlying graph $G$, see also \cref{fig:hardness1} for an illustration.
		
		We start with the labels for edges from the alignment gadget. 
		\begin{itemize}
			\item For every $w\in\hat{W}$ we set $\lambda(\{w^\star,w\})=1$.
			\item Let $\hat{v}_0$ belong to some connector gadget and let $w\in\hat{W}$ be neighbor of $\hat{v}_0$. Then we set $\lambda(\{w,\hat{v}_0\})=n^9$.
			\item Let $y^i$ belong to the verification gadget of color $i$ and let $w\in\hat{W}$ be neighbor of $y^i$. Then we set $\lambda(\{w,y^i\})=n^8-1$. Furthermore, we set $\lambda(\{y_i,v^i_0\})=n^8$.
			\item Let $x_1$ belong to the edge selection gadget for color combination $i,j$ and let $w\in\hat{W}$ be neighbor of $x_1$. Then we set $\lambda(\{w,x_1\})=(i+j)(20n+6)$.
		\end{itemize}
		
		Next, consider a connector gadget with vertices $\hat{v}_0,\hat{v}_0',\hat{v}_1,\hat{v}_2,\hat{v}_3,\hat{v}_3'$ and set $A,B$.
		\begin{itemize}
			\item We set $\lambda(\{\hat{v}_0,\hat{v}_1\})=\lambda(\{\hat{v}_,\hat{v}_3\})=n^9$.
			\item We set $\lambda(\{\hat{v}_0',\hat{v}_1\})=\lambda(\{\hat{v}_,\hat{v}_3'\})=n^9+2$.
			\item We set $\lambda(\{\hat{v}_1,\hat{v}_2\})=n^9+1$.
			\item For all vertices $v\in A\setminus B$ we set $\lambda(\{\hat{v}_1,v\})=n^9$ and $\lambda(\{\hat{v}_2,v\})=n^9+2$.
			\item For all vertices $v\in B$ we set $\lambda(\{\hat{v}_1,v\})=\lambda(\{\hat{v}_2,v\})=n^9$.
			\item For all vertices $v\in V^\star\setminus A$ we set $\lambda(\{\hat{v}_1,v\})=\lambda(\{\hat{v}_2,v\})=n^9+2$. (Recall that $V^\star$ denotes the set of all vertices from all edge selection gadgets and all verification gadgets).
		\end{itemize}

		Recall that the following duration requirements were specified in the construction of the instance. It is straightforward to verify that durations requirements we recall in the following are all met, assuming no faster connections are introduced.
		\begin{itemize}
			\item We have set $d(\hat{v}_0,\hat{v}_2)=d(\hat{v}_3,\hat{v}_1)=d(\hat{v}_2,\hat{v}_0')=d(\hat{v}_1,\hat{v}_3')=2$, and $d(\hat{v}_0,\hat{v}_0')=d(\hat{v}_3,\hat{v}_3')=d(\hat{v}_0,\hat{v}_3')=d(\hat{v}_3,\hat{v}_0')=3$.
			\item Let $v\in A$, then we have set $d(v,\hat{v}_0')=3$ and $d(v,\hat{v}_3')=3$.
			\item Let $v\in V^\star\setminus B$, then we have set $d(\hat{v}_0,v)=3$ and $d(\hat{v}_3,v)=3$.
			\item Let $v\in A$ and $v'\in V^\star\setminus B$ such that $v$ and $v'$ are not neighbors, then we have set $d(v,v')=3$.
		\end{itemize}
		For two connector gadgets, one with vertices $\hat{v}_0,\hat{v}_0',\hat{v}_1,\hat{v}_2,\hat{v}_3,\hat{v}_3'$ and with sets $A$ and $B$, and one with vertices $\hat{v}_0',\hat{v}_0'',\hat{v}_1',\hat{v}_2',\hat{v}_3',\hat{v}_3''$ and with sets $A'$ and $B'$, we have set the following durations.
		\begin{itemize}
			\item If there is a vertex $v\in A$ with $v\notin A'$, then we have set $d(\hat{v}_1,\hat{v}_1')=3$.
			\item If there is a vertex $v\in A$ with $v\in A'\setminus B'$, then we have set $d(\hat{v}_1,\hat{v}_2')=3$.
			\item If there is a vertex $v\in V^\star\setminus (A\setminus B)$ with $v\notin A'$, then we have set $d(\hat{v}_2,\hat{v}_1')=3$.    
			\item If there is a vertex $v\in V^\star\setminus (A\setminus B)$ with $v\in A'\setminus B'$, then we have set $d(\hat{v}_2,\hat{v}_2')=3$.
		\end{itemize}
		
		For the alignment gadget the following requirements were specified. Let $x_1$ belong to the edge selection gadget of color combination $i,j$ and let $w\in \hat{W}$ denote the neighbor of $x_1$ in the alignment gadget. Let $\hat{v}_1$ and $\hat{v}_2$ belong to the first connector gadget of the edge selection gadget for color combination $i,j$. Let $\hat{V}$ contain all vertices $\hat{v}_1$ and $\hat{v}_2$ belonging to the other connector gadgets (different from the first one of the edge selection gadget for color combination $i,j$). 
		\begin{itemize}
			\item We have set $d(w^\star,x_1)=(20n+6)(i+j)$.
			\item We have set $d(w^\star,\hat{v}_1)=n^9$, $d(w,\hat{v}_2)=n^9$, $d(w,\hat{v}_1)=n^9-(20n+6)(i+j)+1$, and $d(w,\hat{v}_2)=n^9-(20n+6)(i+j)+1$.
			\item For each vertex $v\in (V^\star\cup \hat{V})\setminus (X_{i,j}\cup \{v_{i,j}^{\star\star}\})$ we have set $d(w^\star,v)=n^9+2$ and $d(w,v)=n^9-(20n+6)(i+j)+3$.
		\end{itemize}
		
		Let $y^i$ belong to the verification gadget of color $i$ and let $w'\in \hat{W}$ denote the neighbor of $y^i$ in the alignment gadget. Let $\hat{v}_1$ and $\hat{v}_2$ belong to the connector gadget of the verification gadget for color $i$. Let $\hat{V}$ contain all vertices $\hat{v}_1$ and $\hat{v}_2$ belonging to the other connector gadgets (different from the one of the verification gadget for color $i$). Let $V_i$ denote the set of all vertices of the verification gadget of color $i$.
		\begin{itemize}
			\item We have set $d(w^\star,y^i)=n^8-1$, $d(w',v^i_0)=2$, and $d(w^\star,v^i_0)=n^8$.
			\item We have set $d(w^\star,\hat{v}_1)=n^9$, $d(w^\star,\hat{v}_2)=n^9$, $d(w',\hat{v}_1)=n^9-n^8$, and $d(w',\hat{v}_2)=n^9-n^8$.
			\item For each vertex $v\in (V^\star\cup \hat{V})\setminus V_i$ we have set $d(w^\star,v)=n^9+1$, $d(w,v)=n^9-n^8+2$, and $d(y^i,v)=n^9-n^8+2$.
		\end{itemize}
		Let $\hat{v}_1$ belong to some connector gadget. We have set $d(w^\star,\hat{v}_1)=n^9$.
		
		We will make sure that no faster connections are introduced by only using even numbers as labels and labels that are strictly smaller than $n^8-1$. Furthermore, we can already see that no vertex except the ones in $\hat{W}$ can reach $w^\star$ and no two vertices $w,w'\in\hat{W}$ can reach each other, as required.

		Next, consider the edge selection gadget for color combination $i,j$ with $i<j$.
		To describe the labels, we define a permutation $\sigma_{i,j}:[m]\rightarrow [m]$ as follows. Let $\{w^i_{\ell'}\}=X\cap W_i$ and $\{w^j_{\ell''}\}=X\cap W_j$. Then, since $X$ is a clique in $H$, we have that $\{w^i_{\ell'},w^j_{\ell''}\}=e^{i,j}_\ell\in F_{i,j}$. We set $\sigma_{i,j}(\ell)=1$ and $\sigma_{i,j}(1)=\ell$. For all $\ell'''\in[m]$ with $1\neq\ell'''\neq\ell$ we set $\sigma_{i,j}(\ell''')=\ell'''$.
		
		Let $x_1, x_2, \ldots, x_m$ belong to the edge selection gadget for color combination $i,j$.
		\begin{itemize}
			\item For all $\ell'''\in[m]$ we set $\lambda(\{x_{\ell'''},v_{i,j}^{\star}\})=(i+j)\cdot (2n({\ell'''})^2 +18n+6)$.
		\end{itemize}
		Note that using these labels, we obey the following duration constraints.
		\begin{itemize}
			\item For all $1\le \ell'''<\ell''''\le m$ we have set $d(x_{\ell'''},x_{\ell''''})=2n\cdot (i+j)\cdot((\ell'''')^2-(\ell''')^2)+1$.
		\end{itemize}
		Furthermore, we set the following labels.
		\begin{itemize}
			\item  For all $\ell'''\in[m]$ we set $\lambda(\{u^{\ell'''}_{0},v_{i,j}^{\star}\})=(i+j)\cdot (2n\cdot (\sigma_{i,j}(\ell'''))^2 +18n+6)+2$, where $u^{\ell'''}_{0}$ belongs to the edge selection gadget for $i,j$.
			\item For all $\ell'''\in[m]$ and $\ell''''\in[4n]$ we set $\lambda(\{u^{\ell'''}_{\ell''''-1},u^{\ell'''}_{\ell''''}\})=(i+j)\cdot (2n\cdot (\sigma_{i,j}(\ell'''))^2 +18n+6)+2\ell''''+2$, where $u^{\ell'''}_{\ell''''-1}$ and $u^{\ell'''}_{\ell''''}$ belong to the edge selection gadget for $i,j$.
			\item  For all $\ell'''\in[m]$ we set $\lambda(\{u^{\ell'''}_{4n},v_{i,j}^{\star\star}\})=(i+j)\cdot (2n\cdot (\sigma_{i,j}(\ell'''))^2 +18n+6)+8n+4$, where $u^{\ell'''}_{4n}$ belongs to the edge selection gadget for $i,j$.
		\end{itemize}
		
		It is straightforward to verify that with these labels we get for all $\ell'''\in[m]$ that $d(x_{\ell'''},v_{i,j}^{\star\star})=8n+5$, as required.
		Furthermore, we get that for all $\ell'''\in[m]$ that $d(v_{i,j}^{\star\star},x_{\ell'''})=\infty$. To see this, consider the following. Vertex $v_{i,j}^{\star\star}$ is not temporally connected to vertices $x_{\ell'''}$ with $\ell'''\in[m]$ via any of the connector gadgets, since for all connector gadgets where $v_{i,j}^{\star\star}\in A$ we have that all vertices $x_{\ell'''}$ with $\ell'''\in[m]$ are either contained in $B$ or they are not contained in $A$. By the construction of the labels of the connector gadgets, it follows that $v_{i,j}^{\star\star}$ cannot reach any vertex $x_{\ell'''}$ with $\ell'''\in[m]$ via the connector gadgets.
		We can observe that in all other connections in the underlying graph from $v_{i,j}^{\star\star}$ to a vertex $x_{\ell'''}$ with $\ell'''\in[m]$ are paths which have non-increasing labels, hence they also do not provide a temporal connection.
		
		Furthermore, we get that for all $1\le \ell'''\le\ell''''\le m$ we get that $d(x_{\ell'''},x_{\ell''''})=2n\cdot (i+j)\cdot((\ell'''')^2-(\ell''')^2)+1$, through a temporal path via $v_{i,j}^\star$.
		By similar observations as in the previous paragraph, we also have that $d(x_{\ell''''},x_{\ell'''})=\infty$.
		
		Finally, consider the verification gadget for color $i$. 
		Let $1\le j<i$. 
		Let $\{w^i_{\ell'}\}=X\cap W_i$ and $\{w^j_{\ell''}\}=X\cap W_j$ and $\{w^i_{\ell'},w^j_{\ell''}\}=e^{j,i}_\ell\in F_{j,i}$. 
		Recall that we set $\sigma_{j,i}(\ell)=1$ and $\sigma_{j,i}(1)=\ell$. For all $\ell''\in[m]$ with $1\neq\ell''\neq\ell$ we set $\sigma_{j,i}(\ell'')=\ell''$.
		Recall that we set $\lambda(\{u^{\ell}_{\ell'-1},u^{\ell}_{\ell'}\})=(i+j)\cdot (20n+6)+2\ell'+2$, where $u^{\ell}_{\ell'-1}$ and $u^{\ell}_{\ell'}$ belong to the edge selection gadget for $j,i$.
		Now we set for all $\ell''\in[5n-1]$ and all $\ell'''\in[m]$ the following.
		\begin{itemize}
			\item $\lambda(\{a^{i,j,\ell'''}_{5n},u^{\ell'''}_{\ell''''}\})=(i+j)\cdot (20n+6)+2\ell'$ for all $\ell''''$ such that this edge exists.
			\item $\lambda(\{a^{i,j,\ell'''}_{1},v^i_{j-1}\})=(i+j)\cdot (20n+6)+2\ell'-10n-2$.
			\item $\lambda(\{a^{i,j,\ell'''}_{\ell''},a^{i,j,\ell'''}_{\ell''+1}\})=(i+j)\cdot (20n+6)+2\ell'-10n+2\ell''$.
			\item $\lambda(\{b^{i,j,\ell'''}_{5n},u^{\ell'''}_{\ell''''}\})=(i+j)\cdot (20n+6)+2\ell'+4$ for all $\ell''''$ such that this edge exists.
			\item $\lambda(\{b^{i,j,\ell'''}_{1},v^i_{j}\})=(i+j)\cdot (20n+6)+2\ell'+10n+6$.
			\item $\lambda(\{b^{i,j,\ell'''}_{\ell''},b^{i,j,\ell'''}_{\ell''+1}\})=(i+j)\cdot (20n+6)+2\ell'+10n-2\ell''+4$.
		\end{itemize}
		For all $\ell''\in[13n]$ we set the following.
		\begin{itemize}
			\item $\lambda(\{\hat{u}^i_{\ell''},\hat{u}^i_{\ell''+1}\})=2i\cdot (20n+6)+2\ell'-10n+2\ell''-2$.
			\item $\lambda(\{v^i_{i-1},\hat{u}^i_{1}\})=2i\cdot (20n+6)+2\ell'-10n-2$.
			\item $\lambda(\{v^i_i,\hat{u}^i_{13n+1}\})=2i\cdot (20n+6)+2\ell'+16n+4$.
		\end{itemize}
		Let $i<j\le k$. 
		Let $\{w^i_{\ell'}\}=X\cap W_i$ and $\{w^j_{\ell''}\}=X\cap W_j$ and $\{w^i_{\ell'},w^j_{\ell''}\}=e^{i,j}_\ell\in F_{i,j}$. 
		Recall that we set $\sigma_{i,j}(\ell)=1$ and $\sigma_{i,j}(1)=\ell$. For all $\ell''\in[m]$ with $1\neq\ell''\neq\ell$ we set $\sigma_{i,j}(\ell'')=\ell''$.
		Recall that we set $\lambda(\{u^{\ell}_{3n+\ell'-1},u^{\ell}_{3n+\ell'}\})=(i+j)\cdot (20n+6)+2\ell'+6n+2$, where $u^{\ell}_{3n+\ell'-1}$ and $u^{\ell}_{3n+\ell'}$ belong to the edge selection gadget for $i,j$.
		Now we set for all $\ell''\in[5n-1]$ and all $\ell'''\in[m]$ the following.
		\begin{itemize}
			\item $\lambda(\{a^{i,j,\ell'''}_{5n},u^{\ell'''}_{\ell''''}\})=(i+j)\cdot (20n+6)+2\ell'+6n$ for all $\ell''''$ such that this edge exists.
			\item $\lambda(\{a^{i,j,\ell'''}_{1},v^i_{j-1}\})=(i+j)\cdot (20n+6)+2\ell'-4n-2$.
			\item $\lambda(\{a^{i,j,\ell'''}_{\ell''},a^{i,j,\ell'''}_{\ell''+1}\})=(i+j)\cdot (20n+6)+2\ell'-4n+2\ell''$.
			\item $\lambda(\{b^{i,j,\ell'''}_{5n},u^{\ell'''}_{\ell''''}\})=(i+j)\cdot (20n+6)+2\ell'+6n+4$ for all $\ell''''$ such that this edge exists.
			\item $\lambda(\{b^{i,j,\ell'''}_{1},v^i_{j}\})=(i+j)\cdot (20n+6)+2\ell'+16n+6$.
			\item $\lambda(\{b^{i,j,\ell'''}_{\ell''},b^{i,j,\ell'''}_{\ell''+1}\})=(i+j)\cdot (20n+6)+2\ell'+16n-2\ell''+4$.
		\end{itemize}
		
		Now we verify that we meet the duration requirements.
		For all $0\le j<j'<i$ and all $i\le j<j'\le k$ we have set the following.
		\begin{itemize}
			\item We set $d(v^i_j,v^i_{j'})=(20n+6)(j'-j)-1$.
		\end{itemize}
		To see that this holds, we analyse the fastest paths from vertices $v^i_{j-1}$ to vertices $v^i_j$ for $j\in[k]\setminus\{i\}$. 
		Let $\{w^i_{\ell'}\}=X\cap W_i$ and $\{w^j_{\ell''}\}=X\cap W_j$ and $\{w^i_{\ell'},w^j_{\ell''}\}=e^{i,j}_\ell\in F_{i,j}$. Then, starting at $v^i_{j-1}$, we follow the vertices in $\{a^{i,j,\ell}_{\ell''} :  \ell''\in[5n]\}$ to arrive at $u^{\ell}_{\ell'-1}$. From there we move to $u^{\ell}_{\ell'}$ and from there we continue along the vertices in $\{b^{i,j,\ell}_{\ell''} :  \ell''\in[5n]\}$ to arrive at $v^i_j$. By construction this describes a fastest temporal path from $v^i_{j-1}$ to $v_j$ with duration $20n+5$.
		To get from $v^i_j$ to $v^i_{j'}$ for $0\le j<j'<i$ we move from $v^i_j$ to $v^i_{j+1}$ in the above described fashion and from there to $v^i_{j+1}$ and so on until we arrive at $v^i_{j'}$. By construction this yields a fastest temporal path from $v^i_j$ to $v^i_{j'}$ with duration $(20n+6)(j'-j)-1$, as required.
		The case where $i\le j<j'\le k$ is analogous.
		
		For all $0\le j<i$ and all $i\le j'\le k$ we have set the following.
		\begin{itemize}
			\item We set $d(v^i_j,v^i_{j'})=(20n+6)(j'-j)+6n-1$.
		\end{itemize}
		Here we move from $v^i_j$ to $v^i_{i-1}$ in the above described fashion. Then we move from $v^i_{i-1}$ to $v^i_i$ along vertices $\{\hat{u}^i_{\ell''} :  \ell''\in[13+1]\}$ and then we move from $v^i_i$ to $v^i_{j'}$ again in the above described fashion. By construction this yields a fastest temporal path from $v^i_j$ to $v^i_{j'}$ with duration $(20n+6)(j'-j)+6n-1$, as required.
		
		By similar observations as in the analysis for the edge selection gadgets, we also get that for all $1\le j<j'\le k$ that $d(v^i_{j'},v^i_j)=\infty$.
		
		This finishes the proof.

		\subparagraph{Infinity gadget.} Finally, we show how to get rid of the infinity entries in $D$ and how to allow a finite $\Delta$. To this end, we introduce the \emph{infinity gadget}. We add four vertices $z_1, z_2, z_3, z_4$ to the graph and we set $\Delta=n^{11}$. Let $V$ denote the set of all remaining vertices. We set the following durations.
		\begin{itemize}
			\item For all $v\in V$ we set $d(z_1,v)=2$, $d(z_2,v)=d(v,z_2)=1$, $d(z_3,v)=d(v,z_3)=1$, and $d(z_4,v)=2$.
			Furthermore, we set $d(v,z_1)=n^{11}$ and $d(v,z_4)=n^{10}-1$.
			\item We set $d(z_1,z_2)=d(z_2,z_1)=1$, $d(z_2,z_3)=d(z_3,z_2)=1$, and $d(z_3,z_4)=d(z_4,z_3)=1$.
			\item We set $d(z_1,z_3)=3$, $d(z_3,z_1)=n^{11}-1$, $d(z_2,z_4)=n^{10}-2$, and $d(z_4,z_2)=n^{11}-n^{10}+4$.
			\item We set $d(z_1,z_4)=n^{10}$ and $d(z_4,z_1)=2n^{11}-n^{10}+2$.
			\item For every pair of vertices $v,v'\in V$ where previously the duration of a fastest path from $v$ to $v'$ was specified to be infinite, we set $d(v,v')=n^{10}$.
		\end{itemize}
		Now we analyse which implications we get for the labels on the newly introduced edges.
		Assume $\lambda(\{z_1,z_2\})=t$, then we get the following.
		For all $v\in V$ we have that $d(z_1,v)=2$ and hence we get that $\lambda(\{z_2,v\})=t+1$. 
		Since $d(z_1,z_4)=n^{10}$, we have that $\lambda(z_3,z_4)=t+n^{10}-1$.
		From this follows that for all $v\in V$, since $d(z_4,v)=2$, that $\lambda(\{z_3,v\})=t+n^{10}$.
		Finally, since $d(z_1,z_3)=3$, we have that $\lambda(\{z_2,z_3\})=t+2$.
		For an illustration see \cref{fig:hardness2}. It is easy to check that all duration requirements between vertex pairs in $\{z_1,z_2,z_3,z_4\}$ are met and that all duration requirements between each vertex $v\in V$ and each vertex in  $\{z_1,z_2,z_3,z_4\}$ are met. Furthermore, it is easy to check that the gadget increases the feedback vertex set by two ($z_2$ and $z_3$ need to be added).
		
		\begin{figure}[t]
				\centering
				
				\includegraphics{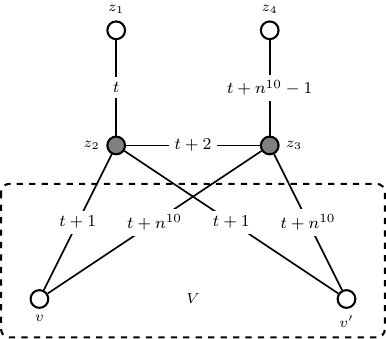}
				\caption{Illustration of the infinity gadget. Gray vertices need to be added to the feedback vertex set.}\label{fig:hardness2}
			\end{figure}
			
			Lastly, consider two vertices $v,v'\in V$. Note that before the addition of the infinity gadget, by construction of $G$ we have that $d(v,v')\le n^9+2$ or $d(v,v')=\infty$. Furthermore, if $D$ is a \textsc{Yes}-instance, we have shown in the correctness proof of the reduction that the difference between the smallest label and the largest label is at most $n^9+1$. This implies that for a vertex pair $v,v'\in V$ with $d(v,v')=\infty$ we have in the periodic case with $\Delta=n^{11}$, that $d(v,v')\ge n^{11}-n^9>n^{10}$. Which means, after adding the vertices and edges of the infinity gadget, we indeed have that $d(v,v')=n^{10}$. For all vertex pairs $v,v'$ where in the original construction we have $d(v,v')\neq\infty$, we can also see that adding the infinity gadget and setting $\Delta=n^{11}$ does not change the duration of a fastest path from $v$ to $v'$, since all newly added temporal paths have duration at least $n^{10}$. We can conclude that the originally constructed instance $D$ is a \textsc{Yes}-instance if and only if it remains a \textsc{Yes}-instance after adding the infinity gadget and setting $\Delta=n^{11}$.
		\end{proof}

		\section{Algorithms for \deltaExact}\label{sec:algos}
		In this section we provide several algorithms for \deltaExact. By \cref{thm:NPhardness} we have that \deltaExact\ is NP-hard in general, hence we start by identifying restricted cases where we can solve the problem in polynomial time.
		We first show in \cref{sec:treealgo} that if the underlying graph $G$ of an instance $(D, \Delta)$ of \deltaExact\ is a tree, then we can determine desired $\Delta$-periodic labeling $\lambda$ of $G$ in polynomial time. In \cref{sec:FPT} we generalize this result. We show that \deltaExact\ is fixed-parameter tractable when parameterized by the feedback edge number of the underlying graph. Note that our parameterized hardness result (\cref{thm:W1wrtFVS}) implies that we presumably cannot replace the feedback edge number with the smaller parameter feedback vertex number, or any other parameter that is smaller than the feedback vertex number, such as e.g.\ the treewidth.

		\subsection{Polynomial-time algorithm for trees}\label{sec:treealgo}
		We now provide a polynomial-time algorithm for \deltaExact\ when the underlying graph is a tree.
		Let $D$ be the input matrix and let the underlying graph $G$ of $D$ be a tree on $n$ vertices $\{v_1, v_2, \dots, v_n\}$.
		Let $v_i,v_j$ be two arbitrary vertices in $G$, then we know that there exists a unique (static) path $P_{i,j}$ from $v_i$ to $v_j$. We will heavily exploit this in our algorithm.

		\begin{theorem} \label{thm:deltaExact-PolyTimeTrees}
			\deltaExact\ can be solved in polynomial time on trees.
		\end{theorem}
		
		\begin{proof}
			Let $D$ be an input matrix for problem \deltaExact\ of dimension $n \times n$.
			Let us fix the vertices of the corresponding graph $G$ of $D$ as $v_1, v_2, \dots, v_n$, where vertex $v_i$ corresponds to the row and column $i$ of matrix $D$.
			This can be done in polynomial time as we need to loop through the matrix $D$ once and connect vertices $v_i, v_j$ for which $D_{i,j} = 1$. At the same time we also check if $D_{i,i} = 0$, for all $i \in [n]$.
			When $G$ is constructed we run DFS algorithm on it and check that it has no cycles.
			If at any step we encounter a problem, our algorithm stops and returns a negative answer.
			
			Having computed $G$, our algorithm proceeds as follows. We pick an arbitrary edge $f$ and give it label one, that is, $\lambda(f)=1$. Now we push all edges incident with $f$ into a (initially empty) queue. Now we repeat the following as long as the queue is not empty:
			\begin{itemize}
				\item Pop edge $e=\{u,v\}$ from the queue. Since $e$ was pushed into the queue, there is an edge $e'$ incident with $e$ that already obtained a label. Let w.l.o.g.\ $e'=\{v,w\}$. Then we set $\lambda(e)=(\lambda(e')-D_{u,w}+1)\bmod \Delta$.
				\item Push all edges incident with $e$ that have not received a label yet into the queue.
			\end{itemize}
			When the queue is empty, all edges have received a label. Iterate over all vertex pairs $u,v$ and check whether the fastest path from $u$ to $v$ in $(G,\lambda)$ has duration $D_{u,v}$. If this check succeeds for all vertex pairs, output the labeling $\lambda$, otherwise abort.
			
			It is easy to see that the described algorithm runs in polynomial time. In the remainder, we prove that it is correct.
			
			$(\Rightarrow)$: Since the algorithm checks at the end whether all durations specified in $D$ are realized by the corresponding fastest paths, we clearly face a yes-instance whenever the algorithm outputs a labeling.
			
			$(\Leftarrow)$: Assume we face a yes-instance, then there exists a labeling $\lambda^\star$ that realizes all durations specified in $D$. Let $e^\star$ denote the edge initially picked by the algorithm. For all edges $e$ let $\lambda(e)=(\lambda^\star(e)-\lambda^\star(e^\star)+1)\bmod \Delta$. Clearly, the labeling $\lambda$ also realizes all durations specified in $D$ since $\lambda$ is obtained by adding the constant $(1-\lambda^\star(e^\star))$ modulo $\Delta$ to all labels of $\lambda^\star$ which does not change the duration of any temporal path, that is all durations in $(G,\lambda^\star)$ are the same as their counterparts in $(G,\lambda)$. We claim that our algorithm computes and outputs $\lambda$.
			
			We prove that our algorithm computes $\lambda$ by induction on the distance of the labeled edges to $e^\star$, where the distance of two edges $e,e'$ is defined as the length of a shortest path that uses $e$ as its first edge and $e'$ as its last edge.
			
			Initially, our algorithm labels $e^\star$ with one, which equals $\lambda(e^\star)$. Now let $e$ be an edge popped off the queue by the algorithm in some iteration, that is on the distance $i$ from $e^\star$. Let $e'$ be the edge incident with $e$ that is on the distance $i-1$ from $e^\star$. Since $G$ is a tree $e'$ has already been considered by the algorithm and thus already has a label. By induction we have that the algorithm labeled $e'$ with $\lambda(e')$. 
			Assume that $e=\{u,v\}$ and $e'=\{v,w\}$. Since $G$ is a tree there is only one path from $u$ to $w$ in $G$ and it uses edges $e$ and $e'$. It follows that $\lambda(e')-\lambda(e)+1=D_{u,w}$ if $\lambda(e')>\lambda(e)$, and $\lambda(e')-\lambda(e)+\Delta+1=D_{u,w}$ otherwise. Our algorithm labels $e$ with $(\lambda(e')-D_{u,w}+1)\bmod \Delta$. It is straightforward to verify that the label of $e$ computed by the algorithm equals $\lambda(e)$. It follows that the algorithm computes~$\lambda$.
		\end{proof}

		\subsection{FPT-algorithm for feedback edge number}\label{sec:FPT}
		
		Recall from \cref{sec:treealgo} that the main reason, for which \deltaExact\ is straightforward to solve on trees, is twofold: 
		\begin{itemize}
			\item between any pair of vertices $v_i$ and $v_j$ in the tree $T$, there is a \emph{unique} path $P$ in $T$ from $v_i$ to $v_j$, and 
			\item in any periodic temporal graph $(T,\lambda,\Delta)$ and any fastest temporal path $P=((e_1,t_1),\ldots,(e_i,t_i),\ldots,(e_j,t_j),\ldots,(e_{\ell-1},t_{\ell-1}))$ from $v_1$ to $v_{\ell}$ we have that the sub-path $P'=((e_i,t_i),\ldots,(e_{j-1},t_{j-1}))$ is also a fastest temporal path from $v_i$ to $v_j$.
		\end{itemize}
		However, these two nice properties do not hold when the underlying graph is not a tree. For example, in~\cref{fig:ftpExample}, the fastest temporal path from $u$ to $v$ is $P_{u,v}$ (depicted in blue) goes through~$w$, however the sub-path of $P_{u,v}$ that stops at $w$ is not the fastest temporal path from $u$ to $w$. The fastest temporal path from $u$ to $w$ consists only of the single edge $uw$ (with label $9$ and duration $1$, depicted in red).
		
		Nevertheless, we prove in this section that we can still solve \deltaExact\ efficiently if the underlying graph is similar to a tree; more specifically we show the following result, which turns out to be non-trivial.
		
		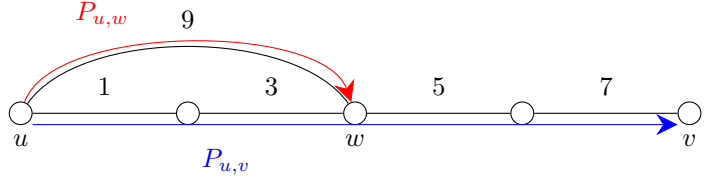
\begin{figure}[t]
			\centering
			\begin{tikzpicture}[xscale=2.2,yscale=0.5]
				\node[vert,label=below:$u$] (v1) at (1,0) {};
				\node[vert] (v2) at (2,0) {};
				\node[vert,label=below:$w$] (v3) at (3,0) {};
				\node[vert] (v4) at (4,0) {};
				\node[vert,label=below:$v$] (v5) at (5,0) {};
				
				\draw[transform canvas={yshift=-1.5mm}, blue]
				(v1) edge[diredge2]  node[pos=0.3,yshift=-2,label=below:$P_{u,v}$] {} (v5) ;
				
				\draw[red]
				(v1) edge[diredge2] [out=85,in=95,distance=2.1cm] node[pos=0.3,yshift=2,label=above:$P_{u,w}$] {} (v3) ;
				
				\draw (v1) -- node[label=above:$1$] {} (v2) -- node[label=above:$3$] {}  (v3) -- node[label=above:$5$] {} (v4) -- node[label=above:$7$] {} (v5);
				\draw (v1) to [out=80,in=100,distance=2cm] node[label=above:$9$] {} (v3);

			\end{tikzpicture}
			\caption{An example of a temporal graph (with $\Delta \geq 9$), where the fastest temporal path $P_{u,v}$ (in blue) from $u$ to $v$ is of duration $7$,
				while the fastest temporal path $P_{u,w}$ (in red) from $u$ to a vertex $w$, that is on a path $P_{u,v}$, is of duration $1$ and is not a subpath of $P_{u,v}$.		
				\label{fig:ftpExample}}
		\end{figure}
		
		\begin{theorem}\label{thm:FPTwrtFES}
			\deltaExact\ is in FPT when parameterized by the feedback edge number of the underlying graph.
		\end{theorem}
		
		From \cref{thm:W1wrtFVS} and \cref{thm:FPTwrtFES} we immediately get the following, which is the main result of the paper.
		
		\begin{corollary}
			\deltaExact\ is:
			\begin{itemize}
				\item in FPT when parameterized by the \emph{feedback edge number} or any larger parameter, such as the \emph{maximum leaf number}.
				\item W[1]-hard when parameterized by the \emph{feedback vertex number} or any smaller parameter, such as: 
				\emph{treewidth}, 
				\emph{degeneracy}, 
				\emph{cliquewidth}, 
				\emph{distance to chordal graphs}, and 
				\emph{distance to outerplanar graphs}.
			\end{itemize}
		\end{corollary}
		
		Before presenting the structure of our algorithm for \cref{thm:FPTwrtFES}, observe that, in a static graph, the number of paths between two vertices can be upper-bounded by a function $f(k)$ of the feedback edge number $k$ of the graph~\cite{casteigts2021finding}.
		This is true as any such path can traverse $0,1,2,\dots k$ feedback edges in different order.
		Therefore, for any fixed pair of vertices $u$ and $v$, we can ``guess'' the edges of the fastest temporal path from $u$ to $v$
		(by guess we mean enumerate and test all possibilities).
		However, 
		for an FPT algorithm with respect to $k$, we cannot afford to guess the edges of the fastest temporal path for each of the $O(n^2)$ pairs of vertices. 
		To overcome this difficulty, our algorithm follows this high-level strategy:
		\begin{itemize}
			\item We identify a small number $f(k)$ of ``important vertices''.
			\item For each pair $u,v$ of important vertices, we guess the edges of the fastest temporal path from $u$ to $v$ (and from $v$ to $u$).
			\item From these guesses we can still not deduce the edges of the fastest temporal paths between many pairs of non-important vertices. However, as we prove, it suffices to guess only a small number of specific auxiliary structures (to be defined later).
			\item From these guesses we deduce fixed relationships between the labels of most of the edges of the graph. 
			\item For all the edges, for which we have not deduced a label yet, we introduce a \emph{variable}. With all these variables, we build an Integer Linear Program (ILP). 
			Among the constraints in this ILP we have that, for each of the $O(n^2)$ pairs of vertices $u,v$ in the graph, the duration of one specific temporal path from $u$ to $v$ (according to our guesses) is \emph{equal} to the desired duration $D_{u,v}$, while the duration of each of the other temporal path from $u$ to $v$ is \emph{at least} $D_{u,v}$.
			\item 
			Each specific configuration of fastest temporal paths among all pairs of vertices corresponds to a specific ILP instance. By exhaustively trying all possible fastest temporal paths configurations it follows that our instance of \deltaExact\ has a solution if and only if at least one of these ILPs has a feasible solution.
			As each ILP can be solved in FPT time with respect to $k$ by Lenstra's Theorem~\cite{Lenstra1983Integer} (the number of variables is upper bounded by a function of $k$), we obtain our FPT algorithm for \deltaExact\ with respect to $k$.
		\end{itemize}

		For the remainder of this section, we fix the following notation. 
		Let $D$ be the input matrix of \deltaExact, \ie
		the matrix of the fastest temporal paths between all pairs of $n$ vertices, and let $G$ be its underlying graph, on $n$ vertices and $m$ edges.
		With $F$ we denote a minimum feedback edge set of $G$, and with $k$ the feedback edge number of $G$.
		We are now ready to present our FPT algorithm. For an easier readability we split the description and analysis of the algorithm in five subsections.
		We start with a preprocessing procedure for graph $G$, where we define a set of interesting vertices which then allows us to guess the desired structures.
		Next we introduce some extra properties of our problem, that we then use to create ILP instances and their constraints.
		At the end we present how to solve all instances and produce the desired labeling $\lambda$ of $G$, if possible.

		\subsubsection{Preprocessing of the input \label{sec:preprocessing-FPT}}
		From the underlying graph $G$ of $D$ we extract a (connected) graph $G'$ by
		iteratively removing vertices of degree one from $G$,
		and denote with 
		\[
		Z = V(G) \setminus V(G').
		\]
		Then we determine a minimum feedback edge set $F$ of $G'$.
		Note that $F$ is also a minimum feedback edge set of $G$.
		Lastly, we determine sets $U$, of \emph{vertices of interest}, and $U^*$ of the neighbors of vertices of interest, in the following way.
		Let $T$ be a spanning tree of $G'$, with $F$ being the corresponding feedback edge set of $G'$.
		Let $V_1 \subseteq V(G')$ be the set of leaves in the spanning tree $T$, $V_2 \subseteq V(G')$ be the set of vertices of degree two in $T$, that are incident to at least one edge in $F$, 
		and let $V_3 \subseteq V(G')$ be the set of vertices of degree at least $3$ in $T$. 
		Then $|V_1| + |V_2| \leq 2k$, since every leaf in $T$ and every vertex in $V_2$ is incident to at least one edge in $F$,
		and $|V_3| \leq |V_1|$ by the properties of trees.
		We denote with 
		\[U = V_1 \cup V_2 \cup V_3\]
		the set of vertices of interest. It follows that $|U| \leq 4k$.
		We set $U^*$ to be the set of vertices in $V(G') \setminus U$ that are neighbors of vertices in $U$, \ie 
		\[U^* = \{v \in V(G') \setminus U  :  u \in U, v \in N(u)\}.\]
		Again, using the tree structure, we get that for any $u \in U$ its neighborhood is of size $|N(u)| \in O(k)$, since every neighbor of $u$ is the first vertex of a (unique) path to another vertex in $U$.
		It follows that $|U^*| \in O(k^2)$.
		
		From the construction of~$Z$ (\ie by exhaustively removing vertices of degree one from~$G$), it follows that $G[Z]$ (the graph induced in $G$ by~$Z$) is a forest, \ie consists of disjoint trees. 
		Each of these trees has a unique neighbor $v$ in~$G'$. 
		Denote by $T_v$ the tree obtained by considering such a vertex $v$ and all the trees from $G[Z]$ that are incident to $v$ in $G$. 
		We then refer to $v$ as the \emph{clip vertex} of the tree $T_v$.
		In the case where $v$ is a vertex of interest we define also the set $Z^*_v$ of \emph{representative vertices} of~$T_v$, as follows.
		We first create an empty set~$C_w$ for every vertex $w$ that is a neighbor of $v$ in~$G'$.
		We then iterate through every vertex~$r$ that is in the first layer of the tree $T_v$ (\ie vertex that is a child of the root $v$ in the tree~$T_v$), check the matrix $D$ and find the vertex~$w \in N_{G'}(v)$ that is on the smallest duration from $r$.
		In other words, for an $r \in N_{T_v}(v)$ we find $w \in N_{G'}(v)$ such that $D_{r,w} \leq D_{r,w'}$ for all $w' \in  N_{G'}(v)$.
		We add vertex $r$ to $C_w$.
		In the case when there exists also another vertex $w' \in  N_{G'}(v) $ for which $D_{r,w'} = D_{r,w}$, we add $r$ also to the set $C_{w'}$. In fact, in this case $C_{w'} = C_w$.
		At the end we create $|N_{G'}(v)| \in O(k)$ sets $C_w$, whose union contains all children of $v$ in $T_v$. 
		For every two sets $C_w$ and $C_{w'}$, where $w,w'\in N_{G'}(v)$, we have that either $C_w = C_{w'}$, or $C_w \cap C_{w'} = \emptyset$. 
		We interpret each of these sets $\{C_w : w \in N_{G'}(v)\}$ as an \emph{equivalence class} of the neighbors of $v$ in the tree $T_v$. 
		Now, from each equivalence class~$C_w$ we choose an arbitrary vertex $r_w \in C_w$ and put it into the set $Z^*_v$.
		We repeat the above procedure for all trees $T_u$ with the clip vertex $u$ from $U$, and define $Z^*$ as  
		\begin{equation}
			Z^* = \bigcup\limits_{v \in U} Z^*_v.
		\end{equation}
		Since $|U| \in O(k)$ and for each $u \in U$ it holds $|N_{G'}(u)| \in O(k)$, we get that $|Z^*| \in O(k^2)$. 
		Finally, the set of \emph{important vertices} is defined as the set $U \cup U^{\ast} \cup Z^{\ast}$.
		For an illustration see \cref{fig:labelingVertices}. Note that determining sets $U, U^*$ and $Z^*$ takes linear time.
		
		\begin{figure}[t]
			\centering
			\includegraphics[width=0.8\columnwidth]{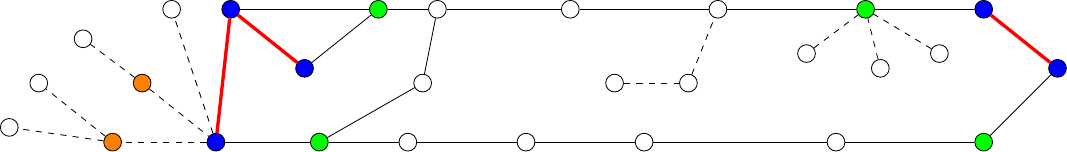}
			\caption{An example of a graph with its important vertices: $U$ (in blue), $U^*$ (in green) and $Z^*$ (in orange).
				Corresponding feedback edges are marked with a thick red line, while dashed edges represent the edges (and vertices) ``removed'' from $G'$ at the initial step.
				\label{fig:labelingVertices}}
		\end{figure}

		Recall that a labeling $\lambda$ of $G$ satisfies $D$ if the duration of a fastest temporal path from each vertex $v_i$ to each other vertex $v_j$ equals $D_{v_i, v_j}$.
		In order to find a labeling that satisfies this property we split our analysis in nine cases.
		We consider the fastest temporal paths 
		where the starting vertex is in one of the sets $U,V(G') \setminus U, Z$,
		and similarly the destination vertex is in one of the sets $U,V(G') \setminus U, Z$.
		In each of these cases, we guess the underlying path $P$ that at least one fastest temporal path from the vertex $v_i$ to $v_j$ follows, 
		which results in one equality constraint for the labels on the path $P$. 
		For all other temporal paths from $v_i$ to $v_j$ we know that they cannot be faster, so we introduce inequality constraints for them.
		This results in producing $f(k) \cdot |D| ^ {O(1)}$ constraints. 
		Note that we have to do this while keeping the total number of variables upper-bounded by some function in $k$.
		
		For an easier understanding and analysis of the algorithm, we give the following definition.
		\begin{definition}\label{def:segments}
			Let $U \subseteq V(G')$ be a set of vertices of interest and let $u,v \in U$.
			A path $P = (u=v_1,v_2, \dots, v_p = v)$  with at least two edges in graph $G'$, where all inner vertices are not in $U$, \ie $v_i \notin U$ for all $i \in \{ 2, 3, \dots, p-1\}$,
			is called a \emph{segment} from $u$ to $v$, which we denote as $S_{u,v}$.
		\end{definition}
		Note from \cref{def:segments} that $S_{u,v} \neq S_{v,u}$ since we consider paths to be directed. It is also worth emphasizing that $S_{v,u}$ is essentially the reverse path of $S_{u,v}$.
		Furthermore, it's important to observe that a temporal path in $G'$ between two vertices of interest is either a segment or consists of a sequence of segments.
		Moreover, any inner vertex $v_i$ in the segment $S_{u,v}$ ($v_i \in S_{u,v} \setminus \{u,v\}$) is part of precisely two segments: $S_{u,v}$ and $S_{v,u}$.
		Given that we have at most $4k$ interesting vertices in $G'$, we can deduce the following crucial result.
		\begin{corollary}\label{obs:FPT-k2segments}
			There are $O(k^2)$ segments in $G'$.
		\end{corollary}
		
		\subsubsection{Guessing necessary structures \label{sec:FPT-guessing}}
		Once the sets $U, U^*$ and $Z^*$ are determined, we are ready to start guessing  the necessary structures.
		Note that whenever we say that we guess the fastest temporal path between two vertices, we mean that we guess the underlying path of a representative fastest temporal path between those two vertices.
		To describe the guesses, we introduce the following notation. Let $u,v,x$ be three vertices in $G'$. We write $u \leadsto x \rightarrow v$ to denote a temporal path from $u$ to $v$ that passes through $x$, and then goes directly to $v$ (via one edge or a unique path in $G'$). 
		In other words, if the fastest path between two vertices is not uniquely determined we denote it by $\leadsto$, while if it is unique we denote it by $\rightarrow$.
		%
		We guess the following paths.
		\begin{enumerate}[G-1.]
			\setcounter{enumi}{\value{guesscounter}}
			\item \label{FPT-guessFTPamongU}
			The fastest temporal paths between all pairs of vertices of $U$.
			For a pair $u,v$ of vertices in $U$, there are $k^{O(k)}$ possible paths in $G'$ between them. 
			Therefore, we have to try all $k^{O(k)}$ possible paths, where at least one of them will be a fastest temporal path from $u$ to $v$, respecting the values from $D$.
			Repeating this procedure for all pairs of vertices $u,v \in U$ we get $k^{O(k^3)}$ different variations of the fastest temporal paths between all pairs of vertices in $U$.
			\item \label{FPT-guessFTPamongZstar}
			The fastest temporal paths between all pairs of vertices in $Z^*$, 
			which by similar arguing as for vertices in $U$, gives us $k^{O(k^5)}$ guesses.
			\item \label{FPT-guessFTPamongUstar}
			The fastest temporal paths between all pairs of vertices in $U^*$.
			This gives us $k^{O(k^5)}$ guesses.
			\item \label{FPT-guessFTPamongUandUstar}
			The fastest temporal paths from vertices of $U$ to vertices in $U^*$,
			and vice versa, the fastest temporal paths from vertices in $U^*$ to vertices in $U$.
			This gives us $k^{O(k^4)}$ guesses.
			\item \label{FPT-guessFTPamongUandZstar}
			The fastest temporal paths from vertices of $U$ to vertices in $Z^*$,
			and vice versa.
			This gives us $k^{O(k^4)}$ guesses.
			\item \label{FPT-guessFTPamongUstarandZstar}
			The fastest temporal paths from vertices of $U^*$ to vertices in $Z^*$,
			and vice versa.
			This gives us $k^{O(k^5)}$ guesses.
			\setcounter{guesscounter}{\value{enumi}}
		\end{enumerate}
		
		\begin{figure}[t]
			\centering
			\includegraphics[width=0.75\columnwidth]{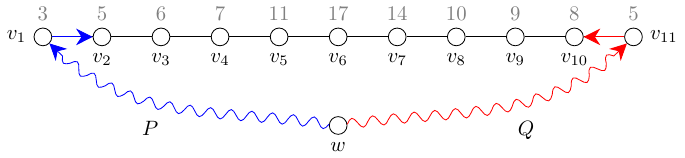}
			\caption{In the above graph vertices $v_1, v_{11}, w$ are in $U$, while $v_2, v_{10}$ are in $U^*$. 
				Numbers above all $v_i$ represent the values of the fastest temporal paths from $w$ to each of them (\ie the entries in the $w$-th row of matrix $D$).
				From the basic guesses we know the fastest temporal path $P$ from $w$ to $v_2$ (depicted in blue) and the fastest temporal path $Q$ from $w$ to $v_{10}$.
				From the values of durations from $w$ to each $v_i$ we cannot 
				determine the fastest paths from $w$ to all $v_i$.
				More precisely, we know that $w$ reaches $v_2, v_3, v_4, v_5$ (resp.~$v_{10}, v_{9}, v_{9}, v_{7}$) 
				by first using the path $P$ (resp.~$Q$) and then proceeding through the vertices,
				but we do not know how $w$ reaches $v_6$ the fastest.
				Therefore we have to introduce some more guesses.
				\label{fig:guesses-advanced}}
		\end{figure}
		
		With the information provided by the described guesses we are still not able to determine all fastest paths. For example consider the case depicted in \cref{fig:guesses-advanced}. 
		Therefore, we introduce additional guesses that provide us with sufficient information to determine all fastest paths.
		We guess the following structures.
		\begin{enumerate}[G-1.]
			\setcounter{enumi}{\value{guesscounter}}
			\item \label{FPT-guessFTPamongv2z2}
			\textbf{Inner segment guess I}.
			Let $S_{u,v} = (u=v_1,v_2, \dots, v_p = v)$ and $S_{w,z} = (w=z_1,z_2, \dots, z_r = z)$ be two segments in $G'$. 
			We want to guess the fastest temporal path
			$v_2 \rightarrow u \leadsto w \rightarrow z_2$. 
			We repeat this procedure for all pairs of segments.
			Since there are $O(k^2)$ segments in $G'$,
			there are $k^{O(k^5)}$ possible paths of this form. \\
			Recall that $S_{u,v}\neq S_{v,u}$ for every $u,v\in U$. Furthermore note that we did not assume that $\{u,v\} \cap \{w,z\} = \emptyset$. Therefore, by repeatedly making the above guesses, we also guess the following fastest temporal paths: 
			${v_2 \rightarrow u \leadsto z \rightarrow z_{r-1}}$,\ \ \ 
			${v_2 \rightarrow u \leadsto v \rightarrow v_{p-1}}$,\ \ \  
			${v_{p-1} \rightarrow v \leadsto w \rightarrow z_{2}}$,\ \ \  
			${v_{p-1} \rightarrow v \leadsto z \rightarrow z_{r-1}}$, and  
			${v_{p-1} \rightarrow v \leadsto u \rightarrow v_{2}}$.
			For an example see~\cref{fig:FPT-guessG4}.
			\item \label{FPT:guess-uToSegmentz2}
			\textbf{Inner segment guess II}.
			Let $S_{u,v} = (u=v_1,v_2, \dots, v_p = v)$ be a segment in $G'$,
			and let $w \in U \cup Z^*$.
			We want to
			guess the following fastest temporal paths
			$w \leadsto u \rightarrow v_2$, $w \leadsto v \rightarrow v_{p-1} \rightarrow \cdots \rightarrow v_2$,
			and
			$v_2 \rightarrow u \leadsto w$, $v_2 \rightarrow v_3 \rightarrow \cdots v \leadsto w$.
			\\
			For fixed $S_{u,v}$ and $w \in U \cup Z^*$ we have $k^{O(k)}$ different possible such paths, therefore we make $k^{O(k^5)}$ guesses for these paths.
			For an example see~\cref{fig:FPT-guessG5}.
			\item \label{FPT:guess-splitFromAnotherSegmentAndPaths}
			\textbf{Split vertex guess I}.
			Let $S_{u,v} = (u=v_1,v_2, \dots, v_p = v)$ be a segment in $G'$, and let us
			fix a vertex $v_i \in S_{u,v} \setminus \{u,v\}$.
			In the case when $S_{u,v}$ is of length $4$, the fixed vertex $v_i$ is the middle vertex, else we fix an arbitrary vertex $v_i \in S_{u,v} \setminus \{u,v\}$.
			Let 
			$S_{w,z} = (w=z_1,z_2, \dots, z_r = z)$ be another segment in $G'$.
			We want to determine the fastest paths from $v_i$ to all inner vertices of $S_{w,z}$. We do this by inspecting the values in matrix $D$ from $v_i$ to inner vertices of $S_{w,z}$.
			We split the analysis into two cases.
			\begin{enumerate}
				\item 
				There is a single vertex $z_j \in S_{w,z}$ for which the duration from $v_i$ is the biggest.
				More specifically, $z_j \in S_{w,z} \setminus \{w,z\}$ is the vertex with the biggest value  $D_{v_i,z_j}$.
				We call this vertex a \emph{split vertex of $v_i$ in the segment $S_{wz}$}.
				Then it holds that $D_{v_i,z_2} < D_{v_i,z_3} < \dots < D_{v_i,z_j}$ and 
				$D_{v_i,z_{r-1}} < D_{v_i,z_{r-2}} < \dots < D_{v_i,z_j}$.
				From this it follows that the fastest temporal paths from $v_i$ to $z_2, z_3, \dots, z_{j-1}$ go through $w$,
				and 
				the fastest temporal paths from $v_i$ to $z_{r-1}, z_{r-2}, \dots, z_{j+1}$ go through $z$.
				We now want to guess which vertex $w$ or $z$ is on a fastest temporal path from $v_i$ to $z_j$.
				Similarly,
				all fastest temporal paths starting at $v_i$ have to go either through $u$ or through $v$,
				which also gives us two extra guesses for the fastest temporal path from $v_i$ to $z_j$.
				Therefore, all together we have $4$ possibilities on how the fastest temporal path from $v_i$ to $z_j$ starts and ends.
				Besides that we want to guess also how the fastest temporal paths from $v_i$ to $z_{j-1}, z_{j+1}$ start and end.
				Note that one of these is the subpath of the fastest temporal path from $v_i$ to $z_j$, and the ending part is uniquely determined for both of them,
				\ie to reach $z_{j-1}$ the fastest temporal path travels through $w$, and to reach $z_{j+1}$ the fastest temporal path travels through $z$.
				Therefore we have to determine only how the path starts, namely if it travels through $u$ or $v$.
				This introduces two extra guesses.
				For a fixed $S_{u,v}, v_i$ and $S_{w,z}$ we find the vertex $z_j$ in polynomial time, 
				or determine that $z_j$ does not exist.
				We then make four guesses where we determine how the fastest temporal path from $v_i$ to $z_j$ passes through vertices $u,v$ and $w,z$ and 
				for each of them two extra guesses to determine the fastest temporal path from $v_i$ to $z_{j-1}$ and from $v_i$ to $z_{j+1}$.
				We repeat this procedure for all pairs of segments,
				which results in producing $k^{O(k^5)}$ new guesses.
				Note, $v_i \in S_{u,v}$ is fixed when calculating the split vertex for all other segments $S_{w,z}$.
				\item 
				There are two vertices $z_j, z_{j+1} \in S_{w,z}$ for which the duration from $v_i$ is the biggest.
				More specifically, $z_j, z_{j+1} \in S_{w,z} \setminus \{w,z\}$ are the vertices with the biggest value  $D_{v_i,z_j} = D_{v_i,z_{j+1}}$.
				Then it holds that $D_{v_i,z_2} < D_{v_i,z_3} < \dots < D_{v_i,z_j} = D_{v_i,z_{j+1}} > D_{v_i,z_{j + 2}} > \cdots > D_{v_i,z_{r-1}}$.
				From this it follows that the fastest temporal paths from $v_i$ to $z_2, z_3, \dots, z_{j}$ go through $w$,
				and 
				the fastest temporal paths from $v_i$ to $z_{r-1}, z_{r-2}, \dots, z_{j+1}$ go through $z$.
				In this case we only need to guess the following two fastest temporal paths 
				$v_i \leadsto w \rightarrow z_2$
				and $v_i \leadsto z \rightarrow z_{r-1}$.
				Each of these paths we then uniquely extend along the segment $S_{w,z}$ up to the vertex $z_j$, resp.~$z_{j+1}$,
				which give us fastest temporal paths from $v_i$ to $z_j$ and from $v_i$ to $z_{j+1}$.
				In this case we introduce only two more guesses.
				We repeat this procedure for all pairs of segments. which results in creating $k^{O(k^5)}$ new guesses.
			\end{enumerate}
			For an example see~\cref{fig:FPT-guessG6}.
			\item \label{FPT:guess-splitFromUtoAnotherSegment}
			\textbf{Split vertex guess II}.
			Let $w \in U \cup Z^*$,
			and let $S_{u,v} = (u=v_1,v_2, \dots, v_p = v)$ be a segment in $G'$.
			We want to guess a split vertex of $w$ in $S_{u,v}$, and the fastest temporal path that reaches it.
			We again have two cases,
			first one where $v_i$ is a unique vertex in $S_{u,v}$ that is furthest away from $w$,
			and the
			second one where $v_i, v_{i+1}$ are two incident vertices in $S_{u,v}$, that are furthest away from $w$.
			All together we make two guesses for each pair $w, S_{u,v}$.
			We repeat this for all vertices in $U \cup Z^*$, and all segments,
			which produces $k^{O(k^5)}$ new guesses.
			For an example see~\cref{fig:FPT-guessG7}.
			\setcounter{guesscounter}{\value{enumi}}
		\end{enumerate}
		\begin{figure}[t]
			\centering
			\begin{subfigure}[b]{0.48\textwidth}
				\centering
				\resizebox{0.88\linewidth}{!}{
					\includegraphics{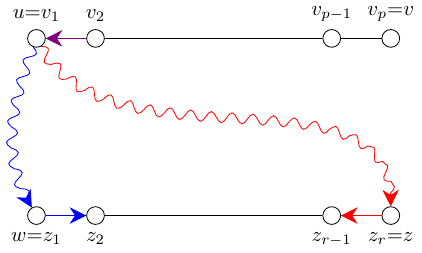}}
				\caption{Example of an Inner segment guess I (\textcolor{lipicsGray}{\textsf{\textbf{G-\ref{FPT-guessFTPamongv2z2}}}}), where we guessed the fastest temporal paths of the form $v_2 \rightarrow u \leadsto w \rightarrow z_2$ (in blue)
					and $v_2 \rightarrow u \leadsto z \rightarrow z_{r-1}$ (in red).
					\label{fig:FPT-guessG4}}
			\end{subfigure}
			\quad
			\begin{subfigure}[b]{0.48\textwidth}
				\centering
				\resizebox{0.88\linewidth}{!}{
					
					\includegraphics{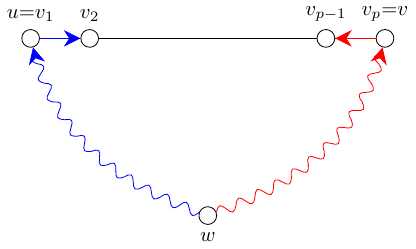}}
				\caption{Example of an Inner segment guess II (\textcolor{lipicsGray}{\textsf{\textbf{G-\ref{FPT:guess-uToSegmentz2}}}}), where we guessed the fastest temporal paths of the form $w \leadsto u \rightarrow v_2$ (in blue) and $w \leadsto v \rightarrow v_{p-1}$ (in red). 
					\label{fig:FPT-guessG5}}
			\end{subfigure}
			
			\begin{subfigure}[b]{0.48\textwidth}
				\vspace{-0.4cm}
				\centering
				\resizebox{0.98\linewidth}{!}{
					\includegraphics{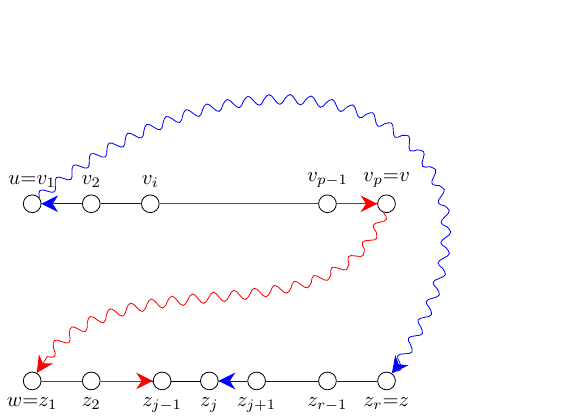}}
				\caption{Example of a Split vertex guess I (\textcolor{lipicsGray}{\textsf{\textbf{G-\ref{FPT:guess-splitFromAnotherSegmentAndPaths}}}}), where, for a fixed vertex $v_i \in S_{u,v}$,
					we calculated its corresponding split vertex $z_j \in S_{w,z}$,
					and guessed the fastest paths of the form
					$v_i \rightarrow v_{i-1} \rightarrow \cdots \rightarrow u \leadsto z \rightarrow z_{r-1} \cdots \rightarrow z_j$ (in blue) 
					and $v_i \rightarrow v_{i+1} \rightarrow \cdots \rightarrow v \leadsto w \rightarrow z_2 \rightarrow \cdots \rightarrow z_{j-1}$ (in red). 
					\label{fig:FPT-guessG6}}
			\end{subfigure}
			\quad
			\begin{subfigure}[b]{0.48\textwidth}
				\centering
				\resizebox{0.88\linewidth}{!}{
					\includegraphics{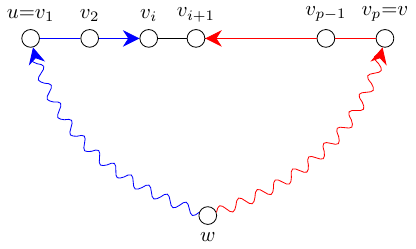}}
				\caption{Example of a Split vertex guess II (\textcolor{lipicsGray}{\textsf{\textbf{G-\ref{FPT:guess-splitFromUtoAnotherSegment}}}}), where, for a vertex of interest $w$, 
					we
					calculated its corresponding split vertex $v_i \in S_{u,v}$,
					and guessed the fastest paths of the form
					$w \leadsto u \rightarrow v_2 \rightarrow \cdots \rightarrow v_i$  (in blue) 
					and $w \leadsto v \rightarrow v_{p-1} \rightarrow \cdots \rightarrow v_{i+1}$ (in red). 
					\label{fig:FPT-guessG7}}
			\end{subfigure}
			\caption{Illustration of the guesses \textcolor{lipicsGray}{\textsf{\textbf{G-\ref{FPT-guessFTPamongv2z2}}}}, \textcolor{lipicsGray}{\textsf{\textbf{G-\ref{FPT:guess-uToSegmentz2}}}}, \textcolor{lipicsGray}{\textsf{\textbf{G-\ref{FPT:guess-splitFromAnotherSegmentAndPaths}}}}, and \textcolor{lipicsGray}{\textsf{\textbf{G-\ref{FPT:guess-splitFromUtoAnotherSegment}}}}.}
		\end{figure}
		%
		
		There are two more guesses \textcolor{lipicsGray}{\textsf{\textbf{G-11}}} and \textcolor{lipicsGray}{\textsf{\textbf{G-12}}} that we make during the creation of the ILP instances, we explain these guesses in detail in \cref{sec:constraints}. 
		We will prove that, for all guesses \textcolor{lipicsGray}{\textsf{\textbf{G-1}}} to \textcolor{lipicsGray}{\textsf{\textbf{G-12}}}, there are in total at most $f(k)$ possible choices, and for each one of them we create an ILP with at most $f(k)$ variables and at most $f(k)\cdot |D|^{O(1)}$ constraints. Each of these ILPs can be solved in FPT time by Lenstra's Theorem~\cite{Lenstra1983Integer}.
		
		\subsubsection{Properties of \deltaExact}
		
		In this section we study the properties of our problem, that then help us creating constraints of our ILP instances.
		Recall that with $G$ we denote our underlying graph of $D$. We want to determine labeling $\lambda$ of each edge of $G$.
		We start with an empty labeling of edges and try to specify each one of them.
		Note, that this does not necessarily mean that we assign numbers to the labels, but we might specify labels as variables or functions of other labels.
		We say that the label of an edge $f$ is \emph{determined with respect to} the label of the edge edge $e$, if we have determined $\lambda(f)$ as a function of $\lambda(e)$. 
		
		We first start with defining certain notions, that will be of use when solving the problem.
		
		\begin{definition}[Travel delays] \label{def:travel-delay}
			Let $(G, \lambda)$ be a temporal graph satisfying conditions of \deltaExact.
			Let $e_1=uv$ and $e_2=vz$ be two incident edges in $G$ with $e_1 \cap e_2 = v$.
			We define the \emph{travel delay} from $u$ to $z$ at vertex $v$, denoted with $\tau_v^{uz}$,
			as the difference of the labels of $e_2$ and $e_1$, where we subtract the value of the label of $e_1$ from the label of $e_2$,  modulo~$\Delta$.
			More specifically:
			\begin{equation}\label{eq:def-VertexWaitingTime}
				\tau_v^{uz} \equiv \lambda (e_2) - \lambda(e_1) \pmod \Delta.
			\end{equation}
			Similarly, $\tau_v^{zu} \equiv \lambda (e_1) - \lambda(e_2) \pmod \Delta$.
		\end{definition}
		Intuitively, the value of $\tau_v^{uz}$ represents how long a temporal path waits at vertex $v$ when first taking edge $e_1=uv$ and then edge $e_2 = vz$.
		
		From the above definition and the definition of the duration of the temporal path $P$ we get the following two observations.
		\begin{observation}\label{obs:durationPwithWaitingTimes}
			Let $P = (v_0, v_1, \dots, v_p)$ be the underlying path of the temporal path
			$(P, \lambda)$ from $v_0$ to $v_p$.
			Then $d(P,\lambda) = \sum_{i = 1}^{p-1} \tau_{v_i}^{v_{i-1}v_i} + 1 $.
		\end{observation}
		\begin{proof}
			For the simplicity of the proof denote $t_i = \lambda(v_{i-1}v_i)$, and suppose that $t_i \leq t_{i+1}$, for all $i \in \{1,2,3,\dots,p\}$.
			Then
			\begin{align*}
				\sum_{i = 1}^{p-1} \tau_{v_i}^{v_{i-1}v_i} + 1  
				&= \sum _{i = 1}^{p-1} (t_{i+1} - t_i) + 1 \\
				& = (t_2 - t_1) + (t_3 - t_2) + \cdots + (t_p - t_{p-1}) + 1  \\
				& = t_{p-1} - t_1 + 1\\
				& = d(P, \lambda)
			\end{align*}
			Now in the case when $t_i > t_{i+1}$ we get that $\tau_{v_i}^{v_{i-1}v_{i+1}} = \Delta + t_{i+1} - t_i$.
			At the end this still results in the correct duration as the last time we traverse the path $P$ is not exactly $t_p$ but $k \lambda + t_p$, for some $k$.
		\end{proof}
		We also get the following.
		\begin{observation}\label{obs:travel-delays-both-directions}
			Let $(G, \lambda)$ be a temporal graph satisfying conditions of the \deltaExact\ problem.
			For any two incident edges $e_1 = uv$ and $e_2 = vz$ on vertices $u,v,z \in V$, with $e_1 \cap e_2 = v$, we have $\tau_v^{zu} = \Delta - \tau_v^{uz} \pmod \Delta$.
		\end{observation}
		
		\begin{proof}
			Let $e_1 = uv$ and $e_2 = vz$ be two edges in $G$ for which $e_1 \cap e_2 = v$. 
			By the definition $\tau_v^{uz} \equiv \lambda (e_2) - \lambda(e_1) \pmod \Delta$ and $\tau_v^{zu} \equiv \lambda (e_1) - \lambda(e_2) \pmod \Delta$.
			Summing now both equations we get $\tau_v^{uz} + \tau_v^{zu} \equiv \lambda(e_2) - \lambda(e_1) + \lambda (e_1) - \lambda(e_2) \pmod \Delta$, and therefore $\tau_v^{uz} + \tau_v^{zu} \equiv 0 \pmod \Delta$, which is equivalent as saying $\tau_v^{uz} \equiv - \tau_v^{zu} \pmod \Delta$ or $\tau_v^{zu} = \Delta - \tau_v^{uz} \pmod \Delta$.
		\end{proof}

		In our analysis we exploit the following greatly, that is why we state is as an observation.
		
		\begin{observation}\label{obs:FirstLabelAndDuration}
			Let $P$ be the underlying path of a fastest temporal path from $u$ to $v$, where $e_1, e_p \in P$ are its first and last edge, respectively.
			Then, knowing the label $\lambda (e_1)$ of the first edge and the duration $d(P,\lambda)$ of the temporal path $(P,\lambda)$, we can uniquely determine the label $\lambda (e_p)$ of the last edge of $P$.
			Symmetrically, knowing $\lambda(e_p)$ and $d(P,\lambda)$, we can uniquely determine $\lambda(e_1)$.
		\end{observation}
		The correctness of the above statement follows directly from \cref{def:temporalPath+Duration}. This is because the duration of $(P,\lambda)$ is calculated as the difference of labels of last and first edge plus $1$,
		where the label of last edge is considered with some delta periods,
		\ie $d(P,\lambda) = p_i \Delta + \lambda(e_p) - \lambda (e_1) + 1$, for some $p_i \geq 0$.
		Therefore $d(P,\lambda) \pmod \Delta \equiv  (\lambda(e_p) - \lambda (e_1) + 1) \pmod \Delta$.
		Note that if $\lambda(e_1)$ and $\lambda(e_p)$ are both unknown, then we can determine one with respect to the other.
		
		In the following we prove that knowing the structure (the underlying path) of a fastest temporal path $P$ from a vertex of interest $u$ to a vertex of interest $v$,
		results in determining the labeling of each edge in the fastest temporal path from $u$ to $v$ 
		(with the exception of some constant number of edges), with respect to the label of the first edge.
		More precisely, if path $P$ from $u$ to $v$ is a segment, then we can determine labels of all edges as a function of the label of the first edge.
		If $P$ consists of $\ell$ segments, then we can determine the labels of all but $\ell -1$ edges as a function of the label of the first edge.
		For the exact formulation and proofs see \cref{lemma:FPT-uv-Labelalledges,lemma:FPT-uv-LabelAlmostalledges}.
		
		\begin{lemma}\label{lemma:FPT-uv-Labelalledges}
			Let $u, v \in U$ be two arbitrary vertices of interest and suppose that $P = (u=v_1,v_2, \dots, v_p = v)$, where $p \geq 2$, 
			is a path in $G'$, which is also the underlying path of a fastest temporal path from $u$ to $v$.
			Moreover suppose also that $P$ is a segment.
			We can determine the labeling $\lambda$ of every edge in $P$ with respect to the label $\lambda(uv_2)$ of the first edge.
		\end{lemma}
		
		\begin{proof}
			We claim that $u$ reaches all of the vertices in $P$ the fastest, when traveling along $P$ (\ie by using a subpath of $P$).
			To prove this suppose for the contradiction that there is a vertex $v_i \in P \setminus \{u,v\}$, that is reached from $v$ on a path different than $P_i = (u, v_2, v_3, \dots, v_i)$ faster than through $P_i$.
			Since the only vertices of interest of $P$ are $u$ and $v$, it follows that all other vertices on $P$ are of degree $2$. 
			Then the only way to reach $v_i$ from $u$, that differs from $P$, would be to go from $u$ to $v$ using a different path $P_2$,
			and then go from $v$ to $v_{p-1}, v_{p-2}, \dots, v_i$.
			But since $P$ is the fastest temporal path from $u$ to $v$, we get that $d(P_2) \geq d(P)$ and $d(P_2 \cup (v,v_{p-1}, \dots, v_i)) > d(P) > d(P_i)$.
			
			Now, to determine the labeling $\lambda$ of the path $P$ we use the property that the fastest temporal path from $u$ to any $v_i \in P$ is a subpath of $P$. 
			We set the label of the first edge of $P$ to be a constant $c\in [\Delta]$
			and
			use \cref{obs:FirstLabelAndDuration} to label all remaining edges,
			where the duration from $u$ to $v_i$ equals to $D_{u,v_i}$.
			This gives us a unique labeling $\lambda$ of $P$ where the label of each edge of $P$ is a function of $c$.
		\end{proof}
		
		\begin{lemma}\label{lemma:FPT-uv-LabelAlmostalledges}
			Let $u, v \in U$ be two arbitrary vertices of interest and suppose that $P = (u=v_1, v_2, \dots, v_p = v)$, where $p \geq 2$, 
			is a path in $G'$, which is also the underlying path of a fastest temporal path from $u$ to $v$.
			Let $\ell_{u,v} \geq 1$ be the number of vertices of interest in $P$ different to $u,v$, namely $\ell_{u,v} = |\{P \setminus \{u,v \} \} \cap U |$.
			We can determine the labeling $\lambda$ of all but $\ell_{u,v}$ edges of $P$, with respect to the label $\lambda(uv_2)$ of the first edge,
			such that the labeling $\lambda$ respects the values from $D$.
		\end{lemma}
		
		For the proof of the above lemma, we first prove a weaker statement,
		for which we need to introduce some extra definitions and fix some notations.
		In the following
		we only consider \emph{wasteless} temporal paths. 
		We call 
		a temporal path $P=((e_{1},t_{1}),\ldots ,(e_{k},t_{k}))$ a \emph{wasteless} temporal path, 
		if for every $i=1,2,\ldots ,k-1$, we have that $t_{i+1}$ is the first time after $t_{i}$ that the edge $e_{i+1}$ appears. 
		
		Let $u,v\in V$, and let $t\in \mathbb{N}$. 
		Given that a temporal path starts within the period $[t,t+\Delta -1]$,
		we denote with $A_{t}(u,v)$
		the \emph{arrival} of the fastest path in $(G,\lambda )$ from $u$ to $v$, and
		with $A_{t}(u,v,P)$,
		the \emph{arrival} along path $P$ in $(G,\lambda )$ from $u$ to $v$.
		Whenever $t=1$, we may omit the index $t$, \ie we may write 
		$A(u,v,P)=A_{1}(u,v,P)$ and $A(u,v)=A_{1}(u,v)$. 
		
		Suppose now that 
		we know the underlying path $P_{u,v} = (u=v_1, v_2, \dots, v_p = v)$ of the fastest temporal path between vertices of interest $u$ and $v$ in $G'$.
		Let $v_i\in U$ with $u\neq v_i \neq v$ be a vertex of interest on the path $P_{u,v}$.
		Suppose that $v_i$ is reached the fastest from $u$ by a path $P = (u = u_1, u_2, \dots, u_{j-1}, v_i)$.
		We split the path with $P_{u,v}$ into a path $ Q = (u=v_1, v_2, \dots, v_i)$ and $R = (v_i, v_{i+1}, \dots, v_p=v)$
		(for details see~\cref{fig:FPT-uv-Labelalledges}).
		
		From the above we get the following assumptions:
		\begin{enumerate}
			\item \label{assumption1-FPT} $d(u,v_{i})=d(u,v_{i},P)\leq d(u,v_{i},Q)$, and
			\item \label{assumption2-FPT} $d(u,v_p)=d(u,v_p,Q\cup R)\leq d(u,v_p,P\cup R)$.
		\end{enumerate}
		In the remainder, we denote with $\delta_{0}$ the difference $d(u,v_{i},Q) - d(u,v_{i})\geq 0$.
		%
		Let $t_{v_2}\in [\Delta]$ be the label of the edge $uv_2$,
		and denote by $t_{u_{2}}$ the appearance of the edge $uu_{2}$ within the
		period $[t_{v_2},t_{v_2}+\Delta -1]$. Note that $1\leq t_{v_2}\leq
		\Delta $ and that 
		\begin{equation}
			t_{v_{2}} \leq t_{u_{2}} \leq t_{v_{2}}+\Delta -1 \leq 2\Delta -1.
			\label{basic-ineq-0}
		\end{equation}
		From Assumption~\ref{assumption1-FPT} we get
		\begin{equation*}
			\delta
			_{0}=d(u,v_{i},Q)-d(u,v_{i})=A_{t_{v_2}}(u,v_{i},Q)-A_{t_{v_2}}(u,v_{i},P)+\left( t_{u_{2}}-t_{v_2}\right)
		\end{equation*}%
		and thus%
		\begin{equation}
			A_{t_{v_2}}(u,v_{i},P)-A_{t_{v_2}}(u,v_{i},Q)=t_{u_{2}}-(t_{v_2}+\delta _{0}).
			\label{basic-eq-1}
		\end{equation}
		
		\begin{figure}[t]
			\centering
			\includegraphics{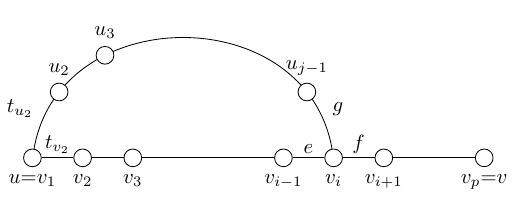}
			\caption{An example of the situation in \cref{lemma:FPT-uv-LabelAlmostalledges},
				where we assume that the fastest temporal path from $u$ to $v$ is $P_{u,v} = (u=v_1, v_2, \dots v_p)$,
				and the fastest temporal path from $u$ to $v_i$ in $P_{u,v}$ is $P = (u, u_2, u_3, \dots, v_i)$.
				We denote with $Q = (u=v_1, v_2, \dots, v_i)$ and with $R = (v_i, v_{i+1}, \dots, v_p=v)$.
				\label{fig:FPT-uv-Labelalledges}}
		\end{figure}
		
		We use all of the above discussion, to prove the following lemma.
		\begin{lemma}
			\label{lem-1}If $t_{u_{2}}\neq t_{v_2}$, then $\delta _{0}\leq \Delta -2$
			and $t_{u_{2}}\geq t_{v_2}+\delta _{0}+1$.
		\end{lemma}
		
		\begin{proof}
			First assume that $\delta _{0}\geq \Delta -1$. Then, it follows by \cref%
			{basic-eq-1} that $%
			A_{t_{v_2}}(u,v_{i},P)-A_{t_{v_2}}(u,v_{i},Q)\leq
			t_{u_{2}}-t_{v_2}-\Delta +1\leq 0$, and thus $A_{t_{v_2}}(u,v_{i},P)%
			\leq A_{t_{v_2}}(u,v_{i},Q)$. Therefore, since we can traverse path $P$
			from $u$ to $v_{i}$ by departing at time $t_{u_{2}}\geq t_{v_2}+1$ and
			by arriving no later than traversing path $Q$, we have that $d(u,v_p,P\cup Q)<d(u,v_p,Q\cup R)$, which is a contradiction to the
			second initial assumption. Therefore $\delta _{0}\leq \Delta -2$.
			
			Now assume that $t_{v_2}+1\leq t_{u_{2}}\leq t_{v_2}+\delta _{0}$. Then,
			it follows by \cref{basic-eq-1} that $A_{t_{v_2}}(u,v_{i},P)\leq
			A_{t_{v_2}}(u,v_{i},Q)$ which is, similarly to the previous case, a
			contradiction. Therefore $t_{u_{2}}\geq t_{v_2}+\delta _{0}+1$.
		\end{proof}
		
		The next corollary follows immediately from \cref{lem-1}.

		\begin{corollary}
			\label{cor-1}If $t_{u_{2}}\neq t_{v_2}$, then $1\leq
			A_{t_{v_2}}(u,v_{i},P)-A_{t_{v_2}}(u,v_{i},Q)\leq \Delta
			-1-\delta _{0}$.
		\end{corollary}
		
		We are now ready to prove the following result.
		
		\begin{lemma}\label{lemma:FPT-twoPathsOneUnlabeledEdge}
			\label{lem-2}$d(u,v_{i-1},P\cup \{v_{i}v_{i-1}\})>d(u,v_{i-1},Q\setminus
			\{v_{i}v_{i-1}\})$.
		\end{lemma}
		\begin{proof}
			Let $e\in [\Delta]$ be the label of the edge $v_{i-1}v_{i}$, and
			let $f\in \lbrack e+1,e+\Delta ]$ be the time of the first appearance of the
			edge $v_{i}v_{i+1}$ after time $e$. Let $A_{t_{v_{2}}}(u,v_{i},Q)=x%
			\Delta +e$. Then $A_{t_{v_{2}}}(u,v_{i+1},Q\cup\{v_{i}v_{i+1}\})=x\Delta +f$. Furthermore let $g$ be such
			that $A_{t_{v_{2}}}(u,v_{i},P)=x\Delta +g$. 
			
			\emph{Case 1: }$t_{u_{2}}\neq t_{v_2}$. Then \cref{cor-1}
			implies that $e+1\leq g\leq e+(\Delta -1-\delta _{0})$. Assume that $g<f$.
			Then, we can traverse path $P$ from $u$ to $v_{i}$ by departing at time $%
			t_{u_{2}}\geq t_{v_2}+1$ and by arriving at most at time $x\Delta +f-1$,
			and thus $d(u,v_p,P\cup R)<d(u,v_p,Q\cup R)$, which is a
			contradiction to the second initial assumption. Therefore $g\geq f$. That is,%
			\begin{equation*}
				e+1\leq f\leq g\leq e+(\Delta -1-\delta _{0}).  
			\end{equation*}
			
			Consider the path $P^{\ast }=P\cup \{v_{i}v_{i-1}\}$. Assume that we start
			traversing $P^{\ast }$ at time $t_{u_{2}}$. Then we arrive at $v_{i}$ at
			time $x\Delta +g$, and we continue by traversing edge $v_{i}v_{i-1}$ at time 
			$(x+1)\Delta +e$. That is, $d(u,v_{i-1},P^{\ast })=(x+1)\Delta
			+e-t_{u_{2}}+1$. 
			
			Now consider the path $Q^{\ast }=Q\setminus \{v_{i}v_{i-1}\}$. Let $h\in
			\lbrack 1,\Delta ]$ be such that $A_{t_{v_{2}}}(u,v_{i-1},Q^{\ast
			})=x\Delta +e-h$. That is, if we start traversing $Q^{\ast }$
			at time $t_{v_2}$, we arrive at $v_{i-1}$ at time $x\Delta +e-h$, \ie $%
			d(u,v_{i-1},Q^{\ast })=x\Delta +e-h-t_{v_2}+1$. Summarizing, we have:%
			\begin{eqnarray*}
				d(u,v_{i-1},P^{\ast })-d(u,v_{i-1},Q^{\ast }) &=&\Delta +h-(t_{u_{2}}-t_{v_2}) \\
				&\geq & \Delta +1-(\Delta-1)>0,
			\end{eqnarray*}%
			where the last inequality follows by (\ref{basic-ineq-0}). This proves the statement of the lemma.
			
			\emph{Case 2: }$t_{u_{2}}=t_{v_2}$. Then, it follows by Equation~(\ref%
			{basic-eq-1}) that $%
			A_{t_{v_2}}(u,v_{i},P)=A_{t_{v_2}}(u,v_{i},Q)-\delta _{0}\leq
			A_{t_{v_2}}(u,v_{i},Q)$. Therefore $g\leq e$. Similarly to Case 1
			above, consider the paths $P^{\ast }=P\cup \{v_{i}v_{i-1}\}$ and $Q^{\ast
			}=Q\setminus \{v_{i}v_{i-1}\}$. Assume that we start traversing $P^{\ast }$
			at time $t_{u_{2}}=t_{v_2}$. Then we arrive at $v_{i}$ at time $x\Delta +g$%
			, and we continue by traversing edge $v_{i}v_{i-1}$, either at time $%
			(x+1)\Delta +e$ (in the case where $g=e$) or at time $x\Delta +e$ (in the
			case where $g\neq e$). That is, $d(u,v_{i-1},P^{\ast })\geq x\Delta
			+e-t_{v_2}+1$.
			
			Similarly to Case 1, let $h\in \lbrack 1,\Delta ]$ be such that $%
			A_{t_{v_{2}}}(u,v_{i-1},Q^{\ast })=x\Delta +e-h$. That is, if we start
			traversing $Q^{\ast }$ at time $t_{v_{2}}$, we arrive at $v_{i-1}$ at time $%
			x\Delta +e-h$, \ie $d(u,v_{i-1},Q^{\ast })=x\Delta +e-h-t_{v_{1}}+1$.
			Summarizing, we have:%
			\begin{equation*}
				d(u,v_{i-1},P^{\ast })-d(u,v_{i-1},Q^{\ast })\geq h\geq 1,
			\end{equation*}%
			which proves the statement of the lemma.
		\end{proof}
		
		From the above it follows that if $P$ is a fastest path from $u$ to $v$,
		then all vertices of $P$, with the exception of vertices of interest $v_i \in P\setminus \{u,v\}$,
		are reached using the same path $P$.
		We use this fact in the following proof.
		
		\begin{proof}[Proof of \cref{lemma:FPT-uv-LabelAlmostalledges}]
			For every vertex of interest $v_i \in U \cap (P \setminus \{u,v\})$
			we have two options.
			First, when the fastest temporal path $P'$ from $u$ to $v_i$ is
			a subpath of $P$.
			In this case we determine the labeling of $P'$ using \cref{lemma:FPT-uv-Labelalledges}.
			Second, when the fastest temporal path $P'$ from $u$ to $v_i$ is 
			not a subpath of $P$. 
			In this case we know exactly how to label all of the edges of $P$,
			with the exception of edges of from $v_{i-1}v_i$, that are incident to $v_i$ in $P$.
		\end{proof}

		\begin{lemma}\label{lemma:FPT-noUndeterminedEdgeInSegment}
			Let $S_{u,v}= (u=v_1,v_2, \dots, v_p = v)$ be a segment in $G$.
			If $S_{u,v}$ is of length at least $5$ ($p > 5$) 
			then it is impossible for an inner edge $f = v_i v_{i+1}$ from $S_{u,v} \setminus \{u,v\}$ (where $f$ is an edge that is not incident to a vertex from $U$)
			to not be a part of any fastest temporal path, of length at least $2$ between vertices in $S_{u,v}$.
			In other words, there must exist a pair $v_j, v_{j'} \in S_{u,v}$ \st the fastest temporal path from $v_j$ to $v_{j'}$ passes through $f$.
			If $S_{u,v}$ is of length $4$ then all temporal paths of length $2$ avoid the inner edge $f$ if and only if $f$ has the same label as both of the edges incident to it, while the label of the last remaining edge is determined with respect to $\lambda(f)$.
		\end{lemma}
		
		\begin{proof}
			For an easier understanding and better readability, we present the proof for $S_{u,v}$ of length $5$.
			The case where $S_{u,v}$ is longer easily follows from the presented results.
			
			Let $S_{u,v} = (u=v_1,v_2, v_3, v_4, v_5, v_6=v)$.
			We distinguish two cases, first when $f = v_2v_3$ (note that the case with $f = v_4v_5$ is symmetrical),
			and the second when  $f = v_3v_4$.
			Throughout the proof we denote with $t_i$ the label of edge $v_i v_{i+1}$.
			Suppose for the contradiction, that none of the fastest temporal paths between vertices of $S_{u,v}$ traverses the edge $f$.
			
			\emph{Case 1: }$f = v_2v_3$.
			Let us observe the case of the fastest temporal paths between $v_{1}$ and $v_{3}$.
			Denote with $Q = (v_{1}, v_2, v_3)$ and with $P' = (v_3,v_4,v_5,v_6)$.
			From our proposition, it follows that
			\begin{itemize}
				\item the fastest temporal path $P ^ +$ from $v_1$ to $v_3$ 
				is of the following form 
				$P^+ = v_{1}  \leadsto v_6 \rightarrow v_5 \rightarrow v_4 \rightarrow v_3$,
				and
				\item the fastest temporal path $P ^ -$ from $v_{3}$ to $v_{1}$ 
				is of the following form 
				$P^- = v_{3} \rightarrow v_{4} \rightarrow v_5 \rightarrow v_6 \leadsto v_1$.
			\end{itemize}
			It follows that 
			$
			d(v_{1}, v_{3}, P^+) \leq d(v_{1}, v_{3}, Q),
			$
			and
			$
			d(v_{1}, v_{3}, P^-) \leq d(v_{1}, v_{3}, Q)
			$. 
			Note that $d(v_{1}, v_{3}, P^+) \geq 1 + d(v_6,v_3,P')$,
			and by the definition $d(v_6,v_3,P') = 1 + (t_4 - t_5)_\Delta + (t_3 - t_4)_\Delta$,
			where $(t_i - t_j)_\Delta$ denotes the difference of two consecutive labels $t_i, t_j$ modulo $\Delta$.
			Similarly holds for $d(v_{1}, v_{3}, P^+)$.
			Summing now both of the above equations we get
			\begin{equation}
				\begin{split} \label{eq:FPT-prf-fUnlabeled}
					d(v_{1}, v_{3}, P^+) + d(v_{3}, v_{1}, P^-) &\leq 
					d(v_{1}, v_{3}, Q) + d(v_{3}, v_{1}, Q) \\
					1 + d(v_6,v_3,P') + 1 + d(v_3,v_6,P') &\leq d(v_{1}, v_{3}, Q) + d(v_{3}, v_{1}, Q) \\
					3 + (t_4 - t_5)_\Delta + (t_3 - t_4)_\Delta +
					1 + (t_4 - t_3)_\Delta + (t_5 - t_4)_\Delta 
					&\leq 
					1 + (t_2 - t_1)_\Delta + 
					1 + (t_1 - t_2)_\Delta\\
					(t_4 - t_5)_\Delta + (t_5 - t_4)_\Delta +
					(t_4 - t_3)_\Delta + (t_3 - t_4)_\Delta  + 2 
					&\leq 
					(t_2 - t_1)_\Delta + (t_1 - t_2)_\Delta.
				\end{split}
			\end{equation}
			Note that if $t_i \neq t_j$ we get that 
			the sum
			$(t_i - t_j)_\Delta + (t_j - t_i)_\Delta$ equals exactly $\Delta$,
			and if $t_i = t_j$ the sum equals $2\Delta$.
			This follows from the definition of travel delays at vertices (see \cref{obs:travel-delays-both-directions}).
			Therefore we get from \cref{eq:FPT-prf-fUnlabeled}, 
			that the right part is at most $2 \Delta$, while the left part is at least $2 \Delta + 1$,
			for any relation of labels $t_1,t_2, \dots, t_5$,
			which is a contradiction.
			
			\emph{Case 2: }$f = v_3v_4$.
			Here we consider the fastest paths between vertices $v_{2}$ and $v_{4}$.
			By similar arguments as above we get
			\begin{align*}
				(t_5 - t_1)_\Delta + (t_4 - t_5)_\Delta + (t_5 - t_4)_\Delta + (t_1 - t_5)_\Delta + 2 \leq (t_3 - t_2)_\Delta + (t_2 - t_3)_\Delta,
			\end{align*}
			which is impossible.
			
			In the case when $S_{u,v}$ is longer, we would get even bigger number on the left hand side of \cref{eq:FPT-prf-fUnlabeled}, 
			so we conclude that in all of the above cases, it cannot happen that all fastest paths of length $2$, between vertices in $S_{u,v}$, avoid edge $f$.
			
			Let us observe now the case when $S_{u,v} = (u=v_1,v_2, v_3, v_4, v_5=v)$ is of length $4$.
			Let $f = v_2 v_3$ (the case with $f = v_3 v_4$ is symmetrical).
			Suppose that the fastest temporal paths between $v_1$ and $v_3$ do not use the edge $f$.
			We denote with $R^+$ the fastest path from $v_1$ to $v_3$, 
			which is of the form $u \leadsto v \rightarrow v_4 \rightarrow v_3$,
			and similarly 
			with $R^-$ the fastest path from $v_3$ to $v_1$, which is
			of the form $v_3 \rightarrow v_4 \rightarrow v \leadsto u$.
			We denote with $R' = (v_3, v_4, v_5)$ and with $S = (v_1,v_2,v_3)$.
			Again we get the following.
			\begin{equation*}
				\begin{split}
					d(v_{1}, v_{3}, R^+) + d(v_{3}, v_{1}, R^-) &\leq 
					d(v_{1}, v_{3}, S) + d(v_{3}, v_{1}, S) \\
					1 + d(v_5,v_3,R') + 1 + d(v_3,v_5,R') &\leq d(v_{1}, v_{3}, S) + d(v_{3}, v_{1}, S) \\
					(t_3 - t_4)_\Delta + (t_4 - t_3)_\Delta + 2 
					&\leq 
					(t_2 - t_1)_\Delta + (t_1 - t_2)_\Delta.        
				\end{split}
			\end{equation*}
			The only case when the equation has a valid solution is when $t_1 = t_2$ and $t_3 \neq t_4$,
			as in this case the left hand side evaluates to $\Delta + 2$, while the right side evaluates to $2 \Delta$.
			Repeating the analysis for the fastest paths between $v_2$ and $v_4$,
			we conclude that the only valid solution is when $t_2 = t_3$ and $t_1 \neq t_4$.
			Altogether, we get that $f$ is not a part of any fastest path of length $2$ in $S_{u,v}$ if and only if the label of edge $f$ is the same as the labels on the edges incident to it, while the last remaining edge has a different label.
			Note now that the fastest temporal path from $v_2$ to $v_4$ must first use the edge $u v_2$ and finish with the edge $v_4 v_5$, and it has to be of duration $D_{v_2, v_4}$. Using \cref{lemma:FPT-uv-Labelalledges} we determine the label of the edge $v_4v_5$ with respect to $\lambda(f)$.
		\end{proof}
		
		We now present some properties involving the vertices from $Z$, that form the trees in~$G[Z]$.
		\begin{lemma}\label{lemma:ftp-tree-neighbor}
			Let $v \in V(G')$ be a clip vertex of the tree $T_v$ in $G[Z \cup \{v\}]$, and let $z \in N_{T_v}(v)$ be an arbitrary child of $v$ in $T_v$. 
			Among all neighbors of $v$ in $G'$, let $w$ be the one that is on the smallest duration away from $z$ with respect to the values of $D$. 
			In other words, $w \in N_{G'}(v)$ such that $D_{z,w} \leq D_{z,w'}$ for all $w' \in N_{G'}(v)$.
			Then, the path $P^*=(z,v,w)$ represents the unique fastest temporal path from $z$ to $w$.
			Moreover, we can determine all labels of the tree $T_v$ with respect to the label $\lambda(vw)$.
		\end{lemma}
		\begin{proof}
			Suppose for contradiction that there exists a path $P^{**} \neq P^*$ from $z$ to $w$ such that $d(P^{**},\lambda) \leq d(P^*,\lambda)$. 
			By the structure of $G$, it follows that $P^{**}$ passes through the clip vertex $v$ of $T_v$ (as this is the only neighbor of $z$ in $G'$), continues through a vertex $w' \in N_{G'}(v) \setminus \{w \}$, 
			and through some other vertices $u_1, u_2, \dots, u_j$ in $G$ ($j \geq 0$) before finishing in $w$.
			Therefore, $P^{**}= (z,v,w',u_1,u_2,\dots,u_j,w)$.
			Now, since $D_{z,w} \leq D_{z,w'}$ by assumption, the first part of $P^{**}$ from $z$ to $w'$ takes at least $D_{z,w'}$ time, and thus it takes at least $D_{z,w}$ time.
			Since $w \neq w'$, we need at least one more time-step (one more edge) to traverse from $w'$ to reach $w$. 
			Therefore, $d(P^{**}, \lambda) \geq D_{z,w} + 1$ 
			which implies that $P^{**}$ is not the fastest temporal path from $z$ to $w$. Therefore, the only fastest temporal path from $z$ to $w$ is $P^*=(z,v,w)$.
			
			For the second part, knowing that the duration of $P^*$ is $D_{z,w}$, we can determine the label of the edge $z v$ with respect to the label $\lambda(vw)$ (see \cref{obs:FirstLabelAndDuration}).
			Furthermore, using the algorithm for trees (see \cref{thm:deltaExact-PolyTimeTrees}), we can now determine all the labels on the edges of $T_v$ with respect to the same label $\lambda(v w)$.
		\end{proof}
		
		\begin{lemma} \label{lemma:treeLabels-notInteresting}
			Let $x \in V(G')$ be a clip vertex of the tree $T_x$ in $G[Z \cup \{x\}]$, where $x \notin U$. 
			Let $v_1$ and $v_2$ be the two neighbors of $x$ in $G'$.
			Then the labels of the tree $T_x$ can be determined with respect to $\lambda(v_1 x)$ and $\lambda(x v_2)$.
		\end{lemma}
		\begin{proof}
			First observe that since $x$ is not a vertex of interest it must be a part of some segment $S_{u,w}$, where $u,v \in U$ and $x \neq u \neq v$. Therefore, $x$ is of degree $2$ in $G'$.
			Let $z \in V(T_x)$ be a child of $x$ in $T_x$, \ie a vertex in the first layer of the tree $T_x$.
			We observe the values $D_{z,v_1}, D_{z,v_2}$ and distinguish the following cases.
			
			First, $D_{z,v_1} = D_{z,v_2}$
			Then, using \cref{lemma:ftp-tree-neighbor} we conclude that the fastest temporal paths from $z$ to $v_1$ and from $z$ to $v_2$ are of length two.
			We know that these two paths consist of the edge $zx$ and $xv_1,xv_2$, respectively.
			This allows us to determine the label of the edge $zx$ (and consequently all other edges of $T_x$) with respect to $\lambda(xv_1)$ and $\lambda(xv_2)$.
			
			Second, $D_{z,v_1} \neq D_{z,v_2}$. 
			Let us denote with $t_1 = \lambda(x v_1), t_2 = \lambda(x v_2)$ and $t_3 = \lambda(zx)$.
			W.l.o.g. suppose that $\min \{t_1,t_2,t_3\} = t_3$, and that $D_{z,v_1} > D_{z,v_2}$ (the other case is analogous). It follows that $t_1 > t_2$.
			We want to now prove that the inequality $D_{v_1,z} < D_{v_2,z}$ holds.
			Suppose for the contradiction that the inequality is false. Then $D_{v_2,z} < D_{v_1,z} \leq (\Delta + t_3 - t_1)$. This implies that the fastest temporal path from $v_2$ to $z$ cannot use the path $(v_1,x,z)$, and is therefore of form $(v_2,x,z)$.
			By the definition, the duration of this path is $D_{v_2,z} = \Delta + t_3 - t_2 + 1$.
			But since $t_1 > t_2$ it follows that $(\Delta + t_3 - t_2) + 1 > (\Delta + t_3 - t_1) + 1$.  We also know that $D_{v_1,z} \leq  (\Delta + t_3 - t_1) + 1$. 
			This implies that $D_{v_2,z} > D_{v_1,z} $, a contradiction.\\
			Knowing $D_{z,v_1} > D_{z,v_2}$ we can determine the label of edge $zx$ (and consequently all other edges of $T_x$) with respect to $\lambda(x v_2)$,
			and similarly knowing $D_{v_1,z} < D_{v_2,z}$ we determine the label of edge $zx$ (and all other edges of $T_x$) with respect to $\lambda(x v_1)$.
		\end{proof}
		
		Remember, in the case where the clip vertex $u$ of the tree $T_u$ in $G[Z \cup \{u\}]$ is a vertex of interest, we split the vertices in the first layer of $T_u$ into at most $|N_{G'}(u)|$ equivalence classes (as explained in \cref{sec:preprocessing-FPT}).
		Let us now show the following important property of these equivalence classes.
		\begin{lemma}\label{lemma:treeEquivClass-intVertices}
			Let $u \in V(G')$ be a clip vertex of the tree $T_u$ in $G[Z \cup \{u\}]$, where $u \in U$,
			and let $z_1,z_2 \in V(T_u)$ be in the same equivalence class of the tree $T_u$.
			Then, the fastest temporal paths from $z_1$ and from $z_2$ to any other vertex in $G'$ coincide on the edges in $G'$.
			Similarly, the fastest temporal paths from any other vertex in $G'$ to $z_1$ and to $z_2$ coincide on the edges in $G'$.
		\end{lemma}
		\begin{proof}
			Let $y \neq u$ be a vertex in $V(G')$.
			Denote with $P_1$ the underlying path of the fastest temporal path from $z_1$ to $y$, which consists of the edge $z_1 u$ and the path $P$ from $u$ to $y$.
			Similarly, let $Q_2$ be the underlying path of the fastest temporal path from $z_2$ to $y$, consisting of the edge $z_2u$ together with the path $Q$ from $u$ to $y$.
			Define $P_2$ as the second path from $z_2$ to $y$ that first uses the edge $z_1u$ and then the path $P$.
			Similarly, $Q_1$ represents the second path from $z_1$ to $y$ that first uses the edge $z_2u$ and then the path $Q$.
			Our objective is to demonstrate that either $P = Q$ or that $d(P_1, \lambda) = d(Q_1, \lambda)$ (and $d(P_2, \lambda) = d(Q_2, \lambda)$).
			This implies that $z_1$ and $z_2$ use temporal paths that coincide on the vertices of $V(G) \setminus V(T_u)$ to reach $y$. For an illustration see \cref{fig:FPT-equivClass}.
			
			\begin{figure}[t]
				\centering
				\includegraphics[width=0.6\columnwidth]{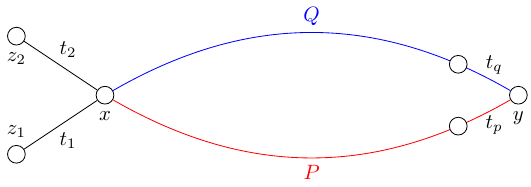}
				\caption{An example of the situation in \cref{lemma:treeEquivClass-intVertices}.
					We have two paths $Q$ (in blue) and $P$ (in red) from $x$ to $y$.
					We assume that the fastest temporal path from $z_1$ to $y$ is $P_1 = \{z_1 u\} \cup P$,
					and the fastest temporal path from $z_2$ to $y$ is $Q_2 = \{z_2 x\} \cup Q$.
					We denote with $Q_1 = \{z_1 u\} \cup Q$ another temporal path from $z_1$ to $y$, 
					with $P_2 = \{z_2 u\} \cup P$ another temporal path from $z_2$ to $y$.
					The labels on the edges $z_1 x, z_2 x$ are $t_1,t_2$, respectively. Similar, the labels of the last edges of paths $P$ and $Q$ are  $t_p, t_q$, respectively.
					\label{fig:FPT-equivClass}}
			\end{figure}
			
			Let us set the label of the edge $z_1u$ to $t_1$, the label of $z_2u$ to $t_2$, the label of the last edge of the path $P$ as $t_p$ and the label of the last edge of the path $Q$ as $t_q$.
			By the definition, since $P_1$ represents the fastest temporal path form $z_1$ to $y$ we get that
			$D_{z_1,y} = t_p - t_1 + c_p \Delta$, where $c_p \in \mathbb{N}$.
			Similarly, for the path $Q_2$ it holds that $D_{z_2,y} = t_q - t_2 + c_q \Delta$ with $c_q \in \mathbb{N}$.
			Note that the difference between the first label of $P$ (resp.~$Q$) with $t_1$ and $t_2$ is smaller than $\Delta$,
			or the difference (with at least one $t_1,t_2$) is $\Delta$ if and only if the first label of $P$ and the first label of $Q$ are the same. 
			This observation is crucial in our arguing below.
			
			We want to first show that $c_p = c_q$.
			Let us assume, for the sake of contradiction, that this is not the case, and suppose that $c_p > c_q$ (the case with $c_q > c_p$ is analogous).
			Then $c_p \leq c_q + 1$.
			Now, since $z_1$ and $z_2$ are in the same equivalence class and by the definition of the duration of a temporal path we get that $d(P_2, \lambda) = t_p - t_2 + c_p \Delta \leq t_p - t_2 + (c_q - 1)\Delta$. 
			Because $Q_2$ is the fastest path from $z_2$ to $y$ we have also that $d(P_2, \lambda) \geq D_{z_2,y}$, which gives us $t_p - t_2 + (c_q - 1) \Delta \geq t_q - t_2 + c_q \Delta$.
			This is equivalent to $t_p \geq t_q + \Delta$, but since $t_p, t_q \in \Delta$ this cannot happen.
			Therefore, we conclude that $c_p = c_q$.
			
			Now, we want to show also, that $t_p = t_q$.
			Let us assume, for the sake of contradiction, that this is not the case, and suppose that $t_p > t_q$ (the case with $t_q > t_p$ is analogous).
			Then the duration of the path $Q_1$ is $d(Q_1,\lambda) = t_q - t_1 + c_q \Delta$ since $c_q = c_p$.
			Above we proved that $c_p = c_q$. We also know that $d(Q_1,\lambda) \geq d(P_1,\lambda)$ as $P_1$ is the fastest path from $z_1$ to $y$. All of this results in $t_q - t_1 + c_p \Delta \geq t_p - t_1 + c_p \Delta$ implying $t_q \geq t_p$, a contradiction.
			Therefore, $t_p = t_q$.
			
			We proved that either $P$ and $Q$ are the same, or if they are different then $P_1$ and $Q_1$ are of the same duration and are both fastest paths from $z_1$ to $y$ (the same holds for $z_2$). 
			
			Proof for the fastest temporal paths in the other direction, namely starting at $y$ and reaching $z_1$ and $z_2$, is done analogously.
		\end{proof}
		
		\begin{observation}\label{obs:tree-only-neighbors}
			Let $v \in V(G')$ be a clip vertex of the tree $T_v$ in $G[Z \cup \{v\}]$,
			$z \in N_{T_v}(v)$ be a child of $v$ in $T_v$, and let $z'$ be a descendant of $z$ in $T_v$.
			Let $x \in V(G)\setminus V(T_v)$ be an arbitrary vertex. 
			Denote by $P_z$ and $P_{z'}$ the underlying paths of the fastest temporal paths from $z$ and from $z'$ to $x$, respectively,
			and denote by $Q$ the (unique) path between $z$ and $z'$ in $T_v$.
			Then $P_z$ and $P_{z'}$ differ only in the edges of $Q$.
		\end{observation}
		The correctness of the above observation is a consequence of \cref{lemma:treeEquivClass-intVertices} and of the fact that $P_z$ and $P_{z'}$ leave the tree $T_v$ using the same edge $zv$.
		
		
		\subsubsection{Adding constraints and variables to the ILP}\label{sec:constraints}
		We start by analyzing the case where we want to determine the labels on fastest temporal paths between vertices of interest.
		We proceed in the following way.
		Let $u,v \in U$ be two vertices of interest and let $P_{u,v}$ be the fastest temporal path from $u$ to $v$.
		If $P_{u,v}$ is a segment we determine all the labels of edges of $P_{u,v}$, with respect to the label of the first edge (see \cref{lemma:FPT-uv-Labelalledges}).
		In the case when $P_{u,v}$ is a sequence of $\ell$ segments, we determine all but $\ell - 1$ labels of edges of $P_{u,v}$, with respect to the label of the first edge (see \cref{lemma:FPT-uv-LabelAlmostalledges}).
		We call these $\ell - 1$ edges, \emph{partially determined} edges.
		After repeating this step for all pairs of vertices in $U$,
		the edges of fastest temporal paths from $u$ to $v$, where $u,v \in U$, are determined with respect to the label of the first edge of each path,
		or are partially determined.
		If the fastest temporal path between two vertices $u,v \in U$ is just an edge $e$, then we treat it as being determined, since it gets assigned a label $\lambda(e)$ with respect to itself.
		All other edges in $G'$ 
		are called the \emph{not yet determined} edges.
		Note that the not yet determined edges are exactly the ones that are not a part of any fastest temporal path between any two vertices in $U$.
		
		Now we want to relate the not yet determined segments with the determined ones.
		Let $S_{u,v}$ and $S_{w,z}$ be two segments.
		At the beginning we have guessed the fastest path from $v_i$ to all vertices in $S_{w,z}$ (see guess G-\ref{FPT:guess-splitFromAnotherSegmentAndPaths}).
		We did this by determining which vertices $z_j, z_{j+1}$ in $S_{w,z}$ are furthest away from $v_i$
		(remember we can have the case when $z_j = z_{j+1}$),
		and then we guessed how the path from $v_i$ leaves the segment $S_{u,v}$ (\ie either through the vertex $u$ or $v$),
		and then how it reaches $z_j$ (in the case when $z_j \neq z_{j+1}$ there is a unique way,
		when $z_j = z_{j+1}$ we determined which of the vertices $w$ or $z$ is on the fastest path).
		W.l.o.g.\ assume that we have guessed that the fastest path from $v_i$ to $z_j$
		passes through $w$ and $z_{j-1}$.
		Then the fastest temporal path from $v_i$ to $z_{j+1}$ passes through $z$.
		And all fastest temporal paths from $v_i$ to any $z_{j'} \in S_{w,z}$
		use all of the edges in $S_{w,z}$ with the exception of the edge $z_j z_{j+1}$.
		Using this information and \cref{obs:FirstLabelAndDuration}, we can determine the labels on all edges, with respect to the first or last label from the segment $S_{u,v}$,
		with the exception of the edge $z_j z_{j+1}$.
		Therefore, all edges of $S_{w,z}$ but $z_j z_{j+1}$ become determined.
		Since we repeat that procedure for all pairs of segments,
		we get that for a fixed segment $S_{w,z}$ we end up with a not yet determined edge $z_j z_{j+1}$
		if and only if this is a not yet determined edge in relation to every other segment $S_{u,v}$ and its fixed vertex $v_i$.
		We repeat this procedure for all pairs of segments.
		Each specific calculation takes linear time. Since there are $O(k^2)$ segments, the whole calculation takes $O(k^4)$ time.
		
		From the above procedure (where we were determining labels of edges of segments with each other) 
		we conclude that all of the edges $e_i=v_i v_{i+1}$ of a segment $S_{u,v}=(u=v_1, v_2, \dots, v_p=v)$ are in one of the following relations.
		First, where all of the edges are determined with respect to each other.
		Second, where
		there are some edges $e_1, e_2, \dots e_{i-1}$, whose label is determined with respect to the label $\lambda(e_1)$, there is an edge $f = e_i = v_i v_{i+1}$ which is not yet determined,
		and then there follow the edges
		$e_{i+1}, e_{i+2}, \dots, e_{p-1}$, whose labels are determined with respect to $\lambda(e_{p-1})$.
		Third, where the first $e_1, \dots, e_{i-1}$ edges are determined with respect to the $\lambda(e_1)$ and all of the remaining edges $e_i, e_{i+1} \dots, e_{p-1}$ are determined with respect to the $\lambda(e_{p-1})$.
		We want to now determine all of the edges in such segment $S_{u,v}$ with respect to just one edge (either the first or the last one).
		In the second case, we use the fact that at least one of the temporal paths between $v_{i-1}$ and $v_{i+1}$ has to pass through $f$, to determine $\lambda(f)$ with respect to $\lambda(e_{i-1})$ (and consequently $\lambda(e_1)$),
		and similarly, one of the temporal paths between $v_{i}$ and $v_{i+2}$ has to pass through $f$, which determines $\lambda(f)$ with respect to $\lambda(e_{i+1})$ (and consequently $\lambda(e_{p-1})$).
		In the third case, knowing the temporal paths between $v_{i-1}$ to $v_{i+1}$ results in determining the label of $\lambda(e_{i-1}$ with $\lambda(e_i)$, which consequently relates labels of all of the edges of the segment against each other.
		To determine the desired paths we proceed as follows.
		\begin{enumerate}[G-1.]
			\setcounter{enumi}{\value{guesscounter}}
			\item \label{FPT-guessFTPinSegmentTgroughEdge}
			Let $S_{u,v} = (u=v_1, v_2, \dots, v_p=v)$  be a segment of length at least $4$.
			If there is a not yet determined edge $v_i v_{i+1} = f$ in $S_{u,v}$ then
			we guess which of the fastest temporal paths: from $v_{i-1}$ to $v_{i+1}$, from $v_{i+1}$ to $v_{i-1}$, from $v_{i}$ to $v_{i+2}$, from $v_{i+2}$ to $v_{i}$ pass through the edge $f$. 
			If there are two incident edges $e = v_{i-1} v_{i}$ and $f = v_i v_{i+1} $ in $S_{u,v}$, that are determined with respect to $\lambda(v_1 v_2)$ and $\lambda (v_{p-1} v_p)$, respectively then we guess which of the fastest temporal paths: from $v_{i-1}$ to $v_{i+1}$, from $v_{i+1}$ to $v_{i-1}$ pass through the edges $e,f$. \\
			We create $O(1)$ guesses for every such segment $S_{u,v}$, and $O(k^2)$ new guesses in total, as there are at most $O(k^2)$ segments.
			\setcounter{guesscounter}{\value{enumi}}
		\end{enumerate}
		Note that the condition for segment length at least four comes from~\cref{lemma:FPT-noUndeterminedEdgeInSegment}.
		We now conclude the following.
		\begin{corollary}\label{cor:FPT-4determined}
			Let $S_{u,v}$ be an arbitrary segment in $G'$.
			If $S_{u,v}$ is of length $3$ or $2$ then it has at most $3$ or $2$ not yet determined edges, respectively.
			If $S_{u,v}$ is of length at least $4$ then the labels of all its edges are determined with respect to the first edge. 
		\end{corollary}
		
		At this point $G'$ is a graph, where each edge $e$ has a value for its label $\lambda(e)$
		that depends on (\ie is a function of) some other label $\lambda(f)$ of edge $f$,
		or it depends on no other label.
		We now describe how we create variables and start building our ILP instances. 
		For every edge $e$ in $G'$ that is incident to a vertex of interest, we create a variable $x_e$ that can have values from $ \{1, 2, \dots, \Delta\}$.
		Besides that, we create one variable for each edge that is still not yet determined on a segment.
		Since each vertex of interest is incident to at most $k$ edges in $G'$, and each segment has at most one extra not yet determined edge, we create $O(k^2)$ variables.
		At the end we create our final guess.
		\begin{enumerate}[G-1.]
			\setcounter{enumi}{\value{guesscounter}}
			\item \label{FPT:guessallPermutations}
			We guess the permutation of all $O(k^2)$ variables.
			So, for any two variables $x_e$ and $x_f$, we know if $x_e < x_f$ or $x_e = x_f$, or $x_e > x_f$.
			This results in $O(k^2)! = k^{O(k^2)}$ guesses
			and consequently
			each of the ILP instances we created up to now is further split into $k^{O(k^2)}$ new ones.
		\end{enumerate}
		
		We have now finished creating all ILP instances.
		From \cref{sec:FPT-guessing} we know the structure of all guessed paths, to which we have just added also the knowledge of permutation of all variables.
		We proceed with adding constraints to each of our ILP instances.
		First we add all constraints for the labels of edges that we have determined up to now.
		We then continue to iterate through all pairs of vertices and start adding equality (resp.~inequality) constraints for the fastest (resp.~not necessarily fastest) temporal paths between them.
		
		We now describe how we add constraints to a path. Whenever we say that a duration of a path gives an equality or inequality constraint, we mean the following.
		Let $P=(u=v_1,v_2, \dots, v_p = v)$ be the underlying path of a fastest temporal path from $u$ to $v$,
		and let $Q = (u=z_1,z_2, \dots, z_r = v)$ be the underlying path of another temporal path from $u$ to $v$.
		Then we know that $d(P,\lambda) = D_{u,v}$ and $d(Q, \lambda) \geq D_{u,v}$.
		Using \cref{obs:durationPwithWaitingTimes}
		we create an \emph{equality constraint} for $P$
		of the form 
		\begin{equation}\label{eq:FPT-equalityConstraint}
			D_{u,v} = \sum _ {i=2} ^ {p-1} (\lambda (v_{i}v_{i+1}) - \lambda (v_{i-1}v_i))_\Delta + 1,
		\end{equation}
		and an \emph{inequality constraint} for $Q$ 
		\begin{equation}\label{eq:FPT-inequalityConstraint}
			D_{u,v} \leq \sum _ {i=2} ^ {r-1} (\lambda (z_{i}z_{i+1}) - \lambda (z_{i-1}z_i))_\Delta + 1.
		\end{equation}
		In both cases we implicitly assume that if the difference of $(\lambda (z_{i}z_{i+1}) - \lambda (z_{i-1}z_i))$ is negative, for some $i$, we add the value $\Delta$ to it (\ie we consider the difference modulo $\Delta$), therefore we have the sign $\Delta$ around the brackets.
		Note that we can determine if the difference of two consecutive labels is positive or negative. 
		In the case when two consecutive labels are determined with respect to the same label $\lambda(e)$ the difference between them is easy to determine.
		If consecutive labels are not determined with respect to the same label, both labels are considered undetermined and are assigned a variable for which we know in what kind of relation they are (see guess \textcolor{lipicsGray}{\textsf{\textbf{G-\ref{FPT:guessallPermutations}}}}).
		Therefore, we know when $\Delta$ has to be added, which implies that \cref{eq:FPT-equalityConstraint,eq:FPT-inequalityConstraint} are calculated correctly for all paths.
		
		We iterate through all pairs of vertices $x,y$ and make sure that the fastest temporal path from $x$ to $y$ produces the equality constrain \cref{eq:FPT-equalityConstraint},
		and all other temporal paths from $x$ to $y$ produce the inequality constraint \cref{eq:FPT-inequalityConstraint}.
		For each pair we argue how we determine these paths.
		
		\subparagraph{\boldmath Fastest paths between $u,v \in U$.}
		Let $u,v\in U$, \ie both $u,v$ are vertices of interest.
		For the path from $u$ to $v$ (resp.~from $v$ to $u$) in $G'$, which we guessed that it coincides with the fastest in \textcolor{lipicsGray}{\textsf{\textbf{G-\ref{FPT-guessFTPamongU}}}}, we introduce an equality constraint. 
		We then iterate over all other paths from $u$ to $v$ (resp.~from $v$ to $u$) in $G'$, and for each one we introduce an inequality constraint.
		There are $k^{O(k)}$ possible paths from $u$ to $v$ in $G'$. Therefore we add $k^{O(k)}$ inequality constraints for the pair $u,v$.
		
		\subparagraph{\boldmath Fastest paths from $u \in U$ to $x \in V(G') \setminus U$.}
		From the guesses \textcolor{lipicsGray}{\textsf{\textbf{G-\ref{FPT:guess-uToSegmentz2}}}} and \textcolor{lipicsGray}{\textsf{\textbf{G-\ref{FPT:guess-splitFromUtoAnotherSegment}}}} we know the fastest temporal paths from $u$ to all vertices in a segment $S_{w,v}$.
		In this case we create an equality constraint for the fastest path and we 
		iterate through all other paths, for which we introduce the inequality constraints.
		There are $k^{O(k)}$ possible paths of the form $u \leadsto w$ (resp.~$u \leadsto v$),
		and a unique way how to extend these paths from $w$ (resp.~$v$) to reach $x$ in $S_{w,v}$.
		Therefore we add $k^{O(k)}$ inequality constraints for the pair $u,x$.

		\subparagraph{\boldmath Fastest paths from $x \in V(G') \setminus U$ to $u \in U$.}
		Let $x$ be a vertex in the segment $S_{w,z} = (w=z_1,z_2, \dots, z_r = z)$, and let $u\in U$.
		If $S_{w,z}$ is of length $3$ or less, then we already know the fastest temporal path from every vertex in the segment to $u$ 
		(since $S_{w,z}$ has at most $2$ inner vertices, we determined the fastest temporal paths from them to $u$ in guess \textcolor{lipicsGray}{\textsf{\textbf{G-\ref{FPT-guessFTPamongUandUstar}}}}).
		
		Assume that $S_{w,z}$ is of length at least $4$. From \cref{cor:FPT-4determined} we know that the labels of all the edges in $S_{w,z}$ are determined with respect to the label of the first edge.
		Moreover, this gives us the knowledge of the exact differences among two consecutive edge labels, which is enough to uniquely determine travel delays at all of the inner vertices $z_i \in S_{w,z}$ (see~\cref{def:travel-delay}).
		
		From the matrix $D$ we can easily determine the two vertices $z_i, z_{i+1} \in S_{w,z}\setminus\{w,z\}$ for which the fastest temporal path from $z_i$ to $u$ has the biggest duration.
		Let us denote with $P^+$ the fastest temporal path of the form $z_2 \rightarrow z \leadsto u$,
		and with $P^-$ the fastest temporal path of the form $z_{r-1} \rightarrow w \leadsto u$ 
		(we know these paths from guess~\textcolor{lipicsGray}{\textsf{\textbf{G-\ref{FPT:guess-uToSegmentz2}}}}).
		It follows that all vertices $z_j$ in $S_{w,z} \setminus\{z_i,z_{i+1}\}$ that are closer to $w$ than $z_i,z_{i+1}$ reach $u$ the fastest using the path $(z_j \rightarrow z_{j-1} \rightarrow \cdots \rightarrow z_2) \cup P^+$ 
		and
		all the vertices $z_j$ in $S_{w,z}\setminus\{z_i,z_{i+1}\}$ that are closer to $z$ than $z_i, z_{i+1}$ reach $u$ the fastest using the path
		$(z_j \rightarrow z_{j+1} \rightarrow \cdots \rightarrow z_{r-1}) \cup P^-$.
		Since the first part of the above path is unique, and we know that the second part is the fastest, it follows that these paths indeed represent the fastest temporal paths to $u$.
		What remains to determine is the fastest temporal paths from $z_i,z_{i+1}$ to $u$.
		We distinguish the following two options.
		\begin{enumerate}[(i)]
			\item \label{FPT:equality-fromUtoX-twosplit}
			$z_i \neq z_{i+1}$.
			Then the fastest temporal path from $z_i$ to $u$ is 
			$(z_i \rightarrow z_{i-1} \rightarrow \cdots \rightarrow z_2) \cup P^+$,
			and 
			the fastest temporal path from $z_{i+1}$ to $u$ is 
			$(z_{i+1} \rightarrow z_{i+2} \rightarrow \cdots \rightarrow z_{r-1}) \cup P^-$.
			\item \label{FPT:equality-fromUtoX-onesplit}
			$z_i = z_{i+1}$, \ie let $z_i$ be the unique vertex, that is furthest away from $u$ in $S_{w,z}$.
			In this case we have to determine if the fastest temporal path from $z_i$ to $u$, travels first through vertex $z_{i-1}$ (and then through $w$),
			or it travels first through $z_{i+1}$ (and then through $z$).
			Since we know the values $D_{z_{i-1},u}, D_{z_{i+1},u}$,
			and we know the value of the waiting time $\tau_{v_{i-1}} ^{v_{i}, v_{i-2}}$ at vertex $v_{i-1}$ when traveling from $v_i$ to $v_{i-2}$, 
			we can uniquely determine the desired path.
			We set $c = D_{z_{i-1},u} +  \tau_{v_{i-1}} ^{v_{i}, v_{i-2}}$
			and compare $c$ to the value $D_{z_{i},u}$.
			If $c < D_{z_{i},u}$ we conclude that our ILP has no solution and we stop with calculations,
			if $c = D_{z_{i},u}$ then the fastest temporal path from $z_i$ to $u$ is of the form $(z_i \rightarrow z_{i-1} \rightarrow \cdots \rightarrow z_2) \cup P^+$,
			if $c > D_{z_{i},u}$ then the fastest temporal path from $z_i$ to $u$ is of the form $(z_i \rightarrow z_{i+1} \rightarrow \cdots \rightarrow z_{r-1}) \cup P^-$.
		\end{enumerate}
		Once the fastest temporal path from $c$ to $u$ is determined, we introduce an equality constraint for it.
		For each of the other $k^{O(k)}$ paths from $x$ to $u$ (which correspond to all paths of the form $w \leadsto u$ and $z \leadsto u$, together with the unique subpath on $S_{w,z}$),
		we introduce an inequality constraint. Therefore we add $k^{O(k)}$ inequality constraints for the pair $x,u$.
		
		\subparagraph{\boldmath Fastest paths between $x,y\in V(G') \setminus U$.}
		Let $x,y \in V(G') \setminus U$.
		There are two options.
		\begin{enumerate}[(i)]
			\item Vertices $x,y$ are in the same segment $S_{u,v} = (u,v_1,v_2, \dots, v_p, v)$. 
			If the length of $S_{u,v}$ is less than $4$ then we know what is the fastest path between vertices, as $x,y \in U^*$.
			Suppose now that $S_{u,v}$ is of length at least $4$ and assume that $x$ is closer to $u$ in $S_{u,v}$ than $y$.
			Then we have two options; either the path from $x$ to $y$ travels only through the edges of $S_{u,v}$, denote such path as $P_{x,y}$,
			or it is of the form $x \rightarrow v_1 \rightarrow u \leadsto v \rightarrow v_p \rightarrow y$, denote is as $P_{x,y}^*$.
			Note that we can determined $P_{x,y}^*$ as it is a concatenation of a unique path from $x$ to $v_2$, together with the fastest path from $v_2$ to $v_{p}$, that travels through $u$ and $v$ (we know this path from \textcolor{lipicsGray}{\textsf{\textbf{G-\ref{FPT-guessFTPamongv2z2}}}}),
			and the unique path from $v_p$ to $y$.
			Because of \cref{cor:FPT-4determined} we can determine $c = d(P_{x,y}, \lambda)$.
			If $c > D_{x,y}$ we set the fastest path to be $P_{x,y}^*$,
			if $c = D_{x,y}$ then the fastest path is $P_{x,y}$, and if 
			$c < D_{x,y}$ we conclude that our ILP has no solution and we stop with calculations.
			\item Vertices $x$ and $y$ are in different segments. 
			Let $x$ be a vertex in the segment $S_{u,v} = (u=v_1,v_2, \dots, v_p = v)$ and let $y$ be a vertex in the segment $S_{w,z} = (w=z_1, z_2,z_3, \dots, z_r = z)$.
			By checking the durations of the fastest paths from $x$ to every vertex in $S_{w,z} \setminus \{w,z\}$
			we can determine the vertex $z_i \in S_{w,z} $, for which the duration from $x$ is the biggest.
			Note that if there are two such vertices $z_i$ and $z_{i+1}$, we know exactly how all fastest temporal paths enter $S_{w,z}$ 
			(we use similar arguing as in case~(\ref{FPT:equality-fromUtoX-twosplit}) from above, where we were determining the fastest path from $x \in V(G')$ to $u \in U$).
			This implies that the fastest temporal paths from $x$ to all vertices $z_2, z_3, \dots, z_{i-1}$ (resp.~$z_{i+1}, z_{i+2}, \dots, z_{r-1}$)  pass through $w$ (resp.~$z$).
			Now we determine the vertex $v_j \in S_{u,v} \setminus \{u,v\}$,
			for which the value of the durations of the fastest paths from it to the vertex $y$ is the biggest.
			Again, if there are two such vertices $v_{j}$ and $v_{j+1}$ we know exactly how the fastest temporal paths, starting in these two vertices,
			leave the segment $S_{u,v}$. 
			We use similar arguing as in case~(\ref{FPT:equality-fromUtoX-twosplit}) from above, when we were determining the fastest path from $x \in V(G')$ to $u \in U$.
			Knowing the vertex $v_j$
			implies that the fastest temporal paths from the vertices $v_2, v_2, \dots, v_{j-1}$ (resp.~$v_{j+1}, v_{j+2}, \dots, v_{p-1}$) to the vertex $y$ passes through $u$ (resp.~$v$).
			Since we know the following fastest temporal paths (see guess~\textcolor{lipicsGray}{\textsf{\textbf{G-\ref{FPT-guessFTPamongv2z2}}}}) 
			$z_2 \rightarrow w \leadsto u \rightarrow v_2$,
			$z_2 \rightarrow w \leadsto v \rightarrow v_{p-1}$,
			$z_{r-1} \rightarrow z \leadsto v \rightarrow v_{p-1}$ and
			$z_{r-1} \rightarrow z \leadsto v \rightarrow v_{p-1}$,
			we can uniquely determine all fastest temporal paths from $x \neq v_j$ to any $y \in S_{u,v} \setminus \{z_i\}$.
			
			In case when $x = v_j$ and $y = z_i$ and the segments are of length at least $4$, we can uniquely determine the fastest path from~$v_j$ to $z_i$,
			using 
			similar arguing as in case~(\ref{FPT:equality-fromUtoX-onesplit}) from above, when we were determining the fastest path from $x \in V(G')$ to $u \in U$.
			If at least one of the segments is of length $3$ or less,
			we can again uniquely determine the fastest path from~$v_j$ to~$z_i$, using the same approach, 
			and the knowledge of fastest paths to (or from) all vertices of the segment of length~$3$
			(as we guessed them in guess~\textcolor{lipicsGray}{\textsf{\textbf{G-\ref{FPT-guessFTPamongv2z2}}}}).
		\end{enumerate}
		Once the fastest path is determined we introduce the equality constraint for it
		and iterate through all other paths, for which we introduce inequality constraints.
		To enumerate all these non-fastest temporal paths, we just consider all possible paths $u \leadsto w$, where 
		$u$ and $w$ are the vertices of interest that are the endpoints of segments to which $x$ and $y$ belong;
		once the correct segment is reached, there is a unique path to the desired vertex $x$ (resp.~$y$). 
		Therefore we introduce $k^{O(k)}$ inequality constraints for each pair of vertices $x,y$.
		
		\subparagraph{\boldmath Fastest paths for vertices in $Z$.}
		All of the above is enough to determine the labeling $\lambda$ of $G'$. We have to extend the labeling to consider also the vertices of~$Z= V(G) \setminus V(G')$ that we initially removed from $G$.
		
		Recall that $G[Z]$ consists of disjoint trees and that each of these trees has a unique neighbor (clip vertex) $v$ in $G'$. We then define the tree $T_v$ in $G[Z \cup \{v\}]$ as the collection of trees from $G[Z]$ with a clip vertex $v$, together with the root $v$. 
		Determining the fastest temporal paths between any two vertices in the same tree is a straightforward process (see~\cref{thm:deltaExact-PolyTimeTrees}), therefore we exclude this case from our upcoming analysis.
		From \cref{obs:tree-only-neighbors} it follows that knowing temporal paths between any $y  \in V(G')$ and all vertices in the first layer of the tree $T_v$ (\ie children of the root $v$),
		it is enough to determine the fastest temporal paths between $y$ and all other vertices in $T_v$.
		Therefore,
		in the upcoming analysis, we focus only on the vertices in the first layer of each tree $T_v$.
		During the process of determining the fastest paths from and to the vertices in $Z$, we use the fact that we have already identified the fastest paths among all vertices in $G'$.
		
		We split our analysis into two cases.
		First, when the clip vertex $v$ of tree $T_v$ is not a vertex of interest,
		and second when the clip vertex is also a vertex of interest in $G'$.
		In the first case we use the fact that $v$ has only two edges $e,f$ incident to it in $G'$, and that we can determine all the labels of the tree edges with respect to $\lambda(e),\lambda(f)$ (see~\cref{lemma:treeLabels-notInteresting}).
		This results to be enough for us to determine the fastest temporal paths among any vertex $r$ in the first layer of the tree $T_v$ and an arbitrary vertex in $V(G)\setminus V(T_v)$.
		In the second case we cannot determine the labeling of the tree with respect to the labels of all edges incident to the clip vertex.
		Therefore, we split the vertices in the first layer of $T_v$ into equivalence classes, and use the fact that the fastest temporal paths between two vertices in the same equivalence class coincide on the edges outside of $T_v$.
		
		\subparagraph{\boldmath Fastest paths from $r \in Z$ to $y \in U \cup U^*$.}
		Let $x \in V(G')$ be a clip vertex of the tree $T_x$ in $G[Z \cup \{x\}]$ with $r \in V(T_i)$ be a vertex in the first layer of $T_x$.
		We distinguish the following two cases.
		\begin{enumerate}[(i)]
			\item \label{FPT:clipInU-1}
			The clip vertex $x =u \in U$ is a vertex of interest.
			In this case we can w.l.o.g. assume that $r$ is a representative vertex in its equivalence class among the first layer vertices of $T_u$.
			From the guesses~\textcolor{lipicsGray}{\textsf{\textbf{G-\ref{FPT-guessFTPamongUandZstar}}}} and \textcolor{lipicsGray}{\textsf{\textbf{G-\ref{FPT-guessFTPamongUstarandZstar}}}}
			we know the fastest temporal path from $r$ to $y$.
			\item \label{FPT:clipNotInU-1}
			The clip vertex $x \in U$ is not a vertex of interest. Then $x = v_j$ is a part of some segment $S_{u,v} = \{u=v_1, v_2, \dots, v_p = v)$, where $j \neq 1 \neq p$.
			Using \cref{lemma:treeLabels-notInteresting} we can determine all the edge labels of $T_x$ with respect to the label $\lambda(v_{j-1}v_j)$ and with respect to the label $\lambda(v_{j}v_{j+1})$.
			Using the calculations of fastest temporal paths among vertices in $G'$ and the performed guesses we know the exact structure (\ie the sequence of vertices and edges) of the following paths:
			\begin{itemize}
				\item path $P_{xy}^*$ which is the fastest temporal path from the vertex $x$ to the vertex $y$,
				\item path $P_{xy}^u$ which is the fastest temporal path from the vertex $x$ to the vertex $y$, that passes through the vertex $u$,
				\item path $P_{xy}^v$ which is the fastest temporal path from the vertex $x$ to the vertex $y$, that passes through the vertex $v$.
			\end{itemize}
			Note that $P_{xy}^*$ is either equal to the path $P_{xy}^u$ or to the path $P_{xy}^v$.
			More precisely, from the guesses performed we know the structure of the fastest path from $v_2$ through $u$, that then continues to any other vertex of interest, and any other neighbor of the vertex of interest (see~guesses \textcolor{lipicsGray}{\textsf{\textbf{G-\ref{FPT-guessFTPamongv2z2}}}} and
			\textcolor{lipicsGray}{\textsf{\textbf{G-\ref{FPT:guess-uToSegmentz2}}}}).
			This path can then be easily (uniquely) extended to start from $x=v_i$, as there is a unique (temporal) path starting at $x$ and finishing in $u$ or $v$. 
			
			Suppose now that $P_{xy}^* = P_{xy}^u$.
			Since the labels of $T_x$ are determined with respect to $\lambda(v_{i-1} x)$ we can calculate the value $c = D_{x,y} + |\lambda(v_{i-1} x) - \lambda (r x)|$.
			We then compare $c$ to $D_{r,y}$ and get one of the following three options.
			First $c = D_{r,y}$, in this case the fastest temporal path from $r$ to $y$ uses first the edge $rx$ and then continues to $y$ using the edges and vertices of $P_{xy}^*$.
			Second $c > D_{r,y}$, in this case the fastest temporal path from $r$ to $y$ uses first the edge $rx$ and then continues to $y$ using the vertices and edges of $P_{xy}^v$.
			Third $c < D_{r,y}$, in this case we stop the calculation and return false, as it cannot happen that a temporal path has a smaller duration than the corresponding value in the matrix $D$.
		\end{enumerate}
		In both cases, we introduce an equality constraint for the determined fastest temporal path and inequality constraints for all the other $k^{O(k)}$ paths.
		
		\subparagraph{\boldmath Fastest paths from $r \in Z$ to $y \in V(G) \setminus (U \cup U^*)$.}
		The proof in this case is similar to the one above.  We still split the analysis into two parts, one where the clip vertex $x$ of a tree $T_x$ that includes $r$ is in $U$ and one where it is not in $U$.
		The difference is that in some cases we need to also extend the ending part of the path (which can be done uniquely, using the same arguments as in the above analysis).
		
		Once we determine the fastest temporal path from $r$ to $y$ we introduce an equality constraint for it, and for all other $k^{O(k)}$ paths we introduce inequality constraints.
		
		The procedure produces one equality constraint (for the fastest path) and $k^{O(k)}$ inequality constraints.
		
		\subparagraph{\boldmath Fastest paths from $y \in V(G)$ to $r \in Z$.}
		The process of determining fastest temporal paths from any vertex in the graph $G$ to a vertex $r$ that is a vertex in the first layer of a tree $T_x \in G[Z \cup \{x\}]$, where $x \in V(G')$,
		is similar to the one above, but performed in the opposite direction.
		
		\subsubsection{Solving ILP instances}
		All of the above finishes our construction of ILP instances.
		We have created $f(k)$ instances (where $f$ is a double exponential function), 
		each with $O(k^2)$ variables and $O(n^2) g(k)$ constraints (again, $g$ is a double exponential function).
		We now solve each ILP instance $I$, using results from Lenstra \cite{Lenstra1983Integer},
		in the FPT time, with respect to $k$.
		If none of the ILP instances gives a positive solution, then there exists no labeling $\lambda$ of $G$ that would realize the matrix $D$ (\ie for any pair of vertices $u,v \in V(G)$ the duration of a fastest temporal path from $u$ to $v$ has to be $D_{u,v}$).
		If there is at least one $I$ that has a valid solution, 
		we use this solution and produce our labeling $\lambda$, for which $(G,\lambda)$ realizes the matrix $D$.
		We have proven in the previous subsections that this is true since each ILP instance corresponds to a specific configuration of fastest temporal paths in the graph (\ie considering all ILP instances is equivalent to exhaustively searching through all possible temporal paths between vertices).
		Besides that, in each ILP instance we add also the constraints for durations of all temporal paths between each pair of vertices.
		This results in setting the duration of a fastest path from a vertex $u \in V(G)$ to a vertex $v \in V(G)$ as $D_{u,v}$,
		and the duration of all other temporal paths from $u$ to $v$, to be greater or equal to $D_{u,v}$,
		for all pairs of vertices $u,v$.
		Therefore, if there is an instance with a positive solution, then this instance gives rise to the desired labeling, as it satisfies all of the constraints.
		For the other direction, we can observe that if there is a labeling $\lambda$ meeting all duration requirements specified by $D$, then this labeling produces a specific configuration of fastest temporal paths. Since we consider all configurations, one of the produced ILP instances will correspond to the configuration implicitly defined by $\lambda$, and hence our algorithm finds a solution.
		
		To create the labeling $\lambda$ from a solution $X$, of a positive ILP instance,
		we use the following procedure.
		First we
		label each edge $e$, that corresponds to the variable $x_e$
		by assigning the value $\lambda(e) = x_e$.
		We then continue to set the labels of all other edges. 
		We know that the labels of all of the remaining edges depend on the label of (at least one) of the edges that were determined in previous step. 
		Therefore, we easily calculate the desired labels for all remaining edges.

		\section{Conclusion}\label{sec:conclusion}
		We have introduced a natural and canonical temporal version of the graph realization problem with respect to distance requirements, called \deltaExactLong. 
		We have shown that the problem is NP-hard in general and polynomial-time solvable if the underlying graph is a tree.
		Building upon those results, we have investigated its parameterized computational complexity with respect to structural parameters of the underlying graph that measure ``tree-likeness''. For those parameters, we essentially gave a tight classification between parameters that allow for tractability (in the FPT sense) and parameters that presumably do not.
		We showed that our problem is W[1]-hard when parameterized by the feedback vertex number of the underlying graph, and that it is in FPT when parameterized by the feedback edge number of the underlying graph. Note that most other common parameters that measure tree-likeness (such as the treewidth) are smaller than the vertex cover number.
		
		We believe that our work spawns several interesting future research directions and builds a base upon which further temporal graph realization problems can be investigated.
		
		\subparagraph{Further parameterizations.} There are several structural parameters which can be considered to obtain tractability which are either larger than or incomparable to the feedback vertex number.
		\begin{itemize}
			\item The \emph{vertex cover number} measures the distance to an independent set, on which we trivially only have no-instances of our problem. We believe this is a promising parameter to obtain tractability.
			\item The \emph{tree-depth} measures ``star-likeness'' of a graph and is incomparable to both the feedback vertex number and the feedback edge number. We leave the parameterized complexity of our problem with respect to this parameter open.
			\item Parameters that measure ``path-likeness'' such as the \emph{pathwidth} or the \emph{vertex deletion distance to disjoint paths} are also natural candidates to investigate.
		\end{itemize}
		Furthermore, we can consider combining a structural parameter with $\Delta$. Our NP-hardness reduction (\cref{thm:NPhardness}) produces instances with constant $\Delta$, so as a single parameter $\Delta$ cannot yield fixed-parameter tractability. However, in our parameterized hardness reduction (\cref{thm:W1wrtFVS}) the value for $\Delta$ in the produced instance is large. This implies that our result does not rule out e.g.\ fixed-parameter tractability for the combination of the treewidth and $\Delta$ as a parameter. We believe that investigating such parameter combinations is a promising future research direction.
		
		\subparagraph{Further problem variants.}
		There are many natural variants of our problem that are well-motivated and warrant consideration. In the following, we give two specific examples. We believe that one of the most natural generalizations of our problem is to allow more than one label per edge in every $\Delta$-period. A well-motivated variant (especially from the network design perspective) of our problem would be to consider the entries of the duration matrix $D$ as upper-bounds on the duration of fastest paths rather than exact durations. This problem variant has very recently been studied by Mertzios et al.~\cite{MMS24}.  

		

	\end{document}